\documentclass[a4paper,reqno]{amsart} 
\usepackage{amsfonts}
\usepackage{amssymb}
\usepackage{amsthm}
\usepackage{amsmath}
\usepackage{mathrsfs}
\usepackage{dsfont}
\usepackage{stmaryrd}
\usepackage[english]{babel}
\usepackage[left=3.5cm,right=3.5cm,top=3.5cm,bottom=3cm,headsep=0.7cm]{geometry}

\usepackage{pgf,tikz}
\usetikzlibrary{arrows}
\usetikzlibrary{intersections}
\usepackage[scriptsize,bf]{caption}
\usepackage{floatrow}
\usepackage{enumerate}
\usepackage{enumitem}
\usepackage{hyperref}
\usepackage[normalem]{ulem}

\theoremstyle{plain}
\begingroup
\theoremstyle{plain}
\newtheorem{theorem}{Theorem}[section]

\newtheorem{proposition}[theorem]{Proposition}
\newtheorem{lemma}[theorem]{Lemma}
\theoremstyle{definition}
\newtheorem{definition}[theorem]{Definition}
\theoremstyle{remark}
\newtheorem{remark}[theorem]{Remark}

\endgroup

\theoremstyle{definition}
\theoremstyle{remark}

\numberwithin{equation}{section}

\newcommand{\RR}{\mathbb{R}}
\newcommand{\NN}{\mathbb{N}}
\newcommand{\ZZ}{\mathbb{Z}}
\newcommand{\CC}{\mathbb{C}}

\renewcommand{\SS}{\mathbb{S}}

\renewcommand{\S}{\mathcal{S}}

\renewcommand{\H}{\mathcal{H}}
\newcommand{\D}{\mathcal{D}}
\newcommand{\DD}{\mathrm{D}}

\renewcommand{\L}{\mathcal{L}}
\newcommand{\M}{\mathcal{M}}

\newcommand{\q}{\mathfrak{q}}

\mathchardef\emptyset="001F
\renewcommand{\d}[1]{\, \mathrm{d} #1}
\newcommand{\de}{\partial}
\newcommand{\e}{\varepsilon}
\renewcommand{\tilde}{\widetilde}
\newcommand{\x}{{\times}}
\newcommand{\ol}{\overline}
\newcommand{\ul}{\underline}
\newcommand{\sm}{\setminus}
\newcommand{\dist}{{\rm dist}}
\newcommand{\argmin}{\mathrm{argmin}}

\newcommand{\cart}{\mathrm{cart}}

\newcommand{\weak}{\rightharpoonup}
\newcommand{\wstar}{\stackrel{*}\rightharpoonup}
\renewcommand{\flat}{\stackrel{\mathrm{f}}\to}
\newcommand{\mres}{\mathbin{\vrule height 1.6ex depth 0pt width 0.13ex\vrule height 0.13ex depth 0pt width 1.3ex}}
\newcommand{\integral}[3]{\int_{#1} \! #2 #3}
\newcommand{\nn}{{\langle i , j \rangle}}
\newcommand{\PC}{\mathcal{PC}}
\newcommand{\geo}{\mathrm{d}_{\SS^1}}
\newcommand{\compact}{\subset\subset}
\renewcommand{\Supset}{\supset\supset}
\newcommand{\w}{{\! \wedge \!}}
\newcommand{\supp}{\mathrm{supp}}
\newcommand{\Geo}{\mathrm{Geo}}

\author{Marco Cicalese}
\address[Marco Cicalese]{TU Munich, Zentrum Mathematik, Boltzmannstr. 3, 85747 Garching bei München, Germany}
\email{cicalese@ma.tum.de}

\author{Gianluca Orlando}
\address[Gianluca Orlando]{Politecnico di Bari, Dipartimento di Meccanica, Matematica e Management, via E. Orabona 4
	70125 Bari BA, Italy}
\email{orlando@ma.tum.de}

\author{Matthias Ruf}
\address[Matthias Ruf]{Ecole polytechnique f\'ed\'erale de Lausanne, SB MATH, Station 8, 1015 Lausanne, Switzerland}
\email{matthias.ruf@epfl.ch}

\title[Variational analysis of the $N$-clock model]{Emergence of concentration effects \\ in the variational analysis of the $N$-clock model}

\begin{document}

\begin{abstract}
	We investigate the relationship between the $N$-clock model (also known as planar Potts model or $\ZZ_N$-model) and the $XY$ model (at zero temperature) through a $\Gamma$-convergence analysis of a suitable rescaling of the energy as both the number of particles and $N$ diverge. We prove the existence of rates of divergence of $N$ for which the continuum limits of the two models differ. With the aid of Cartesian currents we show that the asymptotics of the $N$-clock model in this regime features an energy which may concentrate on geometric objects of various dimensions. This energy prevails over the usual vortex-vortex interaction energy. 
\end{abstract}

\maketitle
 

\vspace{1em}


\setcounter{tocdepth}{1}
\tableofcontents

\vspace{-2em}

\section{Introduction}

Classical ferromagnetic spin systems on lattices represent fundamental models to understand phase transition phenomena. On the one hand, the study of their properties has motivated the introduction of new mathematical tools which have provided useful insights for a number of problems arising in different fields. On the other hand, many techniques borrowed from probability theory, mathematical analysis, topology, and geometry have contributed to a better understanding of the properties of these systems. 

In this paper we make use of fine concepts in geometric measure theory and in the theory of Cartesian currents to understand the relationship between the $XY$-model and the $N$-clock model (also known as planar Potts model or $\ZZ_N$-model) within a variational framework. 
The $N$-clock model is a two-dimensional nearest neighbors ferromagnetic spin model on the square lattice in which the spin field is constrained to take values in a set of $N$ equi-spaced points of $\SS^{1}$. For $N$ large enough, it is usually considered as an approximation of the $XY$ (planar rotator) model, for which the spin field can attain all the values of $\SS^{1}$. The asymptotic behavior of the $N$-clock model for large $N$ has been considered by Fr\"ohlich and Spencer in the seminal paper~\cite{Fro-Spe}. There the authors have proved that both the $N$-clock model (for $N$ large enough) and the $XY$ model present Berezinskii-Kosterlitz-Thouless transitions, i.e., phase transitions mediated by the formation and interaction of topological singularities. The microscopic picture leading to the emergence of such topological phase transitions (first introduced in ~\cite{Ber, Kos, Kos-Tho}) is a result of a nontrivial interplay between entropic and energetic effects that takes place at different length scales. 

This paper contributes to precisely relating the $N$-clock model and the $XY$ model at zero temperature. Specifically, we show that the enhancement of symmetry, from the discrete one of the $N$-clock model to the continuous one of the $XY$ model, comes along with concentration of energy on geometric objects of various dimension. This is achieved by studying a suitably rescaled version of the energy of the $N$-clock model as $N$ diverges, through a coarse graining procedure which is made rigorous by $\Gamma$-convergence, see~\cite{Bra,DM}. A crucial step of this analysis is the choice of the topologies which best identify the relevant variables of the coarse grained model and lead to the effective description of the microscopic/mesoscopic geometry of the spin field. In contrast to the $XY$ model, the sole study of the distributional Jacobian of the spin field turns out to provide not enough information on the concentration effects of the energy; we shall see how these effects can be detected by Cartesian currents, for the first time introduced in the context of lattice spin models. 

In what follows we present the model and our main result. We consider a bounded, open set with Lipschitz boundary $\Omega \subset \RR^2$. 
Given a small parameter $\e > 0$, we consider $\Omega_\e := \Omega \cap \e \ZZ^2$. The classical $XY$ model is defined on spin fields $u \colon \Omega_\e \to \SS^1$ by 
\begin{equation} \label{eq:XY}
- \sum_\nn \e^2 u(\e i) \cdot u(\e j)  \, ,
\end{equation}
where the sum is taken over ordered pairs of nearest neighbors $\nn$, i.e., $(i, j) \in \ZZ^2 \x \ZZ^2$ such that $|i - j| = 1$ and $\e i, \e j \in \Omega_\e$. The variational analysis of the $XY$ model is part of a larger program devoted to the study of systems of spins with continuous symmetry~\cite{Ali-Cic, Ali-Cic-Pon, Ali-DL-Gar-Pon, Bad-Cic-DL-Pon, Can-Seg, Cic-Sol, Cic-Sol-Ruf, Sci-Val,Cic-For-Orl}.

Here we consider an additional parameter $N_\e \in \NN$ or, equivalently, $\theta_\e := \tfrac{2\pi}{N_\e}$. The admissible spin fields we consider here are only those taking values in the discrete set $\S_\e := \{\exp(\iota k \theta_\e) \colon  k = 0, \dots, N_\e-1\} \subset \SS^1$, i.e., we consider the energy 
\begin{equation*}
	F_\e(u) :=  - \sum_\nn \e^2 u(\e i) \cdot u(\e j) \quad  \text{if } u \colon \Omega_\e \to \S_\e \, , 
\end{equation*}
extended to $+\infty$ otherwise. For $N_\e = N \in \NN$, with $N$ independent of $\e$, the spin system described by the energy~$F_\e$ is usually referred to as $N$-clock model, cf.\ \cite{Fro-Spe}. 
The particular case where $N=2$ is the so-called Ising system, recently analyzed in~\cite{Iof-Sch, Bod, Cer-Pis, Caf-DLL, Ali-Bra-Cic}. See also~\cite{Alb-Bel-Cas-Pre, Ali-Cic-Ruf, Ali-Gel, Bra-Cic} for the long-range case. 

As~\eqref{eq:XY} is minimized on constant spin fields, one refers the energy to its minimum 
\begin{equation*}
	XY_\e(u) = - \sum_\nn \e^2 \Big( u(\e i) \cdot u(\e j) - 1 \Big) = \frac{1}{2} \sum_\nn \e^2 | u(\e i) - u(\e j) |^2 .
\end{equation*}
Analogously, we set
\begin{equation} \label{eq:def of E}
	E_\e(u) := F_\e(u) - \min F_\e = XY_\e(u) \quad  \text{if } u \colon \Omega_\e \to \S_\e \, , 
\end{equation}
extended to $+\infty$ otherwise, and we find the scalings $\kappa_\e \to 0$ for which $\frac{1}{\kappa_\e}E_\e$ has a nontrivial variational limit. These are affected by $N_\e$, as it emerges in the two limiting scenarios $N_\e = 2$ and $\S_\e = \SS^1$ (formally corresponding to $N_\e = +\infty$). If $N_\e = 2$, $\frac{1}{\e} E_\e(u_\e)$ approximates an anisotropic interfacial energy between the phases $(1,0)$ and $(-1,0)$, see~\cite{Ali-Bra-Cic}. In contrast, for the $XY$ system, i.e., $\S_\e = \SS^1$, it has been shown in~\cite[Example~1]{Ali-Cic} that no interfacial-type energy emerges at any scaling $\kappa_\e \gg \e^2$. Indeed, if $u_\e$ interpolates (linearly in the angle) from $u^- = \exp(\iota \varphi^-)$ to $u^+ = \exp(\iota \varphi^+)$ on a length-scale of size $\eta_\e$, the energy amounts to
\begin{equation} \label{eq:jump pays nothing}
	\frac{1}{\kappa_\e} XY_\e(u_\e) \sim \Big( 1 - \cos\Big(\frac{\e}{\eta_\e} (\varphi^+ - \varphi^-)\Big) \Big) \frac{\eta_\e}{\kappa_\e} \sim \frac{\e^2}{\eta_\e \kappa_\e}\, ,
\end{equation}
which goes to 0 if $\eta_\e \gg \tfrac{\e^2}{\kappa_\e}$.
This construction may not be feasible when $\S_\e \neq \SS^1$ if the minimal angle $\theta_\e$ satisfies $\theta_\e \gtrsim \frac{\kappa_\e}{\e}$.
In the constrained case, choosing the largest possible length-scale $\eta_\e = |\varphi^+-\varphi^-|\frac{\e}{\theta_\e}$, one gets (denoting by $\geo$ the geodesic distance on $\SS^1$)
\begin{equation} \label{eqintro:angle interpolation}
	\frac{1}{\kappa_\e} E_\e(u_\e) \sim \big( 1 - \cos(\theta_\e ) \big) \frac{\e}{\theta_\e \kappa_\e} |\varphi^+ - \varphi^-| \sim \frac{\e \theta_\e}{\kappa_\e}  |\varphi^+ - \varphi^-| \sim \frac{\e \theta_\e}{\kappa_\e} \geo(u^+,u^-)\, ,
\end{equation}
which suggests that $\kappa_\e = \e \theta_\e$ leads to an energy proportional to a $BV$ total variation (in the sense of~\cite[Formula~(2.11)]{Amb}). In fact, in Proposition~\ref{prop:BV compactness}, we prove that sub-level sets of	$\frac{1}{\e \theta_\e} E_\e$ are pre-compact in $BV(\Omega;\SS^1)$ equipped with the $L^1$ topology.

Given $u_\e \colon \Omega_\e \to \S_\e$ with $\frac{1}{\e \theta_\e} E_\e(u_\e) \leq C$ we have
\begin{equation} \label{eq:no vortices allowed}
	\frac{1}{\e^2 |\log \e|} XY_\e(u_\e) =  \frac{\e \theta_\e}{\e^2 |\log \e|} \frac{1}{\e \theta_\e}  E_\e(u_\e) \sim \frac{\theta_\e}{\e |\log \e|}  \, .
\end{equation}
As it is known from the theory of the $XY$ model~\cite{Ali-Cic,Ali-Cic-Pon} (see also~\cite{Bet-Bre-Hel, San, Jer, Jer-Son, San-Ser, Alb-Bal-Orl, San-Ser-book, Ali-Pon} for the Ginzburg-Landau theory), boundedness of $\frac{1}{\e^2 |\log \e|} XY_\e(u_\e)$ implies flat compactness\footnote{i.e., with respect to the norm induced by duality with compactly supported Lipschitz functions.} of the discrete vorticity measure $\mu_{u_\e}$, which counts the winding number of $u_\e$ at each point of $\e \ZZ^2$, cf.~\eqref{eq:discrete vorticity measure}. If $\e |\log \e| \ll \theta_\e$, \eqref{eq:no vortices allowed} gives no bound on $\frac{1}{\e^2 |\log \e|} XY_\e(u_\e)$ and suggests that $\mu_{u_\e}$ does not play a role in the asymptotics of~$\frac{1}{\e \theta_\e}  E_\e$. In fact, in~\cite{Cic-Orl-Ruf} we prove that $\frac{1}{\e \theta_\e}  E_\e$ $\Gamma$-converges to an anisotropic total variation in $BV(\Omega;\SS^1)$. Here we are interested in regimes for which the limit cannot be exhaustively described in~$BV$. We start by assuming $\theta_\e \ll \e |\log \e|$, which by~\eqref{eq:no vortices allowed} implies $\mu_{u_\e} \flat 0$. This constraint will induce a $\Gamma$-limit (possibly strictly) larger than the anisotropic total variation in $BV(\Omega;\SS^1)$. 
To prove this fact, our idea is to associate to $u_\e$ with $\frac{1}{\e \theta_\e} E_\e(u_\e) \leq C$ the current $G_{u_\e}$ given by the extended graph in $\Omega \x \SS^1$ of its piecewise constant interpolation, see Subsection~\ref{s.discretecurrents}. Since\footnote{By $\llbracket \SS^1 \rrbracket$ we mean the current given by the integration over $\SS^1$ oriented counterclockwise.} $\de G_{u_\e} = - \mu_{u_\e} \x \llbracket \SS^1 \rrbracket$ and  $\mu_{u_\e} \flat 0$, the limit $T$ of the currents~$G_{u_\e}$ satisfies $\de T = 0$ and, more precisely, is a Cartesian current in $\cart(\Omega \x \SS^1)$. For this reason, the limit of $\frac{1}{\e \theta_\e} E_\e$ in this regime shares strong similarities with the $L^1$-relaxation of the $W^{1,1}$-norm of maps in $C^1(\Omega;\SS^1)$, cf.~\cite{Gia-Mod-Sou-S1,Gia-Muc}. The $\Gamma$-limit, cf.~Propositions~\ref{prop:lb} and~\ref{prop:construction of ueps}, features a term reminiscent of the $BV$-type concentration of $|\DD u_\e|$, possibly inevitable to satisfy $\mu_{u_\e} \flat 0$, in general not expressible as an integral functional on the limit of $u_\e$.

Our main theorem concerns the regime where $u_\e$ displays simultaneously vortex-type and $BV$-type concentration effects. The discretization (in domain and codomain) $v_\e \colon \Omega_\e \to \S_\e$ of a vortex $\frac{x-x_0}{|x-x_0|}$ satisfies,  $\frac{1}{\e \theta_\e} E_\e(v_\e) \sim 2 \pi |\log \e | \frac{\e}{\theta_\e}  \to +\infty$ if $\theta_\e \ll \e |\log \e|$, cf.~\eqref{eq:exactconcentration}. To obtain a finer description of the limit, we renormalize $E_\e$ by removing the diverging energy of $M$ vortices and by studying the excess energy $\frac{1}{\e \theta_\e} E_\e(u_\e) - 2\pi M |\log \e| \frac{\e}{\theta_\e}$. A bound on the latter energy yields, cf.~Proposition~\ref{prop:compactness M vortices}, $\mu_{u_\e} \flat \mu  = \sum_{h=1}^N d_h \delta_{x_h}$, $d_h \in \ZZ$, and $|\mu|(\Omega) \leq M$. If $|\mu|(\Omega) = M$, the diverging energy $2\pi M |\log \e| \frac{\e}{\theta_\e}$ has been saturated by $\mu$ and a finite energy $\frac{1}{\e \theta_\e} E_\e$ is still accessible to the system. This might lead to $BV$-type concentration effects, detected by the current $T$, limit of the extended graphs $G_{u_\e}$. Since $\de G_{u_\e} = - \mu_{u_\e} \x \llbracket \SS^1 \rrbracket$ and $\mu_{u_\e} \flat \mu$, $T$ satisfies the nontrivial constraint $\de T = - \mu  \x \llbracket \SS^1 \rrbracket$. This condition couples the vortex-type and $BV$-type concentration effects displayed by the spin field, resulting in a term $\mathcal{J}(\mu,u;\Omega)$ in the $\Gamma$-limit.\footnote{It is given by $\mathcal{J}(\mu,u;\Omega) := \inf \{ \int_{J_T}{\ell_T(x) |\nu_T(x)|_1}{\d \H^1(x)} \colon T \in \mathrm{Adm}(\mu,u;\Omega) \}$, where $\mathrm{Adm}(\mu,u;\Omega)$, defined in~\eqref{eq:def of Adm}, is a suitable class of currents $T$ satisfying, in particular, the constraint $\de T = - \mu \x \llbracket \SS^1 \rrbracket$. Here $J_T$ is the $1$-codimensional jump-concentration set of $T$ oriented by the normal $\nu_T$. At each point $x \in J_T$, the current $T$ has a vertical part, given by a (not necessarily geodesic) arc in $\SS^1$  of length $\ell_T(x)$ which connects the traces of $u$ on the two sides of $J_T$. The set-function $\mathcal{J}(\mu,u;\, \cdot \, )$ is not subadditive.} This leads to our main result.\footnote{In Theorem~\ref{thm:e smaller theta with vortices} the matrix norm $|\, \cdot \,|_{2,1}$ reflects the anisotropy of the lattice, see Section~\ref{sec:notation}.}

\begin{theorem} \label{thm:e smaller theta with vortices}
	Assume that $\e \ll \theta_\e \ll \e |\log \e|$. Then the following results hold true:
	\begin{itemize}[leftmargin=*]
		\item[i)] (Compactness) Let $M \in \NN$ and let $u_\e \colon \Omega_\e \to \S_\e$ be such that $\frac{1}{\e \theta_\e} E_\e(u_\e) - 2\pi M |\log \e| \frac{\e}{\theta_\e} \leq C$.
		Then there exists a measure $\mu = \sum_{h=1}^N d_h \delta_{x_h}$, $x_h \in \Omega$, $d_h \in \ZZ$, such that (up to a subsequence) $\mu_{u_\e} \flat \mu$ and $|\mu|(\Omega) \leq M$. If, in addition, $|\mu|(\Omega) = M$, then there exists a function $u \in BV(\Omega;\SS^1)$ such that (up to a subsequence) $u_\e \to u$ in $L^1(\Omega;\RR^2)$. 
		\item[ii)] ($\Gamma$-liminf inequality) Let $u_\e  \colon \Omega_\e \to \S_\e$ and let  $\mu = \sum_{h=1}^N d_h \delta_{x_h}$, $x_h \in \Omega$, $d_h \in \ZZ$ with $|\mu|(\Omega) = M$. Assume that $\mu_{u_\e} \flat \mu$. Let $u \in BV(\Omega;\SS^1)$ be such that $u_\e \to u$ in $L^1(\Omega;\RR^2)$. Then 
			\begin{equation*}
				 \int_{\Omega}{|\nabla u|_{2,1} }{\d x} + |\DD^{(c)} u|_{2,1}(\Omega) + \mathcal{J}(\mu,u;\Omega) \leq \liminf_{\e \to 0} \Big( \frac{1}{\e \theta_\e} E_\e(u_\e) - 2\pi M |\log \e| \frac{\e}{\theta_\e} \Big) \, .
			\end{equation*} 
		\item[iii)] ($\Gamma$-limsup inequality) Let  $\mu = \sum_{h=1}^N d_h \delta_{x_h}$, $x_h \in \Omega$, $d_h \in \ZZ$ with $|\mu|(\Omega) = M$ and let $u \in BV(\Omega;\SS^1)$. Then there exists a sequence $u_\e  \colon \Omega_\e \to \S_\e$ such that $\mu_{u_\e} \flat \mu$, $u_\e \to u$ in $L^1(\Omega;\RR^2)$, and
		\begin{equation*}
				\limsup_{\e \to 0}\Big( \frac{1}{\e \theta_\e} E_\e(u_\e) - 2\pi M |\log \e| \frac{\e}{\theta_\e} \Big) \leq \int_{\Omega}{|\nabla u|_{2,1} }{\d x} + |\DD^{(c)} u|_{2,1}(\Omega) + \mathcal{J}(\mu,u;\Omega)  \, .
		\end{equation*}  
	\end{itemize}
\end{theorem}

The case $\theta_\e \sim \e |\log \e|$ is studied in~\cite{Cic-Orl-Ruf}. If $\theta_\e \ll \e$, in~\cite{Cic-Orl-Ruf} we prove that $\frac{1}{\e^2} E_\e(u_\e) - 2\pi M |\log \e|$ approximates the renormalized and core energies obtained in the first order analysis of the $XY$ model carried out in~\cite{Ali-DL-Gar-Pon}. Instead, Theorem~\ref{thm:e smaller theta with vortices} points out that the $N$-clock and $XY$ models exhibit different asymptotic behaviors if $\e \ll \theta_\e \ll \e |\log \e|$. This is due to the arising of a surprising interaction between vortex-type and  $BV$-type concentration effects. Coexistence of singularities of two different dimensions has been already observed in other models, e.g.~\cite{Bad-Cic-DL-Pon, Gol-Mer-Mil}. The difference is that here they naturally appear as a result of both the dependence on $\e$ of the codomain and of topological obstructions. 

We highlight here some of the main technical difficulties in the very delicate construction of the recovery sequence in the proof of Theorem~\ref{thm:e smaller theta with vortices}. Given $u \in BV(\Omega;\SS^1)$, we define its recovery sequence following a gradual approximation procedure, which involves a series of steps of increasing complexity. 
At each of these steps, the map $u$ is modified without essentially changing the energy. 

The first main issue is to regularize the map $u$. Maps in $BV(\Omega;\SS^1)$ cannot always be approximated in energy by $\SS^1$-valued smooth functions (in general they cannot be lifted without increasing the $BV$-norm~\cite{Ign,Can-Orl}). Nonetheless, the result in~\cite{Bet} (see also~\cite{Ali-CE-Leo}) guarantees the density of $\SS^1$-valued maps that are smooth outside finitely many point-singularities.
These are related to the vorticity measure~$\mu$ using the Approximation Theorem for Cartesian currents, cf.\ Lemma~\ref{lemma:approximation with sing}.
The next main issue is to construct a recovery sequence~$u_\e$ for such a regularization of~$u$. Close to each singularity, $u_\e$ is defined by discretizing (in domain and codomain) a proper translation of $\frac{x}{|x|}$. The energy carried by this discrete spin field close to a singularity
diverges as $2 \pi |\log \e | \frac{\e}{\theta_\e}$. Far from the singularities, the problem reduces to the construction of a recovery sequence for a smooth $\SS^1$-valued map. This can be further simplified to the case of a piecewise constant $\SS^1$-valued map by introducing a mesoscopic scale into the problem, see Lemma~\ref{lemma:discretization of smooth wout sing}. For such maps, the construction is a refinement of the one described above to obtain~\eqref{eqintro:angle interpolation}. 

The most delicate step is to merge the different parts of the recovery sequence close to and far from the singularities. This is achieved in the proof of Proposition~\ref{p.smoothapprox} (Step~2) by a careful interpolation on dyadic layers of mesoscopic squares, whose size is chosen to be smaller for layers closer to the singularity. At each layer generation, $\frac{x}{|x|}$ is sampled at a different mesoscopic length-scale. The latter is optimized in order to provide the correct control on the energy in progressing from each layer to the next one.

\section{Notation and preliminaries}\label{sec:notation}
We denote the imaginary unit by $\iota$. We shall identify $\RR^2$ with~$\CC$. Given $a = (a_1,a_2) \in \RR^2$, its $1$-norm is $|a|_1 = |a_1|+|a_2|$. We define the $|\cdot|_{2,1}$-norm of a matrix $A = (a_{ij}) \in \RR^{2 \times 2}$ by $|A|_{2,1} := \big(a^2_{11} + a^2_{21}\big)^{1/2} + \big(a^2_{12} + a^2_{22}\big)^{1/2}$.

If $u, v \in \SS^1$, their geodesic distance on $\SS^1$ is denoted by $\geo(u,v)$. It is given by the angle in $[0,\pi]$ between the vectors $u$ and $v$, i.e., $\geo(u,v) = \arccos(u\cdot v)$. Observe that
\begin{equation} \label{eq:geo and eucl}
\tfrac{1}{2}|u - v| = \sin\big( \tfrac{1}{2}\geo(u,v) \big)\quad\text{ and }\quad |u - v| \leq \geo(u,v) \leq \frac{\pi}{2}|u - v|\, .
\end{equation}
Given two sequences $\alpha_{\e}$ and $\beta_{\e}$, we write $\alpha_{\e}\ll \beta_{\e}$ if $\lim_{\e \to 0}\tfrac{\alpha_{\e}}{\beta_{\e}}=0$. We will use the notation $\deg(u)(x_0)$ to denote the topological degree of a continuous map $u \in C(B_\rho(x_0) \sm \{x_0\}; \SS^1)$, i.e., the topological degree of its restriction $u|_{\de B_r(x_0)}$, independent of $r < \rho$. We let $I_\lambda(x)$ be the half-open square given by $I_{\lambda}(x) = x + [0,\lambda)^2$.

By $BV(\Omega;\SS^1)$ we denote the space of $\SS^1$-valued $BV$-functions. We refer the reader to \cite{Amb-Fus-Pal} for a detailed introduction to the theory of $BV$-functions.

\subsection{Results for the classical $XY$ model}
We recall here some results when the spin field $u_\e \colon \Omega_\e \to \SS^1$ is not constrained to take values in a discrete set. Following~\cite{Ali-Cic-Pon}, in order to define the discrete vorticity of $u_{\e}$, we introduce the projection $Q \colon \RR \to 2 \pi \ZZ$ defined by
\begin{equation} \label{eq:def of projection}
Q(t) := \argmin \{|t - s| \ : \ s \in 2 \pi \ZZ \} \, ,
\end{equation}
with the convention that, if the argmin is not unique, then we choose the one with minimal modulus. Then for every $t \in \RR$ we define 
$\Psi(t) := t - Q(t)  \in [-\pi, \pi]$.

%

	
	

Let $u  \colon \e \ZZ^2 \to \SS^1$  and let $\varphi \colon \e \ZZ^2 \to [0, 2\pi)$ be the phase of $u$ defined by the relation $u = \exp(\iota \varphi)$. 
The discrete vorticity of $u$ is defined for every $\e i \in \e \ZZ^2$ by 
\begin{equation} \label{eq:discrete vorticity}
\begin{split}
d_u(\e i) := & \frac{1}{2\pi} \Big[ \Psi\big(\varphi(\e i + \e e_1) - \varphi(\e i) \big) + \Psi\big(\varphi(\e i + \e e_1 + \e e_2) - \varphi(\e i + \e e_1) \big)  \\
& \quad + \Psi\big(\varphi(\e i + \e e_2) - \varphi(\e i + \e e_1 + \e e_2) \big) + \Psi\big(\varphi(\e i) - \varphi(\e i + \e e_2) \big)  \Big] \, . 
\end{split}
\end{equation}
As already noted in~\cite{Ali-Cic-Pon}, it holds that $d_u\in\{-1,0,1\}$, i.e., only singular vortices can be present in the discrete setting. The discrete vorticity measure associated to $u$ is given by
\begin{equation} \label{eq:discrete vorticity measure}
\mu_u := \sum_{\e i \in \e \ZZ^2} d_u(\e i) \delta_{\e i + (\e,\e)} \, . 
\end{equation}

We recall the following compactness and lower bound for the $XY$ model. 

\begin{proposition} \label{prop:XY classical}
Let $u_\e \colon \Omega_\e \to \SS^1$ and assume that  $\frac{1}{\e^2|\log \e|} XY_\e(u_\e) \leq C$ for some $C>0$. Then there exists a measure $\mu \in \M_b(\Omega)$ of the form $\mu=\sum_{h=1}^Nd_h\delta_{x_h}$ with $d_h\in\ZZ$ and $x_h\in\Omega$, and a subsequence (not relabeled) such that $\mu_{u_\e} \mres \Omega \flat \mu$. Moreover 
\begin{equation*}
2 \pi |\mu|(\Omega)  \leq \liminf_{\e \to 0} \frac{1}{\e^2|\log \e|} XY_\e(u_\e) \, .
\end{equation*}
\end{proposition}


\begin{remark} \label{rmk:vortices vanish}
Observe that in the regime $\theta_{\e}\ll\e|\log \e|$ the bound $\frac{1}{\e \theta_\e} E_\e(u_\e) \leq C$ and Proposition~\ref{prop:XY classical} imply that $\mu_{u_\e} \mres \Omega \flat 0$.
\end{remark}

\section{Currents}

For the theory of currents we refer to~\cite{Fed,Gia-Mod-Sou-I, Gia-Mod-Sou-II}. We recall here some basic facts. 

\subsection{Definitions and basic facts} \label{sec:currentsbasics}
Given an open set $O \subset \RR^d$, we denote by $\D^k(O)$ the space of $k$-forms $\omega \colon O \mapsto \Lambda^k \RR^d$ with $C^\infty_c(O)$-coefficients. A {\em $k$-current} $T \in \D_k(O)$ is an element of the dual of~$\D^k(O)$ and we write $T(w)$ for the duality. The {\em boundary} of a $k$-current $T$ is the $(k{-}1)$-current $\de T \in \D_{k-1}(O)$ defined by $\de T(\omega) := T(\! \d \omega)$ for every $\omega \in \D^{k-1}(O)$  (or $\partial T:=0$ if $k=0$). The {\em support} of a current $T$ is the smallest relatively closed set $K$ in $O$ such that $T(\omega) = 0$ if $\omega$ is supported outside $K$. 
Given a  smooth map $f \colon O \to O' \subset \RR^{N'}$ such that $f$ is proper\footnote{that means, $f^{-1}(K)$ is compact in $O$ for all compact sets $K\subset O'$.}, $f^\# \omega  \in \D^k(O)$ denotes the pull-back of a $k$-form $\omega \in \D^k(O')$ through $f$. The {\em push-forward} of a $k$-current $T \in \D_k(O)$ is the $k$-current $f_\# T \in \D_k(O')$ defined by $f_\# T(\omega) := T(f^\# \omega)$. Given a $k$-form $\omega \in \D^k(O)$, we can write it via its components
$\omega = \sum_{|\alpha| = k} \omega_\alpha \d x^\alpha$ with $\omega_\alpha \in C^\infty_c(O)$, where the expression $|\alpha|=k$ denotes all multi-indices $\alpha = (\alpha_1, \dots, \alpha_k)$ with $1 \leq \alpha_i \leq d$, and $\d x^\alpha = \d x^{\alpha_1} \w \dots \w \d x^{\alpha_k}$. The norm of $\omega(x)$ is denoted by $|\omega(x)|$ and it is the Euclidean norm of the vector with components $(\omega_\alpha(x))_{|\alpha|=k}$. The {\em total variation} of $T \in \D_k(O)$ is defined by
\begin{equation*}
|T|(O) := \sup \{T(\omega) \ : \ \omega \in \D^k(O), \ |\omega(x)| \leq 1 \} \, .
\end{equation*}
If $T \in \D_k(O)$ with $|T|(O) < \infty$, then we can define the measure $|T| \in \M_b(O)$ by
\begin{equation*}
|T|(\psi) := \sup \{T(\omega) \ : \ \omega \in \D^k(O), \ |\omega(x)| \leq \psi(x) \},\quad\psi \in C_0(O),\; \psi \geq 0\,.
\end{equation*}
Due to Riesz's Representation Theorem (see \cite[2.2.3, Theorem 1]{Gia-Mod-Sou-I}) there exists a $|T|$-measurable function $\vec{T} \colon O \mapsto \Lambda_k \RR^d$ with $|\vec{T}(x)| = 1$ for $|T|$-a.e.\ $x \in O$ such that 
\begin{equation} \label{eq:representation}
T(\omega) = \integral{O}{\langle \omega(x), \vec{T}(x) \rangle}{ \d |T|(x) }
\end{equation}
for every $\omega \in \D^k(O)$. If $T$ has finite total variation, then it can be extended to a linear functional acting on all forms with bounded, Borel-measurable coefficients  via the dominated convergence theorem. 
In particular, in this case the push-forward $f_\# T$ can be defined also for $f\in C^1(O,O')$ with bounded derivatives, cf.\ the discussion in \cite[p. 132]{Gia-Mod-Sou-I}.

A set $\M \subset O$ is a countably $\H^k$-rectifiable set if it can be covered, up to an $\H^k$-negligible subset, by countably many $k$-manifolds of class $C^1$. As such, it admits at $\H^k$-a.e.\ $x \in \M$ a tangent space $\mathrm{Tan}(\M,x)$ in a measure theoretic sense. A current $T \in \D_k(O)$ is an {\em integer multiplicity (i.m.) rectifiable current} if it is representable as 
\begin{equation} \label{eq:im rectifiable}
	T(\omega) = \integral{\M}{\langle \omega(x), \xi(x) \rangle \theta(x)}{\d \H^k(x)} \, ,\quad \text{for } \omega \in \D^k(O) \, ,
\end{equation}
where $\M \subset O$ is a $\H^k$-measurable and countably $\H^k$-rectifiable set, $\theta \colon \M \to \ZZ$ is  locally $\H^k \mres \M$-summable, and $\xi \colon \M \to \Lambda_k \RR^d$ is a $\H^k$-measurable map such that $\xi(x)$ spans $\mathrm{Tan}(\M,x)$ and $|\xi(x)| = 1$ for $\H^k$-a.e.\ $x \in \M$. 
We use the short-hand notation $T=\tau(\mathcal{M},\theta,\xi)$.  
One can always remove from $\mathcal{M}$  the set $\theta^{-1}(\{0\})$, so that we may always assume that $\theta\neq 0$. 
Then the triple $(\mathcal{M},\theta,\xi)$ is uniquely determined up to $\mathcal{H}^k$-negligible modifications. Moreover, one can show, according to the Riesz's representation in~\eqref{eq:representation}, that $\vec{T}=\xi$ and the total variation\footnote{For i.m.\ rectifiable currents, the total variation coincides with the so-called mass. Hence, we will not distinguish between these two concepts.} is given by $|T|=|\theta|\mathcal{H}^k\mres\mathcal{M}$. 
   
If $T_j$ are i.m.\ rectifiable currents and $T_j \weak T$ in $\D_k(O)$ with $\sup_j (|T_j|(V) + |\de T_j|(V) ) < +\infty$ for every $V \compact O$, then by the Closure Theorem~\cite[2.2.4, Theorem 1]{Gia-Mod-Sou-I} $T$ is an i.m.\ rectifiable current, too. By $\llbracket \M \rrbracket$ we denote the current defined by integration over $\M$.

\subsection{Currents in product spaces} We recall some notation for currents defined on the product space~$\RR^{d_1} \x \RR^{d_2}$. Let us denote by $(x,y)$ the points in this space. The standard basis for $\RR^{d_1}$ is $\{ e_1,\ldots, e_{d_1} \}$, while $\{ \bar e_1,\ldots, \bar e_{d_2} \}$ is the standard basis of for $\RR^{d_2}$. Given $O_1\subset \RR^{d_1},O_2 \subset \RR^{d_2}$ open sets, $T_1 \in \D_{k_1}(O_1)$, $T_2 \in \D_{k_2}(O_2)$ and a $(k_1+k_2)$-form $\omega \in \D^{k_1+k_2}(O_1 \x O_2)$ of the type
\begin{equation*}
\begin{split}
\omega(x,y) = \sum_{\substack{|\alpha|=k_1\\ |\beta|=k_2}} \omega_{\alpha \beta}(x,y) \d x^\alpha \w \d y^\beta \, ,
\end{split}
\end{equation*} 
the product current $T_1 \times T_2 \in \D_{k_1+k_2}(O_1 \x O_2)$ is defined by
\begin{equation*}
T_1 \x T_2( \omega) := T_1 \Big( \sum_{|\alpha|=k_1} T_2\Big(\sum_{|\beta| = k_2} \omega_{\alpha \beta}(x,y) \d y^\beta \Big) \d x^\alpha \Big),
\end{equation*}
while $T_1 \x T_2(\phi \d x^\alpha \w \d y^\beta) = 0$ if $|\alpha|+|\beta| =k_1+k_2$ but $|\alpha|\neq k_1$, $|\beta| \neq k_2$.


\subsection{Graphs} Let $O \subset \RR^d$ be an open set and $u \colon \Omega \to \RR^2$ a Lipschitz map. Then we can consider the $d$-current associated to the graph of $u$ given by $G_u := (\mathrm{id} \x u)_\# \llbracket O \rrbracket \in \D_2(O\x \RR^2)$, where $\mathrm{id} \x u \colon O \to O \x \RR^2$ is the map $(\mathrm{id} \x u)(x) = (x,u(x))$.  Note that
\begin{equation*}
G_u(\omega)=\integral{O}{\langle \omega(x,u(x)),M(\nabla u(x))\rangle}{\d x}
\end{equation*}
for all $\omega\in\mathcal{D}^{d}(O\times \RR^2)$, with the $d$-vector \begin{equation}\label{eq:minors}
M(\nabla u)=(e_1+\partial_{x^1} u^1\bar e_1+\partial_{x^1} u^2\bar e_2)\wedge\ldots\wedge(e_d+\partial_{x^d}u^1\bar e_1+\partial_{x^d} u^2\bar e_2)\,.
\end{equation}

%
%
%


Later on we use the orientation of the graph of a smooth function $u\colon O\subset\RR^2\to\SS^1$ (cf. \cite[2.2.4]{Gia-Mod-Sou-I}). For such maps we have $|G_u| = \H^2 \mres \M$, where $\M = (\mathrm{id} \x u)(\Omega)$, and 
\begin{equation} \label{eq:smooth components}
\begin{split}
\sqrt{1+|\nabla u(x)|^2} \ \vec{G}_u(x,y) = & \ e_1 \wedge e_2 \\
+ & \ \de_{x^2} u^1(x) e_1 \wedge \bar e_1 + \de_{x^2} u^2(x) e_1 \wedge \bar e_2 \\
- & \ \de_{x^1}  u^1(x)  e_2 \wedge \bar e_1 - \de_{x^1}  u^2(x) e_2 \wedge \bar e_2\quad\quad\text{for all }(x,y) \in \M \, .
\end{split} 
\end{equation}

\subsection{Cartesian currents}
Let $O\subset\RR^d$ be a bounded, open set. We recall that the class of {\em cartesian currents} in~$O\x \RR^2$ is defined by
\begin{equation*}
\begin{split}
\cart(O \x \RR^2) := \{ & T \in \D_d(O \x \RR^2) \ : \ T \text{ is i.m.\ rectifiable, } \de T|_{O \x \RR^2} = 0, \\
& \pi^O_\# T = \llbracket O \rrbracket \, , \ T |_{\d x} \geq 0 \, , \ |T| < +\infty \, , \ \|T\|_1 < +\infty \} \, ,
\end{split}
\end{equation*}
where $\pi^O \colon O \x \RR^2 \to O$ denotes the projection on the first component, $T|_{\d x} \geq 0$ means that $T(\phi(x,y) \d x) \geq 0$ for every $\phi \in C^\infty_c(O \x \RR^2)$ with $\phi \geq 0$, and 
\begin{equation*}
\| T \|_1 = \sup \{ T(\phi(x,y)|y| \d x ) \ : \ \phi \in C^\infty_c(O \x \RR^2) \, , \ |\phi| \leq 1 \}  \, . 
\end{equation*}
Note that, if for some function $u$
\begin{equation} \label{eq:norm 1}
T(\phi(x,y) \d x) = \integral{O}{\phi(x,u(x))}{\d x} \quad \text{then} \quad \| T \|_1  = \integral{O}{|u|}{\d x} \, .
\end{equation}

The class of {\em cartesian currents} in $O \x \SS^1$ is 
\begin{equation*}
\cart(O \x \SS^1) := \{ T \in \cart(O \x \RR^2) \ : \ \supp(T) \subset \ol O \x \SS^1 \} \,,
\end{equation*}
(cf.\ \cite[6.2.2]{Gia-Mod-Sou-II} for this definition). We recall the following approximation theorem which explains that cartesian currents in $O \x \SS^1$ are precisely those currents that arise as limits of graphs of $\SS^1$-valued smooth maps. The proof can be found in~\cite[Theorem 7]{Gia-Mod-Sou-S1}.\footnote{Notice that some results in~\cite{Gia-Mod-Sou-S1} require $O$ to have smooth boundary. This is not the case for this theorem, which is based on a local construction regularizing a local lifting of $T$.}

\begin{theorem}[Approximation Theorem] \label{thm:approximation}
Let $T \in \cart(O \x \SS^1)$. Then there exists a sequence of smooth maps $u_h \in C^\infty(O;\SS^1)$ such that 
\begin{equation*}
G_{u_h} \weak T \quad \text{in } \D_d(O \x \RR^2)\quad\text{ and }\quad|G_{u_h}|(O\x \RR^2) \to |T|(O \x \RR^2) \, .
\end{equation*}
\end{theorem}

We state an extension result for cartesian currents, which we could not find in the literature. For a proof we refer the interested reader to~\cite{Cic-Orl-Ruf}.
\begin{lemma}[Extension of cartesian currents] \label{lemma:extension of currents}
Let $O \subset \RR^d$ be a bounded, open set with Lipschitz boundary and let $T \in \cart(O \x \SS^1)$. Then there exist an open set $\tilde O \Supset O$ and a current $T \in \cart(\tilde O \x \SS^1)$ such that $\tilde T|_{O \x \RR^2} = T$ and $|\tilde T|(\de O \x \RR^2) = 0$. 
\end{lemma}

We will also use the structure theorem for cartesian currents in $O \x \SS^1$ \cite[Section 3, Theorems 1, 5, 6]{Gia-Mod-Sou-S1}.\footnote{As for the Approximation Theorem, no boundary regularity is required for this result.} To simplify notation, from now on we focus on dimension two. Recall that $\Omega\subset\RR^2$ is a bounded, open set with Lipschitz boundary. To state the theorem, we recall the following decomposition for a current $T \in \cart(\Omega \x \SS^1)$. Letting $\M$ be the countably $\H^2$-rectifiable set where $T$ is concentrated, we denote by $\M^{(a)}$ the set of points  $(x,y) \in \M$ at which the tangent plane $\mathrm{Tan}(\M,(x,y))$ does not contain vertical vectors (namely, the Jacobian of the projection $\pi^\Omega$ restricted to $\mathrm{Tan}(\M,(x,y))$ has maximal rank), by $\M^{(jc)} := (\M \sm \M^{(a)}) \cap (J_T \x \SS^1)$, where $J_T := \{ x \in \Omega \ : \ \frac{\d \pi^\Omega_\# |T|}{\d \H^1}(x) > 0 \}$, and by $\M^{(c)} := \M \sm (\M^{(a)} \cup \M^{(jc)})$. Then we can split the current via $T = T^{(a)} + T^{(c)} + T^{(jc)}$, where $T^{(a)} := T \mres \M^{(a)}$, $T^{(c)} := T \mres \M^{(c)}$, $T^{(jc)} := T \mres \M^{(jc)}$ are mutually singular measures, and we denote by $\mres$ the restriction of the Radon measure $T$.
Hereafter we use the notation $ \widehat x^1 =  x^2$ and $ \widehat x^2 = x^1$ .

\begin{theorem}[Structure Theorem for $\cart(\Omega \x \SS^1)$] \label{thm:structure} 
Let $T \in \cart(\Omega \x \SS^1)$. Then there exists a unique map $u_T \in BV(\Omega;\SS^1)$ and an (not unique) i.m.\ rectifiable 1-current $L_T =\tau(\L,k,\vec{L}_T) \in \D_1(\Omega)$  such that $T^{(jc)} = T^{(j)} + L_T \x \llbracket \SS^1 \rrbracket$ and
\begin{align}
T(\phi(x,y) \d x) & = T^{(a)}(\phi(x,y) \d x)  = \integral{\Omega}{\phi(x,u_T(x))}{\d x} \, ,   \label{eq:T horiz} \\
T^{(a)}(\phi(x,y) \d \widehat x^l \w \d y^m) & = (-1)^{2-l} \integral{\Omega}{\phi(x,u_T(x)) \de^{(a)}_{x^l} u_T^m(x)}{\d x} \, , \\
T^{(c)}(\phi(x,y) \d \widehat x^l \w \d y^m) &= (-1)^{2-l} \integral{\Omega}{\phi(x,\tilde u_T(x)) }{\d \de^{(c)}_{x^l} u_T^m(x)} \, , \label{eq:T cantor} \\
T^{(j)}(\phi(x,y) \d \widehat x^l \w \d y^m) &= (-1)^{2-l} \integral{J_{u_T}}{\bigg\{\integral{\gamma_x}{\phi(x,y)}{\d y^m}\bigg\} \nu_{u_T}^l(x)}{\d \H^1(x)} \label{eq:T jump}
\end{align}
for every $\phi \in C^\infty_c(\Omega \x \RR^2)$, $\gamma_x$ being the (oriented) geodesic arc in $\SS^1$ that connects $u_T^-(x)$ to $u_T^+(x)$ and $\tilde u_T$ being the precise representative of $u_T$. \footnote{\label{footnote:geodesics} In \cite[Theorem 6]{Gia-Mod-Sou-S1}	the structure of $T^{(j)}$ is formulated slightly differently with the counter-clockwise arc $\gamma_{\varphi^-,\varphi^+}$ between $(\cos(\varphi^-),\sin(\varphi^-))$, $(\cos(\varphi^+),\sin(\varphi^+))$, $\varphi^-<\varphi^+$, and replacing $J_{u_T}$ by $J_{\varphi}$, where $\varphi\in BV(\Omega)$ is such that $T=\chi_{\#}G_{\varphi}$, where $\chi(x,\vartheta)=(x,\cos(\vartheta),\sin(\vartheta))$ and $G_\varphi \in \cart(\Omega \x \RR)$ is the boundary of the subgraph of $\varphi$.
To explain~\eqref{eq:T jump}, we recall the local construction in~\cite{Gia-Mod-Sou-S1}: for every $x\in J_{\varphi}$ one chooses $p^+(x) \in \RR$ and $k'(x)\in\mathbb{N}$ such that $\varphi^+(x)=p^+(x)+2\pi k'(x)$ and $0\leq p^+(x)-\varphi^-(x)<2\pi$.
Then, locally, the $1$-current $L_T'$ in \cite[Theorem 6]{Gia-Mod-Sou-S1} is given by $L_T'=\tau(\mathcal{L}',k'(x),\vec{L}_T')$, where $\mathcal{L}'\subset J_{\varphi}$ is the set of points with $k'(x)\geq 1$ and $\vec{L}_T' = \nu_{\varphi}^2e_1 - \nu_{\varphi}^1e_2$. To obtain the representation via geodesics, we let 
\begin{equation*}
(q^+(x), k(x))=
\begin{cases}
(p^+(x),k'(x)) &\mbox{if $p^+(x)-\varphi^-(x)<\pi$} \, ,
\\
(p^+(x)-2\pi,k'(x)+1) &\mbox{if $p^+(x)-\varphi^-(x)>\pi$} \, ,
\end{cases}
\end{equation*}
The case $p^+(x)-\varphi^-(x)=\pi$ needs special care. In this case we let $\tilde{\varphi}^\pm(x):=\varphi^{\pm}(x)\mod 2\pi \in [0,2 \pi)$ and we set
\begin{equation*}
(q^+(x),k(x))=
\begin{cases}
(p^+(x),k'(x)) &\mbox{if $ \tilde{\varphi}^+(x) - \tilde{\varphi}^-(x)  = \pi$} \, ,
\\
(p^+(x)-2\pi,k'(x)+1) &\mbox{if $ \tilde{\varphi}^+(x) - \tilde{\varphi}^-(x)  = -\pi$} \, .
\end{cases}
\end{equation*}
Replacing $(p^+(x),k'(x))$ by $(q^+(x), k(x))$, one proves~\eqref{eq:T jump} as in~\cite[p.107-108]{Gia-Mod-Sou-S1}. The curves $\gamma_{\varphi^-,\varphi^+}$ are then replaced by the more intrinsic geodesic arcs $\gamma_x$. Exchanging $u_T^-(x)$ and $u_T^+(x)$ will change the orientation of the arc and of the normal $\nu_{u_T}(x)$, making~\eqref{eq:T jump} invariant.}
\end{theorem}

It is convenient to recast the jump-concentration part of $T \in \cart(\Omega \x \SS^1)$ in the following way. Let $L_T = \tau(\L,k,\vec{L}_T)$ as in Theorem~\ref{thm:structure}. We introduce for $\H^1$-a.e.\ $x \in J_T$ the normal $\nu_T(x)$ to the 1-rectifiable set $J_T = J_{u_T} \cup \L$ as 
\begin{equation}\label{eq:defnormal}
\nu_T(x)=
\begin{cases}
\nu_{u_T}(x) &\mbox{if $x\in J_{u_T}$} \, ,
\\
(-\vec{L}^2_T(x), \vec{L}^1_T(x)) &\mbox{if $x\in\mathcal{L}\setminus J_{u_T}$} \, ,
\end{cases}
\end{equation}
where we choose $\nu_{u_T}(x) = (-\vec{L}^2_T(x), \vec{L}^1_T(x))$ if $x \in \mathcal{L} \cap J_{u_T}$. For $\H^1$-a.e.\ $x \in J_T$ we consider the curve~$\gamma^T_x$ given by: the (oriented) geodesic arc $\gamma_x$ connecting~$u_T^-(x)$ to~$u_T^+(x)$ if $x \in J_{u_T} \sm \L$ (in the sense of Footnote~\ref{footnote:geodesics} for antipodal points); the whole~$\SS^1$ turning $k(x)$ times if $x \in \L \sm J_{u_T}$; the sum (in the sense of currents)\footnote{In this case, a more elementary definition of $\gamma_x^T$ is the following: let $\gamma_x \colon [0,1] \to \SS^1$ be the geodesic arc, and let $\varphi_x \colon [0,1] \to \RR$ be a continuous function (unique up to translations of an integer multiple of $2 \pi$) such that $\gamma_x(t) = \exp(\iota \varphi_x(t))$. Then $\gamma_x^T(t) = \exp\big(\iota (1-t) \varphi_x(0) + \iota t(\varphi_x(1) + 2 \pi k(x) )\big)$.} of the oriented geodesic arc~$\gamma_x$ and of $\SS^1$ with multiplicity $k(x)$ if $x \in J_{u_T} \cap \L$. Then
\begin{equation} \label{eq:jc part of T}
T^{(jc)}(\phi(x,y) \d \widehat x^l \w  \d y^m) = (-1)^{2-l}\integral{J_T}{\Big\{ \integral{ \gamma^T_x}{\phi(x,y)}{\d y^m} \Big\} \nu_T^{l}(x)}{\d \H^1(x)}  \, .
\end{equation}
The integration over $\gamma_x^T$ with respect to the form $\d y^m$ in the formula above is intended with the correct multiplicity of the curve $\gamma_x^T$ defined for $\H^1$-a.e.\ $x \in J_T$ by the integer 
\begin{equation} \label{eq:jc multiplicity}
\mathfrak{m}(x,y) := \begin{cases}
\pm 1 \, , & \text{if } x \in J_{u_T} \sm \L \, ,\  y \in \supp(\gamma_x) \, , \\
k(x) \, , & \text{if } x \in \L \sm J_{u_T} \, ,\ y \in \SS^1,\\
k(x) \pm 1 \, , & \text{if } x \in \L \cap J_{u_T} \, ,\  y \in \supp(\gamma_x)  \, ,  \\
k(x) \, , & \text{if } x \in \L \cap J_{u_T} \, ,\  y \in \supp(\gamma^T_x) \sm \supp(\gamma_x)  \, ,
\end{cases}
\end{equation}
where $\pm = +/ -$ if the geodesic arc $\gamma_x$ is oriented counterclockwise/clockwise, respectively. More precisely,
\begin{equation} \label{eq:from dym to H1}
\integral{ \gamma^T_x}{\phi(y)}{\d y^m} = (-1)^m \,  \!\!\! \integral{\supp(\gamma^T_x)}{\phi(y) \widehat y^m \mathfrak{m}(x,y)}{\d \H^1(y)}  \, .
\end{equation}
\begin{remark}\label{rmk:choiceoforientation}
Note that we constructed $\mathfrak{m}(x,y)$ based on the orientation \eqref{eq:defnormal} of $\nu_T$. As discussed in Footnote~\ref{footnote:geodesics}, changing the orientation of $\nu_{u_T}$ changes the orientation of the geodesic $\gamma_x$, while a change of the orientation of $\vec{L}_T$ switches the sign of $k(x)$. Hence changing the orientation of $\nu_T(x)$ changes $\mathfrak{m}(x,y)$ into $-\mathfrak{m}(x,y)$. If we choose locally $\nu_T=\nu_{\varphi}$ as in Footnote~\ref{footnote:geodesics}, our construction above yields $\mathfrak{m}(x,y)\geq 0$.
\end{remark}

Finally, we recall the following result, proven in~\cite[Section 4]{Gia-Mod-Sou-S1}.
\begin{proposition} \label{prop:supporting BV}
For $u \in BV(\Omega;\SS^1)$ there exists $T \in \cart(\Omega \x \SS^1)$ such that $u_T = u$ a.e.
\end{proposition}

\subsection{Currents associated to discrete spin fields}\label{s.discretecurrents}

%
%

 We introduce the piecewise constant interpolations of spin fields. For every set $S$, we put 
\begin{equation*}
\PC_\e(S) := \{u \colon \RR^2 \to S \ : \ u(x) = u(\e i)   \text{ if } x \in \e i + [0, \e)^2 \text{ for some } i \in  \e \ZZ^2 \}\,.
\end{equation*}
Given $u \colon \Omega_\e \to \SS^1$, we can always identify it with its piecewise constant interpolation belonging to $\PC_\e(\SS^1)$, arbitrarily extended to $\RR^2$.

To $u \in \PC_\e(\SS^1)$ we associate the current $G_{u} \in \D_2(\Omega \x \RR^2)$ defined by
\begin{align}
G_{u}(\phi(x,y) \d x^1 \w  \d x^2) &:= \integral{\Omega}{\phi(x,u(x))}{\d x} \, , \label{eq:Gu ac} \\
G_{u}(\phi(x,y) \d \widehat x^l \w  \d y^m) &:= (-1)^{2-l}\integral{J_{u}}{\bigg\{ \integral{\gamma_x}{ \phi(x,y)}{\d y^m} \bigg\} \nu^l_{u}(x) }{\d \H^1(x)} \, , \label{eq:Gu j}\\
G_{u}(\phi(x,y) \d y^1 \w  \d y^2) &:= 0 \, ,  \label{eq:Gu v}
\end{align}
for every $\phi \in C^\infty_c(\Omega \x \RR^2)$, where $J_{u}$ is the jump set of $u$, $\nu_{u}(x)$ is the normal to $J_{u}$ at $x$, and $\gamma_x \subset \SS^1$ is the (oriented) geodesic arc which connects the two traces $u^-(x)$ and $u^+(x)$. If $u^+(x)$ and $u^-(x)$ are opposite vectors, the choice of the geodesic arc $\gamma_x \subset \SS^1$ is done consistently with the choice made in~\eqref{eq:def of projection} for the values $\Psi(\pi)$ and $\Psi(-\pi)$ as follows: let $\varphi^\pm(x) \in [0,2 \pi)$ be the phase of $u^\pm(x)$; if  $\Psi(\varphi^+(x) - \varphi^-(x)) = \pi$, then $\gamma_x$ is the arc that connects $u^-(x)$ to $u^+(x)$ counterclockwise; if $ \Psi(\varphi^+(x) - \varphi^-(x)) =  -\pi$, then $\gamma_x$ is the arc that connects $u^-(x)$ to~$u^+(x)$ clockwise. Note that the choice of the arc $\gamma_x$ is independent of the orientation of the normal~$\nu_{u}(x)$. 

We define for $\H^1$-a.e.\ $x \in J_{u}$  the integer number $\mathfrak{m}(x) = \pm 1$, where $\pm = +/ -$ if the geodesic arc $\gamma_x$ is oriented counterclockwise/clockwise, respectively. Then
\begin{equation} \label{eq:from dym to H1 2}
\integral{ \gamma_x}{\phi(y)}{\d y^m} = (-1)^m \, \mathfrak{m}(x) \!  \integral{ \supp(\gamma_x)}{\phi(y) \widehat y^m}{\d \H^1(y)}  \, .
\end{equation}
 





			  



The proof of the following proposition is standard.

\begin{proposition} \label{prop:Gu is im}
Let $u \in \PC_\e(\SS^1)$ and let $G_{u} \in \D_2(\Omega \x \RR^2)$ be the current defined in \eqref{eq:Gu ac}--\eqref{eq:Gu v}. Then $G_{u}$ is an i.m.\ rectifiable current and, according to the representation formula~\eqref{eq:representation}, $G_{u} = \vec{G}_u   |G_{u}|$, where $|G_{u}| = \H^2 \mres \M$,
\begin{equation*}
\M = \M^{(a)} \cup \M^{(j)} = \{(x,u(x)) \ : \ x \in \Omega \sm J_{u} \} \cup  \{(x,y) \ : \ x \in J_{u}, \ y \in \gamma_x \} \, ,
\end{equation*}
and, for $\H^2$-a.e.\ $(x,y) \in \M^{(a)}$,
\begin{equation} \label{eq:orientation of Gu 1}
\begin{split}
\vec{G}_{u}(x,y) =  e_1  \wedge e_2
\end{split}
\end{equation}
while for $\H^2$-a.e.\ $(x,y) \in \M^{(j)}$ we have
\begin{equation} \label{eq:orientation of Gu 2}
\begin{split}
\vec{G}_{u}(x,y) = \mathrm{sign}(\mathfrak{m}(x)) \big[  - & \  \nu^2_{u}(x) y^2 e_1 \wedge \bar e_1 +  \nu^2_{u}(x) y^1 e_1 \wedge \bar e_2  \\
 + & \ \nu^1_{u}(x) y^2 e_2 \wedge \bar e_1 - \nu^1_{u}(x) y^1 e_2 \wedge \bar e_2 \big ] \,. 
\end{split}
\end{equation}
\end{proposition}

We now relate the boundary of the current $G_u$ with the discrete vorticity $\mu_{u}$. The interested reader can find a detailed proof in~\cite{Cic-Orl-Ruf}.

\begin{proposition} \label{prop:bd of Gu is mu}
Let $u \in \PC_\e(\SS^1)$ and let $G_{u} \in \D_2(\Omega \x \RR^2)$ be the current defined in \eqref{eq:Gu ac}--\eqref{eq:Gu v}. Then $\de G_{u}|_{\Omega \x \RR^2} = - \mu_{u} \x \llbracket \SS^1 \rrbracket$, where $\mu_{u}$ is the discrete vorticity measure defined in~\eqref{eq:discrete vorticity measure} for $u|_{\e \ZZ^2} \colon \e \ZZ^2 \to \SS^1$.
\end{proposition}

The proof of the following fact follows essentially from the definitions.

\begin{lemma} \label{lemma:flat implied D1}
Assume $\mu_\e \flat \mu$ in $\Omega$. Then $\mu_\e \x \llbracket \SS^1 \rrbracket \weak \mu \x \llbracket \SS^1 \rrbracket$ in $\D_1(\Omega \x \RR^2)$.
\end{lemma}

\section{Proofs}\label{sec:theta<<eloge}
Now we give the proof of our main Theorem~\ref{thm:e smaller theta with vortices}. In what follows, for $A \subset \RR^2$ we shall use the localized energy
\begin{equation*}
	E_\e(u;A) :=  \frac{1}{2} \sum_{\substack{\langle i, j \rangle \\ \e i, \e j \in A }} \e^2 |u(\e i) - u(\e j)|^2 .
\end{equation*} 

\subsection{Compactness and lower bound in absence of vortices}
In this section we consider a generic sequence $u_{\e}\colon \e\ZZ^2\to\S_{\e}$ such that $\tfrac{1}{\e\theta_{\e}}E_{\e}(u_{\e})$ is bounded. First we prove that such sequences are compact in $L^1(\Omega)$ with limits in $BV(\Omega;\SS^1)$.

\begin{proposition}[Compactness in $BV$] \label{prop:BV compactness}
Assume that $\theta_{\e}\ll 1$ and $\frac{1}{\e \theta_\e}E_\e(u_\e) \leq C$. Then there exists a subsequence (not relabeled) and a function $u \in BV(\Omega;\SS^1)$ such that $u_\e \to u$ in $L^1(\Omega)$  and $u_{\e}\wstar u$ in $BV_{\rm loc}(\Omega;\RR^2)$.
\end{proposition}

\begin{proof}
Fix $A \compact \Omega$. Since $\geo(u(\e i), u(\e j)) \geq \theta_\e$ when $u(\e i) \neq u(\e j)$, \eqref{eq:geo and eucl} implies that
\begin{equation} \label{eq:rough lower bound}
\begin{split}
C & \geq \frac{1}{\e \theta_\e}E_\e(u_\e) = \frac{1}{2 \e \theta_\e} \sum_{\nn} \e^2 |u(\e i) - u(\e j)|^2 \\
& = \frac{1}{2 \e \theta_\e} \sum_{\nn} \e^2  2 \sin \Big( \tfrac{1}{2}\geo\big(u(\e i),u(\e j) \big) \Big)|u(\e i) - u(\e j)| \\
& \geq \frac{\sin (\tfrac{\theta_\e}{2})}{ \theta_\e} \sum_{\nn} \e |u(\e i) - u(\e j)| \geq \frac{2\sin (\tfrac{\theta_\e}{2})}{\theta_\e} |\DD u_\e|(A) \, .
\end{split}
\end{equation}
Hence $u_{\e}$ is bounded in $BV(A;\SS^1)$ and we conclude that (up to a subsequence) $u_{\e}\to u$ in $L^1(A)$  and $u_{\e}\wstar u$ in $BV(A;\RR^2)$ for some $u\in BV(A;\SS^1)$ with $|\DD u|(A)\leq C$. Since $A\compact\Omega$ was arbitrary and the constant $C$ does not depend on $A$, the claim follows from a diagonal argument and the equiintegrability of $u_{\e}$.
\end{proof}

In the next lemma we prove lower bound for the energy still at the discrete level.
\begin{lemma} \label{lemma:bound eucl from below}
Assume that $\theta_{\e}\ll 1$ and let $\sigma \in (0,1)$. Then for $\e$ small enough we have
\begin{equation} \label{eq:ineq for sin}
|u_\e(\e i) - u_\e(\e j)|^2 \geq (1 - \sigma)\theta_\e \geo\big(u_\e(\e i),u_\e(\e j)\big) \, .
\end{equation}
In particular 
\begin{equation} \label{eq:first lower bound}
\frac{1}{\e \theta_\e} E_\e(u_\e) \geq (1-\sigma) \frac{1}{2}  \sum_{\langle i, j \rangle} \e \geo\big(u_\e(\e i), u_\e(\e j) \big) \, .
\end{equation}
\end{lemma}
\begin{proof}
We show \eqref{eq:ineq for sin}. By~\eqref{eq:geo and eucl} we have that $|u_\e(\e i) - u_\e(\e j)| = 2 \sin \big(\tfrac{1}{2} \geo(u_\e(\e i), u_\e(\e j) \big)$. Since $u_\e$ takes values in $\S_\e$ there exists $k \in \NN$ (depending on $i$, $j$, and $\e$) such that $\geo(u_\e(\e i),u_\e(\e j)) = k \theta_\e$. We can assume $k \neq 0$. Moreover, note that $k \theta_\e \leq \pi$. 

Due to Taylor's formula there exists a $\zeta \in [0, k \tfrac{\theta_\e}{2}]$ such that
\begin{equation*}
\sin\big(k \tfrac{\theta_\e}{2} \big) = k \tfrac{\theta_\e}{2} - \tfrac{1}{6} \cos(\zeta) \big(k \tfrac{\theta_\e}{2}\big)^3  \geq k \tfrac{\theta_\e}{2} - \tfrac{1}{6} \big(k\tfrac{\theta_\e}{2}\big)^3
\end{equation*}
Dividing by $\sqrt{k} \tfrac{\theta_\e}{2}$ we get that
\begin{equation} \label{eq:06041814}
\frac{\sin\big(k \tfrac{\theta_\e}{2} \big)}{\sqrt{k} \tfrac{\theta_\e}{2}} \geq \sqrt{k} \Big[ 1 - \tfrac{1}{6} \big(k\tfrac{\theta_\e}{2}\big)^2 \Big]
\end{equation}
If $k \geq 9$, using the fact that $k \theta_\e \leq \pi \leq 4$ we obtain that
\begin{equation} \label{eq:06041825}
\frac{\sin\big(k \tfrac{\theta_\e}{2} \big)}{\sqrt{k} \tfrac{\theta_\e}{2}} \geq \sqrt{k} \tfrac{1}{3}  \geq 1 \, .
\end{equation}
Otherwise, if $k \leq 8$, \eqref{eq:06041814} directly implies
\begin{equation} \label{eq:06041826}
\frac{\sin\big(k \tfrac{\theta_\e}{2} \big)}{\sqrt{k} \tfrac{\theta_\e}{2}} \geq 1 - 3 \theta_\e^2  \geq  1 - \sigma \, ,
\end{equation}
for $\e$ small enough. Squaring both sides in \eqref{eq:06041825} and \eqref{eq:06041826} (notice that $k \theta_\e \in [0,\pi]$ implies $\sin\big(k \tfrac{\theta_\e}{2} \big) \geq 0$) we have that $4 \sin^2 \big(k \tfrac{\theta_\e}{2} \big) \geq (1 - \sigma) k \theta^2_\e$. We conclude the proof of~\eqref{eq:ineq for sin} by replacing $k \theta_\e = \geo(u_\e(\e i),u_\e(\e j))$ in the last inequality and by \eqref{eq:geo and eucl}.
\end{proof}

We now recast the energy as a parametric integral of the currents $G_{u_\e}$. To do so, we define the convex and positively $1$-homogeneous function $\Phi \colon \Lambda_2 (\RR^2 \x \RR^2) \mapsto \RR$ by
\begin{equation} \label{eq:Phi is 2,1}
\Phi(\xi) := \sqrt{(\xi^{21})^2 + (\xi^{22})^2 }  + \sqrt{(\xi^{11})^2 + (\xi^{12})^2 }
\end{equation}
for every $\xi = \xi^{\ol 0 0} e_1 \wedge e_2 + \xi^{21} e_1 \wedge \bar e_1 + \xi^{22} e_1 \wedge \bar e_2 + \xi^{11} e_2 \wedge \bar e_1 + \xi^{12} e_2 \wedge \bar e_2 + \xi^{0 \ol 0} \bar e_1 \wedge \bar e_2$.

\begin{lemma} \label{lemma:bound with parametric integral}
For every open set $A \compact \Omega$ and $\e$ small enough we have
\begin{equation*}
 \frac{1}{2} \sum_{\substack{\langle i, j \rangle}} \e \geo\big(u_\e(\e i), u_\e(\e j) \big) \geq \integral{A \x \RR^2}{\Phi(\vec G_{u_\e})}{\d |G_{u_\e}|} \, .
\end{equation*}
\end{lemma}
\begin{proof}
By the explicit formulas~\eqref{eq:orientation of Gu 1}--\eqref{eq:orientation of Gu 2} for the orientation of $G_{u_\e}$ we infer that 
\begin{align*}
\Phi(\vec G_{u_\e})(x,y) &= \mathds{1}_{J_{u_\e}}(x)\Big[ |\nu^2_{u_\e}(x)| \sqrt{(y^2)^2 + (y^1)^2 } + |\nu^1_{u_\e}(x)| \sqrt{ (y^2)^2 +   (y^1)^2 } \Big] 
\\
&= \mathds{1}_{J_{u_\e}}(x) |\nu_{u_\e}(x)|_{1} \, .
\end{align*}
Moreover, we recall that $|G_{u_\e}| = \H^2 \mres \M_\e$, where 
\begin{equation*}
\M_\e = \M_\e^{(a)} \cup \M_\e^{(j)} = \{(x,u_\e(x)) \ : \ x \in \Omega \sm J_{u_\e} \} \cup  \{(x,y) \ : \ x \in J_{u_\e}, \ y \in \supp(\gamma^\e_x) \} \, ,
\end{equation*}
$\gamma^\e_x$ being the geodesic arc that connects $u^-_\e(x)$ to $u^+_\e(x)$. Therefore 
\begin{align*}
\integral{A \x \RR^2}{\Phi(\vec G_{u_\e})}{\d |G_{u_\e}|} & = \integral{A \x \RR^2}{|\nu_{u_\e}|_{1}}{\d \H^2 \mres \M_\e^{(j)}} = \integral{J_{u_\e} \cap A}{\bigg\{ \integral{\supp(\gamma^\e_x)}{\!\!\!\!\!\!\!}{\d \H^1(y)} \bigg\}|\nu_{u_\e}(x)|_{1}}{\d \H^1(x)} \\
& = \integral{J_{u_\e} \cap A}{ \geo\big(u_\e^- , u_\e^+  \big) |\nu_{u_\e} |_{1}}{\d \H^1 } 
 \leq  \frac{1}{2} \sum_{\substack{\langle i, j \rangle}} \e \geo\big(u_\e(\e i), u_\e(\e j) \big) \, .
\end{align*}
\end{proof}

Next we show that energy bounds also yield compactness for the associated currents $G_{u_{\e}}$.
\begin{proposition}[Compactness in $\cart(\Omega \x \SS^1)$] \label{prop:current compactness}
Assume that $\theta_{\e}\ll \e|\log \e|$ as well as  $\frac{1}{\e \theta_\e}E_\e(u_\e) \leq C$. Let $G_{u_\e} \in \D_2(\Omega \x \RR^2)$ be the currents associated to $u_\e$ defined as in~\eqref{eq:Gu ac}--\eqref{eq:Gu v}. Then there exists a subsequence (not relabeled) and a current $T \in \D_2(\Omega \x \RR^2)$ such that $G_{u_\e} \weak T$ in $\D_2(\Omega \x \RR^2)$. Moreover, $T \in \cart(\Omega \x \SS^1)$ and $u_T = u$ a.e.\ in $\Omega$, where $u_T$ is the $BV$ function associated to~$T$ given by Theorem~\ref{thm:structure}  and $u \in BV(\Omega;\SS^1)$ is the function given by Proposition~\ref{prop:BV compactness}. 
\end{proposition}
\begin{proof}
Let us fix an open set $A \compact \Omega$. Since $\Phi(\xi) \geq \sqrt{(\xi^{21})^2 + (\xi^{22})^2 + (\xi^{11})^2 + (\xi^{12})^2 }$, we deduce the estimate
\begin{equation} \label{eq:Gu equibounded}
\begin{split}
|G_{u_\e}|(A \x \RR^2) & = |G_{u_\e}|(\M^{(a)} \cap A \x \RR^2) + |G_{u_\e}|(\M^{(j)} \cap A \x \RR^2) \\
& \leq  |A| + \integral{A \x \RR^2}{\Phi(\vec{G}_{u_\e})}{\d |G_{u_\e}|} \leq |\Omega| +  \frac{1}{2} \sum_{\substack{\langle i, j \rangle}} \e \geo\big(u_\e(\e i), u_\e(\e j) \big) \\
& \leq |\Omega| + \frac{2}{\e\theta_\e} E_\e(u_\e) \leq C \, ,
\end{split}
\end{equation}
where in the last inequality we employed~\eqref{eq:first lower bound} with $\sigma = 1/2$. By the Compactness Theorem for currents~\cite[2.2.3, Proposition~2 and Theorem 1-(i)]{Gia-Mod-Sou-I} we deduce that there exists a subsequence (not relabeled) and a current $T \in \D_2(\Omega \x \RR^2)$ with $|T| < \infty$ such that $G_{u_\e} \weak T$ in $\D_2(\Omega \x \RR^2)$. Due to Proposition~\ref{prop:bd of Gu is mu} we have $\de G_{u_\e}|_{\Omega \x \RR^2} = - \mu_{u_\e} \x \llbracket \SS^1 \rrbracket$. By Remark~\ref{rmk:vortices vanish}, Lemma~\ref{lemma:flat implied D1}, and since $\de G_{u_\e} \weak \de T$ in $\D_1(\Omega \x \RR^2)$, we conclude that $\de T|_{\Omega \x \RR^2} = 0$. The other properties to show that $T \in \cart(\Omega \x \SS^1)$ follow from \cite[Proposition 4.1]{Cic-Orl-Ruf}. Finally, it is easy to see that $u = u_T$ a.e.\ in $\Omega$. 
\end{proof}

\begin{proposition}[Lower bound for the parametric integral] \label{prop:lb for parametric} Assume that $\theta_{\e}\ll \e|\log \e|$ and that $\frac{1}{\e \theta_\e}E_\e(u_\e) \leq C$. Let $G_{u_\e} \in \D_2(\Omega \x \RR^2)$ be the currents associated to $u_\e$ defined as in~\eqref{eq:Gu ac}--\eqref{eq:Gu v} and assume that $G_{u_{\e}}\rightharpoonup T$, where $T \in \cart(\Omega \x \SS^1)$ is a current given by Proposition~\ref{prop:current compactness}, represented as $T = \vec{T} |T|$. Then for every open set $A \compact \Omega$
\begin{equation} \label{eq:parametric lower bound}
\integral{A \x \RR^2}{\Phi( \vec{T} )}{\d |T|} \leq \liminf_{\e \to 0} \integral{A \x \RR^2}{\Phi( \vec{G}_{u_\e} )}{\d |G_{u_\e}|} \, .
\end{equation}
\end{proposition}
\begin{proof}
The statement is a consequence of the lower semicontinuity of parametric integrals with respect to mass bounded weak convergence of currents, \cite[1.3.1, Theorem~1]{Gia-Mod-Sou-II}. 
\end{proof}

We can write explicitly the parametric integral in the left-hand side of~\eqref{eq:parametric lower bound} in terms of the limit $u$ of the sequence $u_\e$. By~\eqref{eq:jc part of T} the jump-concentration part of $T$ is given by
\begin{equation*}
T^{(jc)}(\phi(x,y) \d \widehat x^l \w \d y^m) = (-1)^{2-l} \integral{J_T}{\bigg\{\integral{ \gamma^T_x}{\phi(x,y)}{\d y^m} \bigg\} \nu_{T}^l(x)}{\d \H^1(x)} \, .
\end{equation*}
For $\H^1$-a.e.\ $x \in J_T$ we define the number
\begin{equation}\label{eq:deflT}
\ell_T(x) := \mathrm{length}(\gamma^T_x) = \integral{\supp (\gamma^T_x)}{|\mathfrak{m}(x,y)|}{\d \H^1(y)} \, ,
\end{equation}
where $\mathfrak{m}(x,y)$ is the integer defined in~\eqref{eq:jc multiplicity}. Notice that by $\mathrm{length}(\gamma^T_x)$ we mean the length of the curve $\gamma_x^T$ counted with its multiplicity and not the $\H^1$ Hausdorff measure of its support. Observe that, in particular, $\ell_T(x) = \geo \big( u^-(x), u^+(x) \big)$ if $x \in J_u \sm \L$, whilst $\ell_T(x) = 2 \pi |k(x)|$ if $x \in \L \sm J_u$. The full form of the parametric integral is contained in the lemma below. For a detailed proof see~\cite{Cic-Orl-Ruf}.

\begin{lemma} \label{lemma:parametric in terms of u}
Let $T \in \cart(\Omega \x \SS^1)$ and $u \in BV(\Omega; \SS^1)$ be as in Proposition~\ref{prop:current compactness}, and let~$\Phi$ be the parametric integrand defined in~\eqref{eq:Phi is 2,1}. Then
\begin{equation*}
\begin{split}
\integral{\Omega \x \RR^2}{\Phi( \vec{T} )}{\d |T|}  & = \integral{\Omega}{|\nabla u|_{2,1} }{\d x} + |\DD^{(c)} u|_{2,1}(\Omega) +  \integral{J_T}{\ell_T(x) |\nu_T(x)|_1}{\d \H^1(x)} \, . 
\end{split}
\end{equation*}
\end{lemma}

\begin{remark} \label{rmk:punctured domains}
In presence of vortices, we will work with cartesian currents on punctured open sets. Given a measure $\mu = \sum_{h=1}^N d_h \delta_{x_h}$ and an open set $A$, we adopt the notation 
\begin{equation*}
A_\mu := A \sm \supp(\mu) = A \sm \{x_1, \dots, x_N\}
\end{equation*}
and $A_\mu^\rho := A\sm \bigcup_{h=1}^N B_\rho(x_h)$. We observe that a current $T \in \cart(\Omega_\mu \x \SS^1)$ can be extended to a current $T \in \D_2(\Omega  \x \RR^2)$. Indeed, since $T \in \cart(\Omega_\mu \x \SS^1)$, it can be represented as 
\begin{equation*}
T(\omega) = \integral{\Omega_\mu \x \RR^2}{\langle \omega, \xi \rangle \theta}{\d \H^2 \mres \M}\, , \quad \text{for } \omega \in \D^2(\Omega_\mu \x \RR^2) \, ,
\end{equation*}
according to the notation in~\eqref{eq:im rectifiable}, where $\M \subset \Omega_\mu \x \SS^1$ $\H^2$-a.e. The integral above can be extended to a linear functional on forms $\omega \in \D^2(\Omega \x \RR^2)$, namely
\begin{equation*}
T(\omega) = \integral{\Omega \x \RR^2}{\langle \omega, \xi \rangle \theta}{\d \H^2 \mres \M}\, , \quad \text{for } \omega \in \D^2(\Omega \x \RR^2) \, .
\end{equation*}
To prove the continuity of this extension, fix $\omega \in \D^2(\Omega \x \RR^2)$ with $ \sup_{x}| \omega(x)| \leq 1$. Then
\begin{equation} \label{eq:comparing full with punctured}
|T(\omega)|  \leq |T((1-\zeta) \omega)| + \Big|\integral{\Omega \x \RR^2}{\zeta \langle \omega, \xi \rangle \theta}{\d \H^2 \mres \M} \Big|
  \leq |T|(\Omega_\mu \x \RR^2) + \sum_{h=1}^N \integral{B_\rho(x_h) \x \RR^2}{|\theta|}{\d \H^2 \mres \M}
\end{equation}
where $\zeta \in C^\infty_c(\Omega)$ is such that $0 \leq \zeta \leq 1$, $\supp(\zeta) \subset \bigcup_{h=1}^N B_\rho(x_h)$, and $\zeta \equiv 1$ on $B_{\rho/2}(x_h)$ for every $h=1,\dots,N$. Letting $\rho \to 0$ in the inequality above, we get $|T(\omega)| \leq |T|(\Omega_\mu \x \RR^2)$ since $\H^2\big(\M \cap (\{x_h\} \x \RR^2)\big) \leq  \H^2\big(\{x_h\} \x \SS^1\big) = 0$ for $h=1,\dots,N$ and $\theta$ is $\mathcal{H}^2\mres\M$-summable. This shows that $T \in \D_2(\Omega \x \RR^2)$.

Moreover, since $\omega$ in \eqref{eq:comparing full with punctured} was arbitrary we deduce that $|T|(\Omega \x \RR^2)=|T|(\Omega_\mu \x \RR^2)$ and
\begin{equation*}
\integral{\Omega \x \RR^2}{\!\!\Phi( \vec{T} )}{\d |T|} = \integral{\Omega_\mu \x \RR^2}{\!\! \Phi( \vec{T} )}{\d |T|} = \integral{\Omega}{|\nabla u|_{2,1} }{\d x} + |\DD^{(c)} u|_{2,1}(\Omega) +  \integral{J_T}{\! \ell_T(x) |\nu_T(x)|_1}{\d \H^1(x)} \, .
\end{equation*}
\end{remark}

To state the final lower bound result when $M=0$, for every $u \in BV(\Omega;\SS^1)$ we introduce 
\begin{equation}\label{eq:defJ(u,Om)}
\mathcal{J}(u;\Omega) := \inf \bigg\{ \integral{J_T}{\ell_T(x) |\nu_T(x)|_1}{\d \H^1(x)} \ : \ T \in \cart(\Omega \x \SS^1), \ u_T = u \text{ a.e.\ in } \Omega \bigg\}\,.
\end{equation}

\begin{proposition}[$M=0$, Lower bound]\label{prop:lb}
Assume that $\theta_{\e}\ll \e|\log \e|$ and $\frac{1}{\e \theta_\e} E_\e(u_\e) \leq C$. If $u_{\e}\to u$ in $L^1(\Omega)$, where $u \in BV(\Omega; \SS^1)$ is as in Proposition~\ref{prop:BV compactness}, then 
\begin{equation} \label{eq:final lower bound}
\integral{\Omega}{|\nabla u|_{2,1} }{\d x} + |\DD^{(c)} u|_{2,1}(\Omega) + \mathcal{J}(u;\Omega) \leq \liminf_{\e \to 0} \frac{1}{\e \theta_\e}E_\e(u_\e) \, .
\end{equation}
\end{proposition}
\begin{proof}
Let $\sigma \in (0,1)$ and $A \compact \Omega$ be open.  By Lemma~\ref{lemma:bound eucl from below} and Lemma~\ref{lemma:bound with parametric integral} we deduce that
\begin{equation*}
(1-\sigma) \integral{A \x \RR^2}{\Phi( \vec{G}_{u_\e} )}{\d |G_{u_\e}|}  \leq \frac{1}{\e \theta_\e} E(u_\e) \, .
\end{equation*}
Passing to the limit as $\e \to 0$, Proposition~\ref{prop:lb for parametric} implies
\begin{equation*}
(1-\sigma) \integral{A \x \RR^2}{\Phi( \vec{T} )}{\d |T|} \leq \liminf_{\e \to 0} \frac{1}{\e \theta_\e} E(u_\e) \, .
\end{equation*}
Letting $\sigma \to 1$ and $A \to \Omega$, by Lemma~\ref{lemma:parametric in terms of u} we conclude that~\eqref{eq:final lower bound} holds true.
\end{proof}
\begin{remark} \label{rmk:current greater than bv}
The lower bound~\eqref{eq:final lower bound} dominates the anisotropic total variation, namely 
\begin{equation*}
\integral{\Omega}{|\nabla u|_{2,1} }{\d x} + |\DD^{(c)} u|_{2,1}(\Omega) + \mathcal{J}(u;\Omega)  \geq \integral{\Omega}{|\nabla u|_{2,1} }{\d x} + |\DD^{(c)} u|_{2,1}(\Omega) + \integral{J_u}{\geo(u^-,u^+)|\nu_u|_1}{\d \H^1}
\end{equation*}
for all $u\in BV(\Omega;\SS^1)$. This can be seen using the definition of $\ell_T(x)$ for a given $T \in \cart(\Omega \x \SS^1)$ with $u_T = u$. Indeed, for $\H^1$-a.e.\ $x \in J_u \cap \L$ we have $\geo(u^-(x),u^+(x)) \leq \mathrm{length}(\gamma^T_x) = \ell_T(x)$, since $\gamma^T_x$ is a curve connecting $u^-(x)$ and $u^+(x)$ in $\SS^1$.
\end{remark}
\subsection{Compactness and lower bound in presence of vortices}
Next we extend the results of the previous subsection to the case of $M$ vortices. Again we consider a general sequence $u_{\e}\colon \e\ZZ^2\to\S_{\e}$ and the associated current $G_{u_{\e}}$.

\begin{proposition}[$M$ vortices, Compactness] \label{prop:compactness M vortices}
Assume that $\e\ll\theta_{\e}\ll\e|\log \e|$ and that there exist $M \in \NN$ and $C > 0$ such that 
\begin{equation} \label{eq:first order bound}
\frac{1}{\e \theta_\e} E_\e(u_\e) - 2\pi M |\log \e | \frac{\e}{\theta_\e} \leq C \, .
\end{equation}
Then there exists $\mu = \sum_{h=1}^N d_h \delta_{x_h}$ with $d_h\in\ZZ$ such that $\mu_{u_\e} \flat \mu$ (up to a subsequence) and $|\mu|(\Omega) \leq M$. If, in addition, $|\mu|(\Omega) = M$, then there exist $u \in BV(\Omega;\SS^1)$ and $T \in \D_2(\Omega \x \RR^2)$ such that 
\begin{itemize} 
	\item[(i)]  $u_{\e}\to u$ in $L^1(\Omega;\RR^2)$ and $u_\e \wstar u$ weakly* in $BV_{\mathrm{loc}}(\Omega_\mu;\RR^2)$;
	\item[(ii)] $T \in \cart(\Omega_\mu \x \SS^1)$ and $u_T = u$ a.e.\ in $\Omega$;
	\item[(iii)] $G_{u_\e} \weak T$ in $\D_2(\Omega_\mu \x \RR^2)$ (up to a subsequence);
	\item[(iv)] $\de T|_{\Omega \x \RR^2} = - \mu \x \llbracket \SS^1 \rrbracket$.
\end{itemize} 
\end{proposition}

\begin{proof}
From~\eqref{eq:first order bound} it follows that $\frac{1}{\e^2 |\log \e|} E_\e(u_\e) \leq  2\pi M +  C \frac{\theta_\e}{\e |\log \e | }$, so that by Proposition~\ref{prop:XY classical} we get that (up to a subsequence) $\mu_{u_\e} \flat \mu = \sum_{h=1}^N d_h \delta_{x_h}$ and $|\mu|(\Omega) \leq M$. From now on we assume that $|\mu|(\Omega) = M$, that is $\sum_{h=1}^N |d_h| = M$.
	
Let $\rho > 0$ small enough such that the balls $B_\rho(x_h)$ are pairwise disjoint and $\overline{B_{\rho}(x_h)}\subset\Omega$. Recall the localized lower bound for the $XY$-model \cite[Theorem 3.1]{Ali-DL-Gar-Pon}, which states that
\begin{equation}\label{eq:localXYlb}
\liminf_{\e \to 0} \bigg[ \frac{1}{\e^2} E_\e(u_\e; B_\rho(x_h)) - 2 \pi |d_h| \log \frac{\rho}{\e} \bigg] \geq \tilde C \quad \text{ for some } \tilde C \in \RR \, .
\end{equation}
From this inequality and the fact that $\e\ll\theta_{\e}$ we deduce that 
\begin{equation}\label{eq:localizedsign}
\begin{split}
& \liminf_{\e \to 0} \bigg[ \frac{1}{\e \theta_\e} E_\e(u_\e; B_\rho(x_h)) - 2\pi |d_h| |\log \e | \frac{\e}{\theta_\e} \bigg] \\
& \quad = \liminf_{\e \to 0} \bigg[ \frac{1}{\e \theta_\e} E_\e(u_\e; B_\rho(x_h)) - 2\pi |d_h| |\log \e | \frac{\e}{\theta_\e} - 2 \pi|d_h| \log \rho  \frac{\e}{\theta_\e} \bigg] \\
& \quad = \liminf_{\e \to 0} \frac{\e}{\theta_\e}\bigg[ \frac{1}{\e^2} E_\e(u_\e; B_\rho(x_h)) - 2\pi |d_h| \log \frac{\rho}{\e} \bigg] \geq 0 \, .
\end{split}
\end{equation}
Summing over $h=1,\dots, N$, the superadditivity of the $\liminf$ yields
\begin{equation} \label{eq:lb close to vortex}
\liminf_{\e \to 0} \bigg[ \sum_{h=1}^N  \frac{1}{\e \theta_\e} E_\e(u_\e; B_\rho(x_h)) - 2\pi M |\log \e | \frac{\e}{\theta_\e}  \bigg]  \geq 0 \, .
\end{equation}
Therefore the bound~\eqref{eq:first order bound} implies
\begin{equation*}
\limsup_{\e \to 0} \frac{1}{\e \theta_\e} E_\e(u_\e; \Omega^\rho_\mu)  \leq  C - \liminf_{\e \to 0} \bigg[ \sum_{h=1}^N  \frac{1}{\e \theta_\e} E_\e(u_\e; B_\rho(x_h)) - 2\pi M |\log \e | \frac{\e}{\theta_\e}  \bigg]  \leq C \, ,
\end{equation*}
so that, for $\e$ small enough, $\frac{1}{\e \theta_\e} E_\e(u_\e; \Omega^\rho_\mu)  \leq  2 C$, where $C$ is independent of $\rho$. By Proposition~\ref{prop:BV compactness} and Proposition~\ref{prop:current compactness}, with a diagonal argument we obtain that there exist $u \in BV(\Omega;\SS^1)$ and $T \in \cart(\Omega_\mu \x \SS^1)$ such that $u_\e \wstar u$ weakly* in $BV_{\mathrm{loc}}(\Omega_\mu;\SS^1)$, $G_{u_\e} \weak T$ in $\D_2(\Omega_\mu \x \RR^2)$ up to a subsequence, and $u_T = u$ a.e.\ in $\Omega$. Since $u_{\e}$ is equiintegrable, the local weak* $BV$-convergence implies strong $L^1(\Omega)$-convergence. Thus (i)--(iii) hold true.

By Remark~\ref{rmk:punctured domains}, the current $T$ can be extended to a current $T \in \D_2(\Omega  \x \RR^2)$. Thus, it only remains to prove (iv). The argument is local and we can work close to a single atom~$x_h$ of~$\mu$. Without loss of generality assume that $x_h = 0$ and $\Omega = B:=B_1(0)$. First of all let us note that $\supp(\de T) \subset \{0\} \x \SS^1$. Indeed, on the one hand if $\omega \in \D^1(B \x \RR^2)$ is such that $\supp(\omega) \subset (B \x \RR^2) \sm (\{0\} \x \RR^2)$, then $\de T(\omega) = 0$, since $T \in \cart\big((B \sm \{0\}) \x \SS^1\big)$; on the other hand, if $\omega \in \D^1(B \x \RR^2)$ is such that $\supp(\omega) \subset (B \x \RR^2) \sm (B \x \SS^1)$, then $\supp(\! \d  \omega) \subset (B \x \RR^2) \sm (B \x \SS^1)$ and thus $\de T(\omega) = T(\! \d \omega) = 0$, since $\supp(T) \subset B \x \SS^1$. In conclusion $\supp(\de T) \subset (\{0\} \x \RR^2) \cap (B \x \SS^1) = \{0\} \x \SS^1$. Being $\de T$ a boundaryless 1-current with support in a 1-dimensional manifold, the Constancy Theorem~\cite[5.3.1, Theorem~2]{Gia-Mod-Sou-I} gives that $\de T|_{B \x \RR^2} = - c \, \delta_0 \x \llbracket \SS^1 \rrbracket$ for some $c \in \RR$. Now fix a function $\zeta \in C^\infty_c(B)$ with $\zeta \equiv 1$ in the ball $B_{1/2}(0)$ and define the 1-form $\omega = \zeta \omega_{\SS^1}$, $\omega_{\SS^1}$ being the 0-homogeneous extension of the volume form of $\SS^1$ to $\RR^2\sm \{0\}$. Since $ \d \omega \in \D^2\big((B\sm \{0\}) \x \RR^2\big)$, the convergence in (ii), Proposition~\ref{prop:bd of Gu is mu}, and the flat convergence $\mu_{u_\e} \flat \mu$ yield the claimed equality $c = \mu(\{0\})$. Indeed,
\begin{equation*}
- c \, 2 \pi = \de T (\omega) = T(\! \d \omega) = \lim_{\e \to 0} G_{u_\e} (\! \d \omega) = \lim_{\e \to 0} \de G_{u_\e} (\omega) = \lim_{\e \to 0} - \langle \mu_{u_\e}, \zeta \rangle 2 \pi = - \langle \mu, \zeta \rangle 2 \pi.
\end{equation*}
\end{proof}


We now prove the lower bound for $M$ vortices. Let us define the set of admissible currents
\begin{equation} \label{eq:def of Adm} 
	\begin{split}
		\mathrm{Adm}(\mu,u;\Omega)  :=  \bigg\{  T \in \D_2(\Omega \x \RR^2) \colon   T \in \cart(\Omega_\mu \x \SS^1)\, ,  \de T|_{\Omega \x \RR^2} = - \mu \x \llbracket \SS^1 \rrbracket \, , \ u_T = u \text{ a.e.}   \bigg\}
	\end{split}
\end{equation}
and, similarly to \eqref{eq:defJ(u,Om)}, the energy
\begin{equation} \label{eq:def of surface L}
\mathcal{J}(\mu,u;\Omega) := \inf \bigg\{ \integral{J_T}{\ell_T(x) |\nu_T(x)|_1}{\d \H^1(x)} \ : \ T \in \mathrm{Adm}(\mu,u;\Omega)  \bigg\}
\end{equation}
for every $\mu = \sum_{h=1}^N d_h \delta_{x_h}$ and $u \in BV(\Omega;\SS^1)$ with $\ell_T(x)$ defined in \eqref{eq:deflT}\footnote{$\mathrm{Adm}(\mu,u;\Omega)$ is non-empty. Indeed, by Proposition~\ref{prop:supporting BV} there exists $T \in \cart(\Omega \x \SS^1)$ such that $u_T = u$. Let $\gamma_1 , \dots, \gamma_N$ be pairwise disjoint unit speed Lipschitz curves  such that $\gamma_h$ connects $x_h$ to $\de \Omega$. Define $L_h$ to be the 1-current $\tau(\supp(\gamma_h), -d_h, \dot \gamma_h)$, so that $\de L_h = d_h \delta_{x_h}$. Then $T + \sum_{h=1}^N L_h \x \llbracket \SS^1 \rrbracket \in \mathrm{Adm}(\mu, u;\Omega)$.}.

\begin{proposition}[$M$ vortices, Lower bound] \label{prop:lb for M vortices}
Assume that $\theta_{\e}\ll\e|\log \e|$ and ~\eqref{eq:first order bound} holds. Assume further that $\mu_{u_{\e}}\flat\mu= \sum_{h=1}^N d_h \delta_{x_h}$ with $|\mu|(\Omega)=M$, $u_{\e}\to u$ in $L^1(\Omega;\RR^2)$ with $u\in BV( \Omega;\SS^1)$ as in Proposition \ref{prop:compactness M vortices}. Then
\begin{equation} \label{eq:final lb M vortices}
\integral{\Omega}{|\nabla u|_{2,1} }{\d x} + |\DD^{(c)} u|_{2,1}(\Omega) + \mathcal{J}(\mu, u;\Omega) \leq \liminf_{\e \to 0} \bigg[ \frac{1}{\e \theta_\e} E_\e(u_\e) - 2 \pi M |\log \e| \frac{\e}{\theta_\e} \bigg] \, .
\end{equation}
\end{proposition}
\begin{proof}
Let us fix $A \compact \Omega_\mu$ and $\sigma \in (0,1)$. Then there exists $\rho > 0$ such that $A \compact \Omega_\mu^\rho$. Thanks to~\eqref{eq:first lower bound} and Lemma~\ref{lemma:bound with parametric integral}, for $\e$ small enough we infer that
	\begin{equation*}
	(1-\sigma) \integral{A\x \RR^2}{\Phi(\vec{G}_{u_\e})}{\d |G_{u_\e}|} \leq \frac{1}{\e \theta_\e} E_\e(u_\e; \Omega^\rho_\mu) \,.
	\end{equation*}
	Passing to a subsequence, we have that $G_{u_{\e}}\rightharpoonup T$ in $\mathcal{D}_2(\Omega_{\mu}\times\RR^2)$ for some $T\in\mathcal{D}_2(\Omega\times\RR^2)$ given by Proposition \ref{prop:compactness M vortices}. As in Proposition~\ref{prop:lb for parametric} and from the bound~\eqref{eq:lb close to vortex} we infer that
	\begin{equation*}
	\begin{split}
	(1-\sigma)\integral{A \x \RR^2}{\Phi(\vec{T})}{\d |T|} &\leq \liminf_{\e \to 0} \ (1-\sigma)   \integral{A\x \RR^2}{\Phi(\vec{G}_{u_\e})}{\d |G_{u_\e}|} \\
	& \leq \liminf_{\e \to 0} \frac{1}{\e \theta_\e} E_\e(u_\e; \Omega^\rho_\mu) \leq \liminf_{\e \to 0} \bigg[ \frac{1}{\e \theta_\e} E_\e(u_\e) - 2 \pi M |\log \e| \frac{\e}{\theta_\e} \bigg] \, .
	\end{split}
	\end{equation*}
	Letting $A \to \Omega_\mu$ and $\sigma \to 0$ we conclude 
	that
	\begin{equation*}
	\integral{\Omega_{\mu}}{\Phi(\vec{T})}{\d |T|} \leq \liminf_{\e \to 0} \bigg[ \frac{1}{\e \theta_\e} E_\e(u_\e) - 2 \pi M |\log \e| \frac{\e}{\theta_\e} \bigg]\,.
	\end{equation*}
	By Proposition~\ref{prop:compactness M vortices} (ii) \& (iv) we have $T\in \mathrm{Adm}(\mu,u;\Omega)$, so that \eqref{eq:final lb M vortices} is a direct consequence of Lemma~\ref{lemma:parametric in terms of u} and Remark~\ref{rmk:punctured domains}.
\end{proof}

\subsection{Upper bound in absence of vortices} \label{sec:upper bound no vortices}

To reduce notation, for $u\in BV(\Omega;\SS^1)$ we~set
\begin{equation*}
\mathcal{E}(u):=\integral{\Omega}{|\nabla u|_{2,1} }{\d x} + |\DD^{(c)} u|_{2,1}(\Omega) + \mathcal{J}(u;\Omega)
\end{equation*}
with $\mathcal{J}(u;\Omega)$ given by \eqref{eq:defJ(u,Om)}. The proof of the $\Gamma$-limsup inequality is done in several steps which gradually simplify the map $u \in BV(\Omega;\SS^1)$ that we want to approximate. 

In the next proposition we approximate the map $u$ with a sequence of smooth maps. 

\begin{proposition} \label{prop:density smooth}
Let $u \in BV(\Omega;\SS^1)$. Then there exist an open set $\tilde \Omega \Supset \Omega$ and a sequence $u_h \in C^\infty(\tilde \Omega; \SS^1) \cap W^{1,1}(\tilde \Omega; \SS^1)$  such that $u_h \to u$ strongly in $L^1(\Omega;\RR^2)$ and  
\begin{equation*}
\limsup_{h \to +\infty} \integral{\Omega}{|\nabla u_h|_{2,1} }{\d x} \leq \integral{\Omega}{|\nabla u|_{2,1} }{\d x} + |\DD^{(c)} u|_{2,1}(\Omega) + \mathcal{J}(u;\Omega) \, .
\end{equation*}
\end{proposition}
\begin{proof}
Let $\eta > 0$ and let $T \in \cart(\Omega \x \SS^1)$ with $u_T = u$ a.e.\ in $\Omega$ be such that 
\begin{equation} \label{eq:T almost u}
\integral{J_T}{\ell_T(x) |\nu_T(x)|_1}{\d \H^1(x)} \leq \mathcal{J}(u;\Omega)  + \eta  \, .
\end{equation}
Note that by Lemma~\ref{lemma:parametric in terms of u}
\begin{equation*}
\integral{\Omega}{|\nabla u|_{2,1} }{\d x} + |\DD^{(c)} u|_{2,1}(\Omega)  + \integral{J_T}{\ell_T(x) |\nu_T(x)|_1}{\d \H^1(x)} = \integral{\Omega \x \RR^2}{\Phi(\vec{T})}{\d |T|} \, ,
\end{equation*}
where $\Phi$ is the parametric integrand defined in~\eqref{eq:Phi is 2,1}. By Lemma~\ref{lemma:extension of currents} we can extend the current $T$ to $\tilde \Omega \x \SS^1$ for some $\tilde \Omega \Supset \Omega$ such that $T \in \cart(\tilde \Omega \x \SS^1)$ and $| T|(\de \Omega \x \RR^2) = 0$.

Thanks to the Approximation Theorem~\ref{thm:approximation} we find a sequence $u_h \in C^\infty(\tilde \Omega; \SS^1)$ such that $G_{u_h} \weak T$ in $\D_2(\tilde \Omega \x \RR^2)$ and $|G_{u_h}|(\tilde \Omega \x \RR^2) \to |T|(\tilde \Omega \x \RR^2)$. In particular, since $|T|$ does not charge $\partial\Omega\times\RR^2$, we have $|G_{u_h}|(\Omega \x \RR^2) \to |T|(\Omega \x \RR^2)$ and the convergence $u_h \to u_T = u$ in~$L^1(\Omega;\RR^2)$. Therefore, by Reshetnyak's Continuity Theorem \cite[Theorem 2.39]{Amb-Fus-Pal} we have 
\begin{equation*}
\integral{\Omega \x \RR^2}{\Phi(\vec{G}_{u_h})}{\d |G_{u_h}|} \to \integral{\Omega \x \RR^2}{\Phi(\vec{T})}{\d |T|} \, .
\end{equation*}
By~\eqref{eq:smooth components} and the Area Formula we can write 
\begin{equation*}
\integral{\Omega \x \RR^2}{\Phi(\vec{G}_{u_h})}{\d |G_{u_h}|} = \integral{\Omega}{|\nabla u_h|_{2,1}(x)}{\d x} \, .
\end{equation*}
This implies that
\begin{equation*}
\limsup_{h \to +\infty} \integral{\Omega}{|\nabla u_h|_{2,1}(x)}{\d x} \leq \integral{\Omega}{|\nabla u|_{2,1} }{\d x} + |\DD^{(c)} u|_{2,1}(\Omega) + \mathcal{J}(u;\Omega) + \eta \, .
\end{equation*}
Since $\eta > 0$ was arbitrary, we conclude the proof.
\end{proof}

The next lemma states that we can discretize on a lattice $\lambda_n \ZZ^2$ any smooth map with values in $\SS^1$ in such a way that the anisotropic $BV$ norm does not increase. The discretized maps $u_n$ satisfy in addition an 'almost continuity property', cf.\ \eqref{eq:almostcontinuity}, which states that for $\lambda_n$ small enough the constant values of $u_n$ in two neighboring cubes are close.  
\begin{lemma}[Discretization of smooth $\SS^1$-valued maps] \label{lemma:discretization of smooth wout sing}
Let $\lambda_n := 2^{-n}$, $n \in \NN$ and let~$O$, $\tilde O$ be bounded, open sets such that $ O\compact \tilde O$. Assume that $u \in C^\infty(\tilde O; \SS^1) \cap W^{1,1}(\tilde O; \SS^1)$.
Then there exist a sequence of piecewise constant maps $u_n \in \PC_{\lambda_n}(\SS^1)$ such that $u_n \to u$ strongly in $L^1(O;\RR^2)$ as $n \to +\infty$ and
\begin{equation} \label{eq:limsup with pc}
\limsup_{n\to +\infty}\integral{J_{u_n} \cap O^{\lambda_n} }{\geo(u_n^+,u_n^-)|\nu_{u_n}|_1}{\d \H^1} \leq \integral{O}{|\nabla u|_{2,1}}{\d x} \, ,
\end{equation} 
where $O^{\lambda_n}$ is the union of half-open squares given by
\begin{equation*}
O^{\lambda_n} := \bigcup \{ I_{\lambda_n}(\lambda_n z) \colon z \in \ZZ^2 \text{ such that } I_{\lambda_n}(\lambda_n z) \cap   O \neq \emptyset\} \, .
\end{equation*}
Moreover , for every $\delta > 0$ there exists $\ol n = \ol n(u,\delta,\tilde O)$ such that for every $n \geq \ol n$ and for every~$z_1, z_2 \in \ZZ^2$ with $\ol I_{\lambda_n}(\lambda_n z_1) \cap \ol I_{\lambda_n}(\lambda_n z_2) \neq \emptyset$ and $I_{\lambda_n}(\lambda_n z_i) \cap  O \neq \emptyset$ we have 
\begin{equation} \label{eq:almostcontinuity}
\geo\big(u_n(\lambda_n(z_1)),u_n(\lambda_n(z_2)) \big) \leq \delta \, .
\end{equation}
\end{lemma}
\begin{proof}
	Let $O'$ be an open set such that $O \compact O' \compact \tilde O$ and let $\ol n$ be so large that for every~$n \geq \ol n$ we have $O^{ \lambda_n} \subset O'$.
	For every $z \in \ZZ^2$ such that $I_{\lambda_n}(\lambda_n z) \cap O  \neq \emptyset$ we define
	\begin{equation*}
		u_n(\lambda_n z) := u\big(\lambda_n (z + \tfrac{1}{2}e_1 + \tfrac{1}{2}e_2 ) \big) \, ,	
	\end{equation*}
	$\lambda_n (z + \tfrac{1}{2}e_1 + \tfrac{1}{2}e_2 )$ being the center of the square $I_{\lambda_n}(\lambda_n z)$. The definition is well-posed, since $\lambda_n (z + \tfrac{1}{2}e_1 + \tfrac{1}{2}e_2 ) \in \tilde O$. Then we extend $u_n$ to~$\lambda_n \ZZ^2$ by choosing an arbitrary value in $\SS^1$. This defines a piecewise constant map $u_n \in \PC_{\lambda_n}(\SS^1)$. 

	Since $u$ is continuous on $O'$, it follows that $u_n \to u$ pointwise on $O$ and thus also strongly in $L^1(O;\RR^2)$ by dominated convergence. 
	%
	Next we show~\eqref{eq:limsup with pc}. For $i \in \{1,2\}$ define 
	\begin{equation*}
		\mathcal{Z}_i(\lambda_n) := \{ z \in \ZZ^2 \colon I_{\lambda_n}(\lambda_n z) \cap O \neq \emptyset \text{ and } I_{\lambda_n}(\lambda_n (z+e_i)) \cap O \neq \emptyset \} \,. 
	\end{equation*} 
	Let $z \in \mathcal{Z}_i(\lambda_n)$. Since $u$ is $C^\infty$ in the interior of the rectangle $I_{\lambda_n}(\lambda_n z) \cup I_{\lambda_n}(\lambda_n (z+e_i))$, it admits a $C^\infty$ lifting $\varphi$ such that $u = \exp(\iota \varphi)$ in the interior of $I_{\lambda_n}(\lambda_n z) \cup I_{\lambda_n}(\lambda_n (z+e_i))$. Then, by the fundamental theorem of calculus and the definition of $u_n$,
	\begin{equation} \label{eq:2311181951}
		\begin{split}
			\geo\big(u_n(\lambda_n(z+e_i)), u_n(\lambda_n z)\big) & \leq \Big|\varphi\big(\lambda_n(z+e_i+\tfrac{1}{2}e_1 + \tfrac{1}{2}e_2 )\big)  -  \varphi\big(\lambda_n(z+\tfrac{1}{2}e_1 + \tfrac{1}{2}e_2 )\big) \Big| \\
			& \leq \lambda_n \int_0^1 \big| \de_i \varphi\big(\lambda_n(z+te_i+\tfrac{1}{2}e_1 + \tfrac{1}{2}e_2 )\big) \big| \d t \\
			& = \lambda_n \int_0^1 \big| \de_i u\big(\lambda_n(z+te_i+\tfrac{1}{2}e_1 + \tfrac{1}{2}e_2 )\big) \big| \d t \, .
		\end{split}
	\end{equation}
We notice, in addition, that for every $t \in [0,1]$ and $z\in\mathcal{Z}_i(\lambda_n)$
\begin{equation*}
\begin{split}
& \bigg| \integral{I_{\lambda_n}(\lambda_n z)}{\hspace{-1em}\big|\de_i u(x) \big|}{\d x}  - \lambda_n^2  \big| \de_i u\big(\lambda_n(z+te_i+\tfrac{1}{2}e_1 + \tfrac{1}{2}e_2 )\big) \big| \ \bigg| \\
& \quad \leq \integral{I_{\lambda_n}(\lambda_n z)}{\hspace{-1em} \big|\de_i u(x) -  \de_i u\big(\lambda_n(z+te_i+\tfrac{1}{2}e_1 + \tfrac{1}{2}e_2 )\big) \big|}{\d x} \leq 2 \lambda_n^3  \|\nabla^2 u\|_{L^\infty(O')  }  =: C(u) \lambda_n^3 \, .
\end{split}
\end{equation*}
From~\eqref{eq:2311181951} and the previous estimate it follows that 
	\begin{equation*}
		\begin{split}
			\integral{J_{u_n} \cap O^{\lambda_n} }{  \geo(u_n^+,u_n^-)|\nu_{u_n}|_1}{\d \H^1} & \leq \sum_{i=1}^2  \sum_{z \in \mathcal{Z}_i(\lambda_n)} \hspace{-1em} \lambda_n \geo\big(u_n(\lambda_n(z+e_i)), u_n(\lambda_n z)\big) \\
			& \leq \sum_{i=1}^2 \int_0^1 \hspace{-0.5em} \sum_{z \in \mathcal{Z}_i(\lambda_n)}  \hspace{-0.5em} \lambda_n^2  \big| \de_i u\big(\lambda_n(z+te_i+\tfrac{1}{2}e_1 + \tfrac{1}{2}e_2 )\big) \big| \d t \\
			& \leq \sum_{i=1}^2 \int_0^1 \hspace{-0.5em} \sum_{z \in \mathcal{Z}_i(\lambda_n)} \bigg\{  \integral{I_{\lambda_n}(\lambda_n z)}{\hspace{-1em}\big|\de_i u(x) \big|}{\d x}  + C(u) \lambda_n^3 \bigg\}\d t  \\
			& \leq   \integral{O'}{  |\nabla u|_{2,1}}{\d x} + C(u) |O'| \lambda_n \, .
		\end{split}
	\end{equation*}
	We conclude the proof of~\eqref{eq:limsup with pc} letting $n \to +\infty$ and then $O' \searrow O$.

	Finally, in order to prove~\eqref{eq:almostcontinuity}, observe that the condition $\ol I_{\lambda_n}(\lambda_n z_1) \cap \ol I_{\lambda_n}(\lambda_n z_2) \neq \emptyset$ implies that $|\lambda_n(z_1+\tfrac{1}{2}e_1+e_2)-\lambda_n(z_2+\tfrac{1}{2}e_1+\tfrac{1}{2}e_2)|\leq \sqrt{2}\lambda_n$, so the claim follows from the Lipschitz continuity of $u$ on the larger set $\tilde{O}$. 
	\end{proof}

Now we can construct a recovery sequence $u_\e \colon \e \ZZ^2 \to \S_\e$. Due to the previous simplifications, it suffices to approximate the energy of piecewise constant maps on the lattice~$\lambda_n \ZZ^2$ which come from Lemma~\ref{lemma:discretization of smooth wout sing}. To define the recovery sequence we shall construct a minimal transition (in $\SS^1$) between two constant values of $\SS^1$. First we introduce some notation about geodesics in $\SS^1$ and state an elementary stability property, whose proof we omit.
%
\begin{definition}\label{d.geodesic}
For $u^1,u^2\in\SS^1$ denote by $\Geo[u^1,u^2] \colon [0,\geo(u^1,u^2)]\to\SS^1$ the (in case of non-uniqueness counterclockwise rotating) unit speed geodesic between $u^1$ and $u^2$ which we extend by $\Geo[u^1,u^2](t)=u^1$ for $t<0$ and $\Geo[u^1,u^2](t)=u^2$ for $t>\geo(u^1,u^2)$. As such the geodesics are $1$-Lipschitz continuous functions on $\RR$. We further set ${\rm mid}(u^1,u^2)=\Geo(\tfrac{1}{2}\geo(u^1,u^2))$ as the midpoint on that geodesic.	
\end{definition}
\begin{lemma}\label{l.ongeodesics}
There exists a constant $c>0$ such that whenever $u^1,u^2,b\in\SS^1$ are such that $u^1,u^2\in B_{c}(b)$, then for all $t\in\RR$
\begin{equation*}
|\Geo[u^1,b](t)-\Geo[u^2,b](t)|\leq \geo(u^1,u^2) \, .
\end{equation*}  	
\end{lemma}

We introduce a map which will be used to project vectors of $\SS^1$ on $\S_\e$. Given $u\in\SS^1$ we let $\varphi_u\in [0,2\pi)$ be the unique angle such that $u=\exp(\iota\varphi_u)$. We define $\mathfrak{P}_{\e}\colon \SS^1\to\S_\e$ by 
\begin{equation}\label{eq:defproj}
\mathfrak{P}_{\e}(u)=\exp\left(\iota \theta_{\e}\left\lfloor\tfrac{\varphi_u}{\theta_{\e}}\right\rfloor\right).
\end{equation}
Combined with Propositions \ref{prop:BV compactness} and \ref{prop:lb} the next result completes the proof of Theorem~\ref{thm:e smaller theta with vortices} when $M=0$.

\begin{proposition}[$M=0$, Upper bound] \label{prop:construction of ueps}
Assume $\e\ll\theta_{\e}\ll 1$. Let $u \in BV(\Omega;\SS^1)$. Then there exist $u_\e \in \PC_\e(\S_\e)$ such that $u_\e \to u$ strongly in $L^1(\Omega;\RR^2)$ and
\begin{equation*}
\limsup_{\e \to 0} \frac{1}{\e \theta_\e} E_\e(u_\e) \leq \integral{\Omega}{|\nabla u|_{2,1} }{\d x} + |\DD^{(c)} u|_{2,1}(\Omega) + \mathcal{J}(u;\Omega) \, . 
\end{equation*}
\end{proposition}
\begin{proof}
By Proposition~\ref{prop:density smooth}, Lemma~\ref{lemma:discretization of smooth wout sing} and by the $L^1$-lower semicontinuity of the the $\Gamma$-limsup, it is enough to prove that for $u_n \in \PC_{\lambda_n}(\SS^1)$ 
\begin{equation}\label{eq:toshow}
\Gamma\text{-}\limsup_{\e \to 0}\frac{1}{\e \theta_\e} E_\e(u_n) \leq 	\integral{J_{u_n} \cap \Omega^{\lambda_n} }{\geo(u_n^+,u_n^-)|\nu_{u_n}|_1}{\d \H^1}\, .
\end{equation}
Since $u_n$ is fixed in the following discussion, to simplify the notation we denote $u_n$ by $u$ and $\lambda_n$ by $\lambda$, always assuming that $\lambda\ll 1$.
	
We will define a recovery sequence locally on each half-open cube $I_{\lambda}(\lambda z)$ for $z\in\mathbb{Z}^2$. First, we define a boundary condition on $\partial I_{\lambda}(\lambda z)$. For a side $S=\{\lambda z^{\prime}+te_i:\,t\in[0,\lambda]\}$ with $z^{\prime}\in\mathbb{Z}^2$ and $i\in\{1,2\}$, and three values $v=(v^1,v^2,v^3)\in(\SS^1)^3$, we set $b^{\e}_S[v]\colon S\to\SS^1$~as
	\begin{equation}\label{eq:defbcases}
	b^{\e}_S[v](\lambda z^{\prime}+t e_i)=
	\begin{cases}
	v^1 &\mbox{if $t\in c_0 \frac{\e}{\theta_{\e}}[0, 1)$} \, ,
	\\
	\Geo[v^1,v^2]\left(\frac{\geo(v^1,v^2)\theta_{\e}}{c_0\e}(t-c_0\frac{\e}{\theta_{\e}})\right) &\mbox{if $t\in c_0 \frac{\e}{\theta_{\e}}[1,2)$} \, ,
	\\
	v^2 &\mbox{if $t\in[2c_0\frac{\e}{\theta_{\e}},\lambda-2 c_0\frac{\e}{\theta_{\e}})$} \, ,
	\\
	\Geo[v^2,v^3]\left(\frac{\geo(v^2,v^3)\theta_{\e}}{c_0\e}(t-(\lambda- 2c_0\frac{\e}{\theta_{\e}})\right) &\mbox{if $t\in\lambda-c_0\frac{\e}{\theta_{\e}}(1,2]$} \, , 
	\\
	v^3 &\mbox{if $t\in\lambda-c_0\frac{\e}{\theta_{\e}}[0,1]$} \, .
	\end{cases}
	\end{equation}
	The particular choice of the constant $c_0$ is not important. For this proof it suffices that $c_0 > 2 \pi$. This condition will be clear only after~\eqref{eq:jlocate0}. (However, to apply this construction also in the proof of Proposition~\ref{p.smoothapprox} we need to choose a larger constant, namely $c_0 = 393$.) Since $\frac{\e}{\theta_{\e}}\to 0$ by assumption, the function $b_S[v]$ can be interpreted as follows: in a small neighborhood of the two endpoints of $S$ we set the two values $v^1$ and $v^3$, while in a contiguous small neighborhood we use the geodesic for a transition to the value $v^2$, which is taken on most of the side. 
	
	Next, given $u\in\mathcal{PC}_{\lambda}(\SS^1)$ and a side $S$ as above we specify the values
	\begin{equation}\label{eq:defvSu}
	v_S(u)=(u(\lambda z^{\prime}),{\rm mid}(u^-_{S},u^+_{S}),u(\lambda(z^{\prime}+e_i))) \, ,
	\end{equation}
	where $u^-_{S}$ and $u^+_{S}$ denote the (constant) traces of $u$ along the side $S$ and the midpoint is given by Definition \ref{d.geodesic}. The boundary values $b_{z,\e}\colon\partial I_{\lambda}(\lambda z)\to\SS^1$ are then defined by
	\begin{align*}
	b_{z,\e}(x)=
	b_{S}^{\e}[v_S(u)](x) \quad\mbox{if $x=\lambda z^{\prime}+t e_i\in S$ for some $z^{\prime}\in\mathbb{Z}^2$ and $t\in[0,\lambda]$}  \, .
	\end{align*} 
	Note that this function is well-defined also in the corners with $b_{z,\e}(\lambda z_0)=u(\lambda z_0)$ for all $z_0\in\mathbb{Z}^2$. Moreover, since we have chosen unit speed geodesics and $c_0>2\pi$, on each side $S$ the function $b_S^{\e}[v_S(u)]$ satisfies a Lipschitz-estimate of the form
	\begin{equation}\label{eq:bLip}
	|b_S^{\e}[v_S(u)](x)-b_S^{\e}[v_S(u)](y)|\leq  \frac{1}{2} \frac{\theta_{\e}}{\e}|x-y| \, \quad x,y \in S \, .  
	\end{equation}
	Repeating the construction on every half-open cube we obtain a continuous function on the skeleton $\bigcup_{z \in \ZZ^2} \de I_{\lambda}(\lambda z)$. 
	
	We are now in a position to define the recovery sequence of $u$. We will interpolate between the constant  $u(\lambda z)$ and the boundary value $b_{z,\e}$ in $I_{\lambda}(\lambda z)$. This will be done on a mesoscale towards the boundary $\partial I_{\lambda}(\lambda z)$. Let $P \colon I_{\lambda}(\lambda z)\to\partial I_{\lambda}(\lambda z)$ be a function satisfying $|P(x)-x|=\dist(x,\partial I_{\lambda}(\lambda z))$ for all $x\in I_{\lambda}(\lambda z)$ (such a function can be defined globally by periodicity). To reduce notation, let $u_z=u(\lambda z)$. Set $\bar u_{\e}  \colon \e\ZZ^2\cap I_{\lambda}(\lambda z)\to\SS^1$ as 
	\begin{equation} \label{eq:transition}
	 \bar u_{\e} (\e i)=
	\Geo\left[b_{z,\e}(P(\e i)),u_z\right]\left(\theta_{\e}\e^{-1}\dist\big(\e i,\partial I_{\lambda}(\lambda z)\big)\right),
	\end{equation}
	with the extended geodesics given by Definition~\ref{d.geodesic}. Note that in general $\bar{u}_{\e}(\e i)\notin\S_\e$. Hence we define $u_{\e}\in\mathcal{PC}_{\e}(\S_\e)$ by $u_{\e}:=\mathfrak{P}_{\e}(\bar{u}_{\e})$ with $\mathfrak{P}_{\e}$ given by \eqref{eq:defproj}. We claim that $u_{\e}$ converges to $u$ in $L^1(\Omega;\RR^2)$. Indeed, for all $\e i\in I_{\lambda}(\lambda z)$ we have by Definition~\ref{d.geodesic}
	\begin{equation}\label{eq:interiorpoints}
	u_{\e}(\e i)=\mathfrak{P}_{\e}(u_z)\quad\quad\text{if }\dist(\e i,\partial I_{\lambda}(\lambda z))\geq \pi\frac{\e}{\theta_{\e}} \,,
	\end{equation}
	so that the assumptions $\frac{\e}{\theta_{\e}}\to 0$ and $\theta_{\e}\to0$ yield that $u_{\e}  \to u_z$ in measure on $I_{\lambda}(\lambda z)$. Here we used that $|\mathfrak{P}_{\e}-I|\leq\theta_{\e}$. Vitali's theorem then implies $u_{\e}\to u$ in $L^1(\Omega;\RR^2)$. 
	
	Next we bound the differences $u_{\e}(\e i)-u_{\e}(\e j)$ for all $i,j\in\mathbb{Z}^2$ with $|i-j|=1$.
	
	\medskip
	\ul{Step 1} (interactions within one cube)
	
	\noindent We start with $\e i,\e j\in \e\ZZ^2\cap I_{\lambda}(\lambda z)$ for the same $z$  and $|i-j|=1$.  Let us write $I=I_{\lambda}(\lambda z)$ for short. One has to distinguish several cases:
	
	\medskip
	\noindent {\it Case 1}: If $\dist(\e i,\partial I)\geq \pi \e\theta_{\e}^{-1}$ and $\dist(\e j,\partial I)\geq \pi \e\theta_{\e}^{-1}$, then \eqref{eq:interiorpoints} yields
	\begin{equation*}
	|u_{\e}(\e i)-u_{\e}(\e j)|=0 \, .
	\end{equation*}
	
	\noindent Since for neighboring lattice points it holds that
	\begin{equation} \label{eq:dist is lipschitz}
	\frac{\theta_{\e}}{\e}|\dist(\e i,\partial I)-\dist(\e j,\partial I)|\leq \theta_{\e}\, ,
	\end{equation}
	for the remaining cases we can assume that 
	\begin{equation}\label{eq:standing assump}
	\max\{\dist(\e i,\partial I),\dist(\e j,\partial I)\}< (\pi+1)\e\theta_{\e}^{-1}.
	\end{equation}
	
	\noindent {\it Case 2}: We next analyze when $P(\e i)$ and $P(\e j)$ lie on different $1$-dimensional boundary segments $S_i\neq S_j$ of $I$. We claim that $P(\e i)$ and $P(\e j)$ are then close to a node of the lattice  $\lambda\mathbb{Z}^2$. Indeed, denote by $\Pi_{S_i}$ and $\Pi_{S_j}$ the projections onto the subspaces spanned by the segments $S_i$ and $S_j$, respectively. Since by \eqref{eq:standing assump}
	\begin{equation*}
	|P(\e i)-P(\e j)|\leq \e|i-j|+\dist(\e i,\partial I)+\dist(\e j,\partial I)\leq (2\pi+2)\e\theta_{\e}^{-1}+\e \, ,
	\end{equation*}
	for $\e$ small enough the sides $S_i$ and $S_j$ cannot be parallel. Hence the point $\lambda z_{i,j}:=\Pi_{S_i}(\Pi_{S_j}(\e i))$ belongs to $S_i\cap S_j\subset\lambda\ZZ^2$ and therefore the $1$-Lipschitz continuity of $\Pi_{S_i}$ and $\Pi_{S_j}$ combined with \eqref{eq:standing assump} implies
	\begin{align}\label{eq:ilocate0}
	\dist(P(\e i),\lambda\ZZ^2)&=\dist(\Pi_{S_i}(\e i),\lambda\ZZ^2)\leq |\e i - \Pi_{S_j}(\e i)|\nonumber
	\\
	&\leq |\e i- \e j| +|\e j - \Pi_{S_j}(\e j)|  + |\Pi_{S_j}(\e j)-\Pi_{S_j}(\e i)| \leq 2\e+(\pi+1)\e\theta_{\e}^{-1}.
	\end{align}
	Exchanging the roles of $i$ and $j$ we derive by the same argument the bound
	\begin{equation}\label{eq:jlocate0}
	\dist(P(\e j),\lambda\ZZ^2)\leq 2\e+(\pi+1)\e\theta_{\e}^{-1}.
	\end{equation}
	For $\e$ small enough both terms can be bounded by $2\pi\e\theta_{\e}^{-1}$. In particular, the distance to $\lambda\ZZ^2$ of both $P(\e i)$ and $P(\e j)$ is realized by the point $\lambda z_{i,j}$, which is an endpoint of both the sides $S_i$ and $S_j$. Hence from the definition of the boundary condition $b_{z,\e}$ in \eqref{eq:defbcases} and \eqref{eq:defvSu}, and the fact that $c_0 > 2 \pi$ we deduce that
	\begin{equation*}
	b_{z,\e}(P(\e i)) = b_{z,\e}(P(\e j)) = u_{z_{i,j}}\,.
	\end{equation*}
	Equation~\eqref{eq:dist is lipschitz}, the $1$-Lipschitz continuity of $\Geo[u_{z_{i,j}},u_z]$ and the construction of $\bar{u}_{\e}$ yield $|\bar{u}_{\e}(\e i)-\bar{u}_{\e}(\e j)|\leq \theta_{\e}$. Due to the definition of the function $\mathfrak{P}_{\e}$ this inequality implies
	\begin{equation}\label{eq:verygoodbound}
	|u_{\e}(\e i)-u_{\e}(\e j)|\leq \theta_{\e} \, .
	\end{equation}
	Moreover, note that by \eqref{eq:standing assump},  \eqref{eq:ilocate0}, and \eqref{eq:jlocate0}, for $\e$ small enough,
	\begin{equation}\label{eq:reclocate1}
	\begin{split}
	\dist(\e j,\lambda \ZZ^2) & \leq |P(\e j)-\e j|+\dist(P(\e j),\lambda \ZZ^2)\leq (2\pi+2)\e\theta_{\e}^{-1}+2\e<2c_0\e\theta_{\e}^{-1}, \\
	\dist(\e i,\lambda \ZZ^2) & \leq |P(\e i)-\e i|+\dist(P(\e i),\lambda \ZZ^2)\leq (2\pi+2)\e\theta_{\e}^{-1}+2\e<2c_0 \e\theta_{\e}^{-1}.
	\end{split}	
	\end{equation}
	These inequalities will be used in Step~3 to count how many interactions fall into Case 2. 
	{\it Case 3}: Now consider points $i$ and $j$ such that $P(\e i)=\Pi_{S_i}(\e i)$ and $P(\e j)=\Pi_{S_i}(\e j)$ and assume additionally that $\dist(P(\e j),\lambda \ZZ^2)\geq  3c_0\e\theta_{\e}^{-1}$. Since $\Pi_{S_i}$ is $1$-Lipschitz, this implies that $\dist(P(\e i),\lambda \ZZ^2)\geq  2c_0\e\theta_{\e}^{-1}$ for $\e$ small enough. Hence by the definition of the boundary condition (cf. \eqref{eq:defbcases} and \eqref{eq:defvSu})   
	\begin{equation*}
	b_{z,\e}(P(\e i))=b_{z,\e}(P(\e j))={\rm mid}(u^-_{S_i},u^+_{S_i})  \, . 
	\end{equation*}
	Using again the $1$-Lipschitz-continuity of the geodesic $\Geo[{\rm mid}(u^-_{S_i},u^+_{S_i}),u_z]$, similar to \eqref{eq:verygoodbound} we obtain that
	\begin{equation}\label{eq:verygoodbound2}
	|u_{\e}(\e i)-u_{\e}(\e j)|\leq \theta_{\e} \, .
	\end{equation} 
	However, we need to analyze more accurately which points yield a non-zero difference. On the one hand, the projection property of $P$ and the definition of $\bar{u}_{\e}$ yield the implication
	\begin{equation}\label{eq:nointeraction}
	\text{\bf if} \quad \dist(\e j,S_i)=\dist(\e j,\partial I)\geq  \geo({\rm mid}( u^-_{S_i}, u^+_{S_i}), u_z) \e\theta_{\e}^{-1}  \quad \text{\bf then} \quad \bar{u}_{\e}(\e j)=  u_z \,.
	\end{equation}
	The same conclusion holds true for $\e i$. Hence for {\it Case 3} the estimate \eqref{eq:verygoodbound2} needs to be taken into account only for $(i,j)$ such that one of them violates the condition in \eqref{eq:nointeraction}, while for other couples $\e i,\e j$ the difference vanishes as in {\it Case~1}. 
	
	On the other hand, using that $P(\e i)=\Pi_{S_i}(\e i)$ and $P(\e j)=\Pi_{S_i}(\e j)$, one can show the following implication (where $\parallel$ means parallel):
	\begin{equation}\label{eq:notallinteractions}
	(\e i-\e j) \parallel S_i =0 \quad\implies \quad\dist(\e i,\partial I)=\dist(\e j,\partial I)\quad\implies\quad|u_{\e}(\e i)-u_{\e}(\e j)|=0.
	\end{equation}
	{\it Case 4}: It remains to treat the case of points $i$ and $j$ such that $P(\e i)=\Pi_{S_i}(\e i)$ and $P(\e j)=\Pi_{S_i}(\e j)$, but $\dist(P(\e j),\lambda \ZZ^2)< 3c_0\e\theta_{\e}^{-1}$. Here we use the Lipschitz-continuity of $b_{z,\e}$ on $S_i$ and the stability estimate of Lemma \ref{l.ongeodesics}. For the latter, we need that $b_{z,\e}(P(\e i))$ and $b_{z,\e}(P(\e j))$ are sufficiently close to $u_z$. Since on $S_i$ the boundary condition $b_{z,\e}$ is defined by geodesic interpolation between the elements of the vector $v_{S_i}(u)\in(\SS^1)^3$ defined in~\eqref{eq:defvSu} and $u_z\in\{u^-_{S_i},u^+_{S_i}\}$, we know that
	\begin{align*}
	|b_{z,\e}(P(\e i))-u_z|&\leq \geo(b_{z,\e}(P(\e i)),u_z)
	\\
	&\leq \max_{r=1,3}\geo((v_{S_i}(u),e_r),{\rm mid}(u^-_{S_i},u^+_{S_i}))+\geo({\rm mid}(u^-_{S_i},u^+_{S_i}),u_z)
	\\
	&=\max_{r=1,3}\geo((v_{S_i}(u),e_r),{\rm mid}(u^-_{S_i},u^+_{S_i}))+\frac{1}{2}\geo(u^-_{S_i},u^+_{S_i})
	\\
	&\leq\max_{r=1,3}\geo((v_{S_i}(u),e_r),u_z)+\geo(u^-_{S_i},u^+_{S_i})\,.
	\end{align*}
	Recall that the first and third component of $v_{S_i}[u]$ are given by the evaluation of $u$ at the endpoints of $S_i$. Hence by the almost continuity estimate \eqref{eq:almostcontinuity} we deduce for $\lambda\ll 1$ that $|b_{z,\e}(P(\e i))-u_z|< c$, where $c$ is the constant given by Lemma \ref{l.ongeodesics}. Repeating the argument one proves the analogue estimate for $P(\e j)$. To reduce notation, we set $d_{\e,i}=\theta_{\e}\e^{-1}\dist(\e i,\partial I)$ and $d_{\e,j}=\theta_{\e}\e^{-1}\dist(\e j,\partial I)$. Then by~\eqref{eq:bLip} and the applicable Lemma \ref{l.ongeodesics} we have
	\begin{align*}
	|\bar{u}_{\e}(\e i)-\bar{u}_{\e}(\e j)|\leq&\, \big|\Geo[b_{z,\e}(P(\e i)),u_z](d_{\e,i})-\Geo[b_{z,\e}(P(\e i)),u_z](d_{\e,j})\big|
	\\
	&+\big|\Geo[b_{z,\e}(P(\e i)),u_z](d_{\e,j})-\Geo[b_{z,\e}(P(\e j)),u_z](d_{\e,j})\big|
	\\
	\leq &\,|d_{\e,i}-d_{\e,j}|+\geo(b_{z,\e}(P(\e i)),b_{z,\e}(P(\e j)))
	\\
	\leq &\,\theta_{\e}+\frac{\pi}{4}\theta_{\e}\e^{-1}|\Pi_{S_i}(\e i)-\Pi_{S_i}(\e j)|\leq 2\theta_{\e} \, .
	\end{align*}
	Hence in {\it Case 4} we deduce the slightly weaker bound 
	\begin{equation}\label{eq:worstbound}
	|u_{\e}(\e i)-u_{\e}(\e j)|\leq 2\theta_{\e} \, .
	\end{equation}
	Finally the location condition on $j$ and \eqref{eq:standing assump} imply that 
	\begin{equation}\label{eq:reclocate2}
	\dist(\e j,\lambda \ZZ^2)\leq |P(\e j)-\e j|+\dist(P(\e j),\lambda \ZZ^2)< 4c_0\e\theta_{\e}^{-1}.
	\end{equation}

	\ul{Step 2} (interactions between different cubes)
	
	\noindent Now we consider points $\e i\in I_{\lambda}(\lambda z_i)$ and $\e j\in I_{\lambda}(\lambda z_j)$ with $z_i\neq z_j$ and $|i-j|=1$. By the definition of $\bar{u}_{\e}$ via geodesics and by the $1$-Lipschitz continuity of the latter we have 
	\begin{align*}
	\begin{split}
	|\bar{u}_{\e }(\e i)  -b_{z_i,\e}(P(\e i))|&=|\bar{u}_{\e}(\e i)-\Geo[b_{z_i,\e}(P(\e i)),\bar{u}_{z_i}](0) |\leq \frac{\theta_{\e}}{\e}\dist(\e i,\partial I_{\lambda}(\lambda z_i)) \, ,
	\\
	|\bar{u}_{\e }(\e j)  -b_{z_j,\e}(P(\e j))|&=|\bar{u}_{\e}(\e j)-\Geo[b_{z_j,\e}(P(\e j)),\bar{u}_{z_j}](0) |\leq \frac{\theta_{\e}}{\e}\dist(\e j,\partial I_{\lambda}(\lambda z_j)) \, .
	\end{split}
	\end{align*}
	Note that there exists a boundary segment $S_{ij}$ of $\de I_{\lambda}(\lambda z_j)$ such that the line segment $[\e i, \e j]$ intersects $S_{ij}$ orthogonally and moreover $S_{ij}\subset \partial I_{\lambda}(\lambda z_i)$. In particular, \begin{equation*}
	\dist(\e i,\partial I_{\lambda}(\lambda z_i))+\dist(\e j,\partial I_{\lambda}(\lambda z_j))\leq \e\,.
	\end{equation*} 
	Summing the previous two estimates then yields
	\begin{equation}\label{eq:generalbounds}
	| \bar{u}_{\e}(\e i)  -b_{z_i,\e}(P(\e i))|+| \bar{u}_{\e}(\e j)  -b_{z_j,\e}(P(\e j))|\leq\theta_{\e} \,.
	\end{equation}
	We claim that either $P(\e i)\in S_{ij}$ and $P(\e j)\in S_{ij}$ or that both $P(\e i)$ and $P(\e j)$ are close to $\lambda\ZZ^2$. Indeed, first assume that $P(\e j)\notin S_{ij}$. Then there exists another facet $S_j$ of~$I_{\lambda}(\lambda z_j)$ such that $P(\e j)\in S_j$. Since $\dist(\e j,S_j)\leq \e$ and $\dist(\e j,S_{ij})\leq\e$, the sides $S_j$ and $S_{ij}$ cannot be parallel. Denoting by $\Pi_S$ the projection onto the subspace spanned by a segment $S$, we deduce that $\Pi_{S_{j}}(\Pi_{S_{ij}}(\e j))\in S_j\cap S_{ij}\subset\lambda\ZZ^2$. Hence
	\begin{equation*}
	\dist(P(\e j),\lambda\ZZ^2)=\dist(\Pi_{S_j}(\e j),\lambda\ZZ^2)\leq |\e j - \Pi_{S_{ij}}(\e j)|\leq \e \, .
	\end{equation*}
	For $P(\e i)$ we check two possibilities. First consider $P(\e i)\in S_{ij}$. Then we may assume that $P(\e j)\notin S_{ij}$ as above. From the Lipschitz-continuity of $\Pi_{S_{ij}}$ we infer
	\begin{align*}
	\dist(P(\e i),\lambda \ZZ^2) & = \dist(\Pi_{S_{ij}}(\e i),\lambda\ZZ^2)\leq |\Pi_{S_{ij}}(\e i)-\Pi_{S_{ij}}(\Pi_{S_j}(\e j))|
	\\
	& \leq |\e i-\e j|+|\e j-\Pi_{S_j}(\e j)|\leq 2\e \, .
	\end{align*}
	On the contrary, if $P(\e i)\notin S_{ij}$, then there exists a facet $S_i\neq S_{ij}$ of $I_{\lambda}(\lambda z_i)$ such that~$P(\e i)\in S_i$. Since $S_i$ and $S_{ij}$ are both sides of the cube $I_{\lambda}(\lambda z_i)$ which cannot be parallel, we deduce that $\Pi_{S_i}(\Pi_{S_{ij}}(\e i))\in {\rm span}(S_i)\cap{\rm span}(S_{ij})\subset\lambda\ZZ^2$ and thus the defining property of $S_{ij}$ implies that
	\begin{equation*}
	\dist(P(\e i),\lambda\ZZ^2)=\dist(\Pi_{S_i}(\e i),\lambda\ZZ^2)\leq |\e i - \Pi_{S_{ij}}(\e i)|\leq \e
	\, .
	\end{equation*}
	It remains to establish an estimate for $\dist(P(\e j),\lambda\ZZ^2)$ when $P(\e i)\notin S_{ij}$ and $P(\e j)\in S_{ij}$. In this case we have
	\begin{align*}
	\dist(P(\e j),\lambda\ZZ^2) & =\dist(\Pi_{S_{ij}}(\e j),\lambda\ZZ^2)\leq |\Pi_{S_{ij}}(\e j)-\Pi_{S_{ij}}(\Pi_{S_i}(\e i)|
	\\
	& \leq |\e j-\e i|+|\e i-\Pi_{S_i}(\e i)|\leq 2\e \, .
	\end{align*} 
	To sum up, we have proved the following two alternatives: 
	\begin{itemize}
		\item[(i)] $P(\e i),P(\e j)\in S_{ij}$;
		\item[(ii)]  $\max\{\dist(P(\e i),\lambda\ZZ^2),\dist(P(\e j),\lambda\ZZ^2)\}\leq 2\e$.
	\end{itemize}
	Again we treat the two cases separately.

	\medskip
	\noindent {\it Case 5}: Note that the conditions in (ii) above imply that the unique points $\lambda \bar{z}_i,\lambda \bar{z}_j\in\lambda\ZZ^2$ realizing the minimal distance satisfy
	\begin{equation*}
	|\lambda \bar{z}_i-\lambda \bar{z}_j|\leq |\lambda \bar{z}_i-P(\e i)|+|P(\e i)-\e i|+|\e i-\e j|+|\e j-P(\e j)|+|P(\e j)-\lambda \bar{z}_j|\leq 7\e,
	\end{equation*}
	so that necessarily $\bar{z}_i=\bar{z}_j$ for $\e$ small enough. In particular, the construction of the boundary condition forces $b_{z_i,\e}(P(\e i))=b_{z_j,\e}(P(\e j))=u_{\bar{z}_i}$.  From~\eqref{eq:generalbounds} we infer
	\begin{equation*}
	|\bar{u}_{\e}(\e i)-\bar{u}_{\e}(\e j)|\leq |\bar{u}_{\e}(\e i)-u_{\bar{z}_i}|+|\bar{u}_{\e}(\e j)-u_{\bar{z}_j}|\leq \theta_{\e} \, , 
	\end{equation*}
	which by the definition of $\mathfrak{P}_{\e}$ allows to conclude that
	\begin{equation}\label{eq:verygoodbound3}
	|u_{\e}(\e i)-u_{\e}(\e j)|\leq\theta_{\e} \, . 
	\end{equation}
	Furthermore we know that
	\begin{equation} \label{eq:reclocate3}
	\dist(\e j, \lambda \ZZ^2) \leq \dist(P(\e j), \lambda \ZZ^2) + \dist(\e j, \de I_{\lambda}(\lambda z)) \leq 3 \e\, . 	
	\end{equation}
	
	\noindent{\it Case 6}:  We now analyze the case $P(\e i),P(\e j)\in S_{ij}$. By the symmetric definition  $b_{z_i,\e}$ and~$b_{z_j,\e}$ coincide on $S_{ij}$. Since by assumption the segment~$[\e i, \e j]$ is orthogonal to $S_{ij}$ and $S_{ij}\subset \partial I_{\lambda}(\lambda z_i)\cap\partial I_{\lambda}(\lambda z_j)$, we have $P(\e i) =\Pi_{S_{ij}}(\e i)= \Pi_{S_{ij}}(\e j)= P(\e j)$. Hence 
	estimate \eqref{eq:generalbounds} yields
	\begin{equation*}
	|\bar{u}_{\e}(\e i)-\bar{u}_{\e}(\e j)|   = |\bar{u}_{\e}(\e i)- b_{z_i,\e}(P(\e i))| + |b_{z_j,\e}(P(\e j))-\bar{u}_{\e}(\e j)|\leq \theta_{\e}   \, ,
	\end{equation*} 
	which again can be turned into an estimate for  $u_{\e}$ that reads
	\begin{equation}\label{eq:worstbound2}
	|u_{\e}(\e i)-u_{\e}(\e j)|\leq  \theta_{\e} \, . 
	\end{equation}
	Moreover, we can give an estimate for the location of $\e j$ by
	\begin{equation}\label{eq:reclocate4}
	\dist(\e j,S_{ij})=\dist(\e j,\partial I_{\lambda}(\lambda z_j))\leq\e.
	\end{equation}
	
	\medskip
	\ul{Step 3} (energy estimates)
	
	\noindent Let us first sum up our analysis hitherto. The interactions of couples $(\e i,\e j)$ with $|i-j|=1$ and at least one point in an half-open cube $I_{\lambda}(\lambda z)$ can be grouped as follows:
	\begin{itemize}[leftmargin=*]
		\item[(1)] In {\it Case 1} it holds that $|u_{\e}(\e i)-u_{\e}(\e j)|=0$. 
		\item[(2)] In the {\it Cases 2, 4}, and {\it 5} we have for $\e$ small enough $\dist(\e j,\lambda\ZZ^2)\leq 4c_0 \e\theta_{\e}^{-1}$ (see \eqref{eq:reclocate1}, \eqref{eq:reclocate2}, and \eqref{eq:reclocate3}) and by \eqref{eq:verygoodbound}, \eqref{eq:worstbound}, and \eqref{eq:verygoodbound3} the continuity estimate
		\begin{equation*}
		|u_{\e}(\e i)-u_{\e}(\e j)|\leq 2\theta_{\e} \, .
		\end{equation*}
		\item[(3)] In {\it Cases 3} and {\it 6}, according to \eqref{eq:verygoodbound2}--\eqref{eq:notallinteractions} respectively \eqref{eq:worstbound2}--\eqref{eq:reclocate4}, there exists a side $S$ of $I_{\lambda}(\lambda z)$ such that, setting $\kappa_{S}(z):=\frac{1}{2}\geo(u^-_{S},u^+_S)=\geo({\rm mid}(u^-_{S},u^+_S),u_z)$,
		\begin{equation*}
		|u_{\e}(\e i)-u_{\e}(\e j)|  \leq \begin{cases}
		\theta_{\e} &\mbox{if $\dist(\e j,S)<\kappa_{S}(z) \frac{\e}{\theta_{\e}} + \e$ and $(i-j) \perp S$} \, ,
		\\
		0 &\mbox{otherwise.}
		\end{cases}
		\end{equation*}
	\end{itemize}
	
	\noindent Now we can estimate the discrete energy. Due to (1)-(3) above it is bounded by
	\begin{align*}
	\frac{1}{\e \theta_\e} E_\e(u_{\e})\leq& \,C\e\theta_{\e}\#\Big(\Omega_{\e}\cap\big\{\dist(\cdot,\lambda\ZZ^2)\leq C\e\theta_{\e}^{-1}\big\}\Big)
	\\
	&+\sum_{I_{\lambda}(\lambda z)\cap \Omega\neq\emptyset}\sum_{S\subset\partial I_{\lambda}(\lambda z)}\e\theta_{\e}\#\left(\e\ZZ^2\cap I_{\lambda}(\lambda z)\cap\left\{\dist(\cdot,S)< \tfrac{\e}{\theta_{\e}}\kappa_{S}(z)+2\e\right\}\right) \, ,
	\end{align*}
	where we used that each point in $\mathbb{Z}^2$ has four neighbors, but for (3) we have to count only half of the interactions. We claim that the first right hand side term vanishes when $\e\to 0$. To this end, fix a large cube $Q$ such that $\Omega \compact Q$. For $\e$ small enough we have
	\begin{equation*}
	\e^2\#\Big(\Omega_{\e}\cap\big\{\dist(\cdot,\lambda\ZZ^2)\leq C\e\theta_{\e}^{-1}\big\}\Big)\leq \sum_{x\in\lambda\ZZ^2\cap Q}|B_{2C\e\theta_{\e}^{-1}}(x)|
	\leq C|Q|\lambda^{-2}\e^2\theta_{\e}^{-2}.
	\end{equation*} 
	Inserting this estimate into the first term, we obtain
	\begin{equation}\label{eq:bound1st}
	\e\theta_{\e}\#\Big(\Omega_{\e}\cap\big\{\dist(\cdot,\lambda\ZZ^2)\leq C\e\theta_{\e}^{-1}\big\}\Big)\leq C|Q|\lambda^{-2}\frac{\e}{\theta_{\e}} \, ,
	\end{equation}
	which vanishes when $\e\to 0$ due to the assumption $\e\ll\theta_{\e}$.
	
	Now we treat the second term. Since each segment $S$ has length $\lambda$, for any fixed $\kappa>0$ it holds that
	\begin{equation*}
	\e\theta_{\e}\#\left(\e\ZZ^2\cap I_{\lambda}(\lambda z)\cap\left\{\dist(\cdot,S)< \tfrac{\e}{\theta_{\e}}\kappa+2\e\right\}\right)\leq \frac{\theta_{\e}}{\e}\Big(\lambda+\frac{2\e}{\theta_{\e}}\kappa+6\e\Big)\Big(\frac{\e}{\theta_{\e}}\kappa+4\e\Big) \, . 
	\end{equation*}
	For only finitely many cubes $I_{\lambda}(\lambda z)$ intersecting $\Omega$, we can insert this estimate with $\kappa=\kappa_{S}(z)$, pass to the limit in $\e$ and obtain by \eqref{eq:bound1st} that
	\begin{equation}\label{eq:limiteps}
		\limsup_{\e\to 0}\frac{1}{\e \theta_\e} E_\e(u_{\e})  \leq\sum_{I_{\lambda}(\lambda z)\cap \Omega\neq\emptyset}\sum_{S\subset\partial I_{\lambda}(\lambda z)} \frac{\lambda}{2}\geo(u^-_{S},u^+_{S})  
		=  \integral{J_{u}\cap \Omega^{\lambda}}{\geo(u^-,u^+)|\nu_{u}|_1}{\d\mathcal{H}^1} \, .
	\end{equation}
	This estimate agrees with \eqref{eq:toshow} and hence concludes the proof. 
\end{proof}

\subsection{Upper bound in presence of vortices}  \label{sec:with vortices}

Also in the case of vortices the construction of the recovery sequence is done by gradually simplifying the map $u \in BV(\Omega;\SS^1)$, following the main idea of Section~\ref{sec:upper bound no vortices}. However, due to the presence of the vortex measure $\mu = \sum_{h=1}^N d_h \delta_{x_h}$, in general the map $u$ cannot be approximated by smooth maps with values in $\SS^1$. This requires additional steps in the simplification of $u$. For notational convenience, set
\begin{equation*}
\mathcal{E}(\mu,u)=\integral{\Omega}{|\nabla u|_{2,1} }{\d x} + |\DD^{(c)} u|_{2,1}(\Omega) + \mathcal{J}(\mu,u;\Omega)
\end{equation*}
with $\mathcal{J}(\mu,u;\Omega)$ given by \eqref{eq:def of surface L}. 
We start with the approximation result for currents $T$ with boundary $\de T = - \mu \x \llbracket \SS^1 \rrbracket$ which appear in the definition of the limit functional. 

\begin{lemma}[Approximations creating finitely many singularities] \label{lemma:approximation with sing}
Let $\mu = \sum_{h=1}^N d_h \delta_{x_h}$ and let $T \in \D_2(\Omega \x \RR^2)$ be such that $T \in \cart(\Omega_\mu \x \SS^1)$ and  $\de T|_{\Omega \x \RR^2} = - \mu \x \llbracket \SS^1 \rrbracket$. Then there exist an open set $\tilde \Omega \Supset \Omega$ and a sequence $u_k \in C^\infty(\tilde \Omega_\mu;\SS^1) \cap W^{1,1}(\tilde \Omega; \SS^1)$ such that 
\begin{align}
u_k & \to u_T \quad \text{in } L^1(\Omega;\RR^2) \, , \label{eq:smooth converge L1}\\
G_{u_k} & \weak T \quad \text{in } \D_2(\Omega_\mu \x \RR^2) \, , \label{eq:weak conv outside points} \\
|G_{u_k}|(\Omega \x \RR^2) & \to |T|(\Omega \x \RR^2)\, , \label{eq:mass smooth}\\
\deg(u_k)(x_h) & = d_h \, , \quad \text{for } h = 1,\dots, N\,. \label{eq:degree of smooth}
\end{align}
\end{lemma}
\begin{proof}
Let $\Omega'$ and $\Omega''$ be open sets with Lipschitz boundary such that $\Omega' \compact \Omega'' \compact \Omega$ and $\{x_1,\dots,x_N\} \subset \Omega'$ and let us define the open set $O := \Omega \sm \ol \Omega'$. Since $\de T|_{O \x \RR^2} = 0$, we have $T \in \cart(O \x \SS^1)$. By Lemma~\ref{lemma:extension of currents} there exist an open set $\tilde O \Supset O$ and a current $\tilde T \in \cart(\tilde O \x \SS^1)$ such that $\tilde T|_{O \x \RR^2} = T|_{O \x \RR^2}$ and $|\tilde T|(\de O \x \RR^2) = 0$. In particular,
\begin{equation} \label{eq:0904191851}
\tilde T|_{(\Omega'' \sm \ol \Omega') \x \RR^2} = T |_{(\Omega'' \sm \ol \Omega') \x \RR^2} \, . 
\end{equation}
This allows us to glue together the currents $T$ and $\tilde T$. To do so, we define the set~$\tilde \Omega := \Omega \cup \tilde O$ and the current $S \in \D_2(\tilde \Omega \x \RR^2)$ as follows. Fix a cut-off function $\zeta \in C^\infty_c(\Omega'')$ such that $0 \leq \zeta   \leq 1$ and $\zeta \equiv 1$ on  a neighborhood of  $\ol \Omega'$. For every $\omega \in \D^2(\tilde \Omega \x \RR^2)$ we put $S(\omega) :=  T( \zeta \omega) + \tilde T\big((1-\zeta) \omega \big)$. Then by~\eqref{eq:0904191851} it follows that $S |_{\Omega \x \RR^2} = T$, $S|_{(\tilde \Omega \sm \ol \Omega') \x \RR^2} = \tilde T|_{(\tilde \Omega \sm \ol \Omega') \x \RR^2}$, and~$|S|(\de \Omega \x \RR^2) =0$. In particular, using the product rule for the exterior derivative, for any $1$-form $\omega \in\mathcal{D}^1(\tilde{\Omega}\x\SS^1)$ we find that
\begin{align*}
\partial S(\omega) & = T(\zeta \d \omega)+\tilde{T}((1-\zeta)\d \omega)
\\
& = \partial T(\zeta \omega)-T(\d \zeta\wedge \omega)+\partial \tilde{T}((1-\zeta)\omega)+\tilde{T}(\d\zeta\wedge\omega)
=-\mu\x\llbracket\SS^1\rrbracket(\omega)\,,
\end{align*}
where we used that $\zeta\equiv 1$ on ${\supp(\mu)}$ and $\d\zeta\wedge\omega\in \mathcal{D}^2(\Omega''\setminus\ol{\Omega}'\x\RR^2)$.
Hence $S \in \cart(\tilde \Omega_\mu \x \SS^1)$ and $\de S|_{\tilde \Omega \x \RR^2} = - \mu \x \llbracket \SS^1 \rrbracket$.

Since $S \in \cart(\tilde \Omega_\mu \x \SS^1)$, by the Approximation Theorem for Cartesian currents (Theorem~\ref{thm:approximation}) there exists a sequence $u_k \in C^\infty(\tilde \Omega_\mu; \SS^1) \cap W^{1,1}(\tilde \Omega; \SS^1)$ such that $G_{u_k} \weak S$ in $\D_2(\tilde \Omega_\mu \x \RR^2)$ and $|G_{u_k}|(\tilde \Omega_\mu \x \RR^2) \to |S|(\tilde \Omega_\mu \x \RR^2)$. In particular, we get~\eqref{eq:weak conv outside points} and thus~\eqref{eq:smooth converge L1}. Moreover
\begin{equation} \label{eq:1311181701}
|G_{u_k}|(\tilde \Omega \x \RR^2) \to |S|(\tilde \Omega \x \RR^2) \, ,
\end{equation}
since $G_{u_k}$ and $S$ do not charge the sets $\{x_h\} \x \RR^2$ (being i.m.\ rectifiable 2-currents concentrated on a subset of $\RR^2 \x \SS^1$, see also Remark~\ref{rmk:punctured domains}). 

Thanks to the convergence in~\eqref{eq:weak conv outside points} we can prove~\eqref{eq:degree of smooth}. Indeed, let $h = 1,\dots,N$. For~$\rho > 0$ small enough we have (e.g., by~\cite[Section 6, Proposition 1]{Gia-Mod-Sou-S1})
\begin{equation*}
\de G_{u_k}|_{B_\rho(x_h) \x \RR^2} = - \deg(u_k)(x_h) \delta_{x_h} \x \llbracket \SS^1 \rrbracket \, .
\end{equation*}
Fix a cut-off function $\zeta \in C^\infty_c(B_\rho(x_h))$ such that $\zeta \equiv 1$ on $B_{\rho/2}(x_h)$ and define the 1-form $\omega = \zeta \omega_{\SS^1}$, $\omega_{\SS^1}$ being the 0-homogeneous extension of the volume form of $\SS^1$ to $\RR^2 \sm \{0\}$. Observing that $ \d \omega \in \D^2\big((B_\rho(x_h)\sm \{x_h\}) \x \RR^2\big)$, the convergence in~\eqref{eq:weak conv outside points} implies that 
\begin{equation*}
- \deg(u_k)(x_h) = \de G_{u_k} (\omega) = G_{u_k}(\! \d \omega) \to T(\! \d \omega) = \de T(\omega) = - d_h   \, .
\end{equation*}
Then $\deg(u_k)(x_h)$ is a sequence of integer numbers which converges to the integer number~$d_h$. Thus for $k$ large enough $\deg(u_k)(x_h) = d_h$. 

To conclude, we observe that~\eqref{eq:1311181701} and $|S|(\de \Omega \x \RR^2) = 0$ imply~\eqref{eq:mass smooth}.
\end{proof}
The next result shows how to reduce the analysis to singularities with degree $\pm 1$.

\begin{lemma}[Splitting of the degree] \label{lemma:splitting degree}
Let $V:= \{x_1,\dots, x_N\} \subset \Omega$ and let $u \in C^\infty(\Omega \sm V;\SS^1) \cap W^{1,1}(\Omega;\SS^1)$ be such that $\deg(u)(x_h) \neq 0$ for $h=1,\dots,N$. Then for $0<\tau\ll 1$ there exist a set $V_{\tau} = \{x_1^{\tau}, \dots, x_{N_{\tau}}^{\tau}\} \subset \Omega$ and a function $u^{\tau} \in C^\infty(\Omega \sm V_{\tau};\SS^1) \cap W^{1,1}(\Omega;\SS^1)$ such that $u^{\tau} \to u$ strongly in $L^1(\Omega;\RR^2)$,
$|\deg(u^{\tau})(x^{\tau}_h)| = 1$ for $h=1,\ldots,N_{\tau}$ and
\begin{align} 
	\lim_{{\tau}\to 0}  \int_{\Omega}|\nabla u^{\tau}|_{2,1}\,\mathrm{d}x &= \int_{\Omega}|\nabla u|_{2,1}\,\mathrm{d}x \, , \label{eq:convergence less degree} \\
	N_{\tau}  = \sum_{h=1}^{N_{\tau}}|\deg(u^{\tau})(x^{\tau}_h)| &= \sum_{h=1}^{N}|\deg(u)({x}_h)| \, . \label{eq:degree splitted}
\end{align}
Moreover, defining the measures $\mu^{\tau}=\sum_{h=1}^{N_{\tau}}\deg(u^{\tau})(x_h^{\tau})\delta_{x_h^{\tau}}$, we have that \begin{equation*}
\mu^{\tau}\flat\sum_{h=1}^{N}\deg(u)(x_h)\delta_{x_h}\quad\text{ as }\tau\to 0\,.
\end{equation*}
Finally, if $u \in C^\infty(\tilde \Omega \sm V;\SS^1) \cap W^{1,1}(\tilde \Omega;\SS^1)$ for some $\tilde \Omega \Supset \Omega$, then one can additionally choose $u^{\tau} \in C^\infty(\tilde \Omega \sm V;\SS^1) \cap W^{1,1}(\tilde \Omega;\SS^1)$.
\end{lemma}

\begin{remark} \label{rmk:also with zero degree}
	In this section we shall apply Lemma~\ref{lemma:splitting degree} to functions given by Lemma~\ref{lemma:approximation with sing}, cf.\ \eqref{eq:degree of smooth}. 
	In~\cite{Cic-Orl-Ruf} we have to consider $u \in C^\infty(\Omega \sm V;\SS^1) \cap W^{1,1}(\Omega;\SS^1)$ without assuming that $\deg(u)(x_h) \neq 0$ for every $h=1,\dots,N$. In that case, the statement of the lemma holds true, but~\eqref{eq:degree splitted} needs to be adapted to 
	\begin{equation*}
		N_{\tau} = \sum_{h=1}^{N}|\deg(u)({x}_h)| + 2 \# \{ x_h \ : \ \deg(u)({x}_h) = 0 \}\, .
	\end{equation*}
	The argument in the proof remains unchanged. 
\end{remark}

\begin{proof}[Proof of Lemma~\ref{lemma:splitting degree}]
	
	\noindent Via an iterative construction we create $|\deg(u)(x_h)|$ different singularities out of one singularity whenever $|\deg(u)(x_h)|>1$ in such a way that the new function is close in energy. To reduce notation, we assume that $x_1=0$ and $\deg(u)(x_1)>1$ (the case of a negative degree less than $-1$ can be treated similarly). We equip $\RR^2$ with the complex product, which we denote by $\odot$. Given $0<\tau\ll 1$ we set $u^{\tau}\in C^{\infty}(\Omega\setminus (V\cup \{\tau e_1\});\SS^1)$ as 
	\begin{equation*}
	u^{\tau}(x)= u(x)\odot\big(\tfrac{x}{|x|}\big)^{-1}\odot\tfrac{(x-\tau e_1)}{|x-\tau e_1|} \, .
	\end{equation*}
	Defined as above, it follows that $u^{\tau}\to u$ in $L^1(\Omega;\RR^2)$ by dominated convergence. Next, we estimate its anisotropic gradient norm. By the product rule, for $i=1,2$ we obtain
	\begin{align*}
	\partial_iu^{\tau}(x)& = \partial_i u (x) \odot\big(\tfrac{x}{|x|}\big)^{-1}\odot\tfrac{(x-\tau e_1)}{|x-\tau e_1|}
	\\
	&\quad +u(x)\odot\left\{\partial_i\big(\tfrac{x}{|x|}\big)^{-1}\odot\tfrac{(x-\tau e_1)}{|x-\tau e_1|}+ \big(\tfrac{x}{|x|}\big)^{-1}\odot\partial_i\tfrac{(x-\tau e_1)}{|x-\tau e_1|}\right\} \, .
	\end{align*}
	A straightforward computation shows that, for a.e.\ $x\in\Omega$, we have
	\begin{align*}
	\lim_{\tau\to 0}\partial_i u^{\tau}(x) & = \partial_i u(x) +u\odot\left\{\partial_i\big(\tfrac{x}{|x|}\big)^{-1}\odot\tfrac{x}{|x|}+ \big(\tfrac{x}{|x|}\big)^{-1}\odot\partial_i\tfrac{x}{|x|}\right\}
	\\
	& = \partial_i u(x) +u\odot\partial_i\left\{\big(\tfrac{x}{|x|}\big)^{-1}\odot\tfrac{x}{|x|}\right\} = \partial_i u(x) \, .
	\end{align*} 
	In order to use dominated convergence, we observe that $\big(\tfrac{x}{|x|}\big)^{-1}=\tfrac{1}{|x|}(x_1,-x_2)$, so that
	\begin{equation*}
	|\partial_iu^{\tau}(x)|\leq |\partial_i u(x)|+\left|\partial_i\big(\tfrac{x}{|x|}\big)^{-1}\right|+\left|\partial_i\big(\tfrac{x-\tau e_1}{|x-\tau e_1|}\big)\right|\leq |\partial_i u(x)|+\tfrac{2}{|x|}+\tfrac{2}{|x-\tau|} \, .
	\end{equation*}
	The right-hand side is equi-integrable on $\Omega\subset\RR^2$, so that we conclude
	\begin{equation*}
	\lim_{\tau\to 0}\int_{\Omega}|\nabla u^{\tau}|_{2,1}\,\mathrm{d}x=\int_{\Omega}|\nabla u|_{2,1}\,\mathrm{d}x \, .
	\end{equation*}
	Finally, we need to compute the degree of $u^{\tau}$. Let us introduce the complex-valued functions
	\begin{equation*}
	\tilde{u}(x)=u_1(x)+\iota u_2(x),\quad{f}(x)=\tfrac{1}{|x|}(x_1-\iota x_2),\quad{g}(x)=\tfrac{1}{|x-\tau e_1|}\big((x_1-\tau)+\iota x_2\big) \, .
	\end{equation*}
	In $\RR^2$ the degree around a point $\ol x$ can be expressed via the winding number, that means
	\begin{align*}
	(2\pi \iota)\deg(u^{\tau})(\ol x)=&\int_{\partial B_r(\ol x)}\frac{\d(\tilde{u}fg)}{\tilde{u}fg}=\int_{\partial B_r(\ol x)}\frac{(\! \d\tilde{u}) fg+\tilde{u}(\! \d f) g+\tilde{u}f (\!\d g)}{\tilde{u}fg}
	\\
	=&2\pi\iota \left(\deg(u)(\ol x)-\delta_0(\ol x)+\delta_{\tau e_1}(\ol x)\right) \, .
	\end{align*}
	We deduce that the degree of $u^{\tau}$ is of the form (recalling that $x_1=0$)
	\begin{equation*}
	\deg(u^{\tau})(\ol x)=\begin{cases}
	\deg(u)(x_1)-1 &\mbox{if } \ol x=x_1 \, ,
	\\
	1 &\mbox{if } \ol x=x_1+\tau e_1 \, ,
	\\
	\deg(u)(\ol x) &\mbox{otherwise} \, ,
	\end{cases}
	\end{equation*}
	where for the second equality we used that $\deg(u)(x_1+\tau e_1)=0$ due to the local smoothness of $u$ around $x_1+\tau e_1$ (see \cite[Corollary 8]{Br-Ni}). Repeating this construction (with $\tau/2$, $\tau/4$, and so on) we find a finite set $V_{\tau}=\{{x}^{\tau}_1,\dots,{x}^{\tau}_{N_{\tau}}\}$ and a sequence $u^{\tau}\in C^{\infty}(\Omega\setminus V_{\tau};\SS^1)\cap W^{1,1}(\Omega;\SS^1)$ such that $u^{\tau}\to u$ in $L^1(\Omega;\RR^2)$,  $\deg(u^{\tau})({x}^{\tau}_h)\in\{\pm 1\}$ and~\eqref{eq:convergence less degree}--\eqref{eq:degree splitted} hold true. The claim on the flat convergence follows by the construction.
\end{proof}
In the next lemma we move the singularities onto a lattice $\lambda_n\ZZ^2$ which makes them compatible with a piecewise constant approximation $u_n\in\mathcal{PC}_{\lambda_n}(\SS^1)$.  
\begin{lemma}[Moving singularities on a lattice] \label{lemma:correct points}
Let $V:= \{x_1, \dots, x_N\} \subset \Omega$, let $\tilde \Omega \Supset \Omega$, and let  $u \in C^\infty(\tilde \Omega \sm V;\SS^1) \cap W^{1,1}(\tilde \Omega;\SS^1)$. Then for every $\lambda > 0$ there exist a set $V_\lambda = \{x_1^\lambda,\dots,x_N^\lambda\} \subset \lambda \ZZ^2 \cap \Omega$ and a map $u^\lambda \in C^\infty(\tilde \Omega \sm V_\lambda;\SS^1) \cap W^{1,1}(\tilde \Omega;\SS^1)$ such that $u^\lambda \to u$ strongly in $W^{1,1}(\tilde \Omega;\SS^1)$ as $\lambda \to 0$. Moreover, $\deg(u^{\lambda})(x_h^{\lambda}) = \deg(u)(x_h)$ for $h = 1, \dots, N$ for $\lambda$ small enough and, defining the measures $\mu^{\lambda}=\sum_{h=1}^N\deg(u^{\lambda})(x_h^{\lambda})\delta_{x_h^{\lambda}}$, it holds that $\mu^{\lambda}\flat \sum_{h=1}^N\deg(u)(x_h)\delta_{x_h}$.
\end{lemma}
\begin{proof}
	We set $u^\lambda := u \circ \psi_\lambda$, where $\psi_\lambda \colon \tilde \Omega \to \tilde \Omega$ is a suitable diffeomorphism with $\psi_\lambda(x^\lambda_h) = x_h$ (see, e.g., \cite[p.\ 210]{Guo-Ran-Jos} for a construction). The details are omitted as they are standard.
\end{proof}
We modify the target sequence one last time  close to the singularities.
\begin{lemma}[Modification near a singularity] \label{lemma:modifications near sing}
Let $\rho \leq 1$ and let $u \in C^\infty(B_\rho \sm \{0\};\SS^1) \cap W^{1,1}(B_\rho;\SS^1)$ with $|\deg(u)(0)|= 1$. Then for every $\sigma > 0$ there exist $\tilde u \in C^\infty(B_\rho \sm \{0\};\SS^1)\cap W^{1,1}(B_\rho;\SS^1)$ and a radius $\eta_0 \in (0,\rho^2)$ such that 
\begin{itemize}
\item[i)] $\tilde u(x) = \frac{(x_1, \pm x_2)}{|x|}$ for every $x \in \ol B_{\eta_0} \sm \{0\}$;
\item[ii)] $\tilde u(x) = u(x)$ for every $x \in B_\rho \sm \ol B_{\sqrt{\eta_0}}$;
\item[iii)] it holds that
\begin{equation*}
\integral{B_\rho}{|\nabla \tilde u|_{2,1}}{\d x} \leq \integral{B_\rho}{|\nabla u|_{2,1}}{\d x} + \sigma \, .
\end{equation*}
\end{itemize}
\end{lemma}
\begin{proof}
We give a proof in the case $\deg(u)(0) = 1$, the case $\deg(u)(0) = -1$ being completely analogous. We also assume, without loss of generality, that $\rho = 1$ and we denote $B_\rho$ simply by~$B$. Let us consider the set $\Sigma := \{ (x_1,0)  \ : \ 0 \leq x_1 \leq 1 \}$. To modify the map $u$ we will actually modify its lifting $\varphi$. We start by discussing some useful properties of $\varphi$.  

Since $B \sm \Sigma$ is simply connected, the map $u \colon B \sm \Sigma \to \SS^1$ admits a lifting $\varphi \colon B \sm \Sigma \to \RR$, i.e., a function satisfying $u = \exp(\iota \varphi)$. The function $\varphi$ is unique up to integer multiples of~$2 \pi$ and has the same regularity of $u$, namely $\varphi \in C^\infty(B\sm \Sigma;\RR) \cap W^{1,1}(B \sm \Sigma;\RR)$. 

The fact that $u \in C^\infty(B \sm \{0\};\SS^1)$ can be translated in terms of the regularity of $\varphi$ as follows: for every $x \in \Sigma \cap (B \sm \{0\})$ and $r > 0$ such that $B_{r}(x) \compact B \sm \{0\}$, we have
\begin{equation} \label{eq:lifting extended}
\mathds{1}_{\ol B^+_r(x) } \varphi + \mathds{1}_{B^-_r(x)} (\varphi - 2\pi) \in C^\infty(B_r(x)) \, ,
\end{equation}
where $B^\pm_r(x)  = \{ x \in B_r(x) \colon \pm x_2 > 0 \}$. To show this, we observe that $u \in C^\infty(B_{r}(x);\SS^1)$, thus it admits a lifting $\varphi_x \in C^\infty(B_{r}(x))$. Up to adding an integer multiple of $2\pi$ to $\varphi_x$, by the uniqueness of the lifting up to integer multiples of $2\pi$ we have $\varphi = \varphi_x$ in $B^+_r(x)$ and there exists a $k_x \in \ZZ$ such that $\varphi = \varphi_x + 2 \pi k_x$ in $B^-_{r}(x)$, thus $\mathds{1}_{\ol B^+_r(x) } \varphi + \mathds{1}_{B^-_r(x)} (\varphi - 2 \pi k_x) = \varphi_x \in C^\infty(B_r(x))$. To prove that $k_x = 1$, we observe that the proven regularity implies, in particular, that the restrictions $\varphi|_{{B^\pm_r(x)}}$ admit traces $\varphi^\pm$ on~$\Sigma$ in the classical sense. To compute the jump $\varphi^+ - \varphi^- = -2\pi k_x$ at a point $x \in \Sigma \cap (B \sm \{0\})$, we parametrize the circle $\de B_{|x|}$ counterclockwise with the closed path $\gamma(t) = |x|\exp(\iota 2 \pi t)$, $t \in [0,1]$. Observing that $\nabla \varphi  = u_1 \nabla u_2 - u_2 \nabla u_1$ in $B \sm \Sigma$, we infer that 
\begin{equation*}
\begin{split}
\varphi^+(x) - \varphi^-(x) & = \varphi(\gamma(0^+)) - \varphi(\gamma(1^-)) = - \int_0^1 \frac{\d}{\d t} \big(\varphi(\gamma(t)) \big) \d t = - \int_0^1 \nabla \varphi(\gamma(t)) \cdot \dot \gamma(t) \d t  \\
& = - \int_\gamma (u_1 \nabla u_2 - u_2 \nabla u_1) \cdot \tau \d \H^1 = - 2 \pi \deg(u)(0) = - 2 \pi \, .
\end{split}
\end{equation*}
This proves $k_x = 1$, and in turn~\eqref{eq:lifting extended}.  

Finally, it holds that $\varphi \in L^2(B)$ due to the Sobolev Embedding Theorem. 

We are now in a position to define a modification $\tilde \varphi$ of $\varphi$. Let us fix $\sigma > 0$. By a classical capacity argument, we find $\eta_0 > 0$ small enough and a cut-off function $\zeta \in C^\infty_c(B_{\sqrt{\eta_0}})$, $0 \leq \zeta \leq 1$, $\zeta \equiv 1$ on $\ol B_{\eta_0}$ satisfying  
\begin{equation*}
\|\nabla \zeta\|^2_{L^2(B)} \leq \frac{C}{\log \tfrac{\sqrt{\eta_0}}{\eta_0}} \leq \frac{C}{| \log \eta_0 | } < \sigma^2  .
\end{equation*}

We use the cut-off $\zeta$ to interpolate between $\varphi$ and the principal argument $\arg$. The function $\arg$ is defined in polar coordinates $(\rho,\vartheta)$ by $\arg(\rho,\vartheta) = \vartheta$ and satisfies $\frac{x}{|x|} = \exp(\iota \arg(x))$ in $B \sm \{0\}$. In particular, also $\arg \in C^\infty(B \sm  \Sigma;\RR) \cap W^{1,1}(B\sm \Sigma;\RR)$ and it satisfies the regularity property as in~\eqref{eq:lifting extended}. Let us define
\begin{equation*}
\tilde \varphi := \zeta \arg + (1-\zeta) \varphi \, , \quad \tilde u := \exp(\iota \tilde \varphi) \, .
\end{equation*}
Since $\tilde \varphi \in C^\infty(B \sm  \Sigma;\RR)$ and for every $x \in \Sigma \cap (B \sm \{0\})$ and $r > 0$ such that $B_{r}(x) \compact B \sm \{0\}$, we have $\mathds{1}_{\ol B^+_r(x) } \tilde \varphi + \mathds{1}_{B^-_r(x)} (\tilde \varphi - 2\pi) \in C^\infty(B_r(x);\RR)$, we deduce that   $\tilde u \in C^\infty(B \sm \{0\}; \SS^1)$. By definition $\tilde u$ satisfies (i) and (ii). To prove (iii), let us compute 
\begin{equation*}
\begin{split}
\integral{B}{|\nabla \tilde u|_{2,1}}{\d x} & = \integral{B \sm \Sigma}{|\nabla \tilde \varphi|_1}{\d x} \leq \integral{B \sm \Sigma}{|\nabla \tilde \varphi - \nabla \varphi|_1}{\d x}  +\integral{B \sm \Sigma}{|\nabla \varphi|_1}{\d x}  \\
& = \integral{B \sm \Sigma}{|\nabla \tilde \varphi - \nabla \varphi|_1}{\d x}  +\integral{B}{|\nabla u|_{2,1}}{\d x} \, ,
\end{split}
\end{equation*}
thus it only remains to estimate the first integral in the right-hand side. We have
\begin{equation*}
\begin{split}
\integral{B \sm \Sigma}{|\nabla \tilde \varphi - \nabla \varphi|_1}{\d x} & \leq \sqrt{2}\integral{B \sm \Sigma}{|\nabla \zeta| |\arg - \varphi|}{\d x} + \sqrt{2}\integral{B \sm \Sigma}{\zeta | \nabla \arg - \nabla \varphi|}{\d x} \\
& \leq  \sqrt{2}\|\nabla \zeta\|_{L^2(B)} \big(\|\arg\|_{L^2(B)} + \|\varphi\|_{L^2(B)}\big) + \sqrt{2} \integral{B_{\sqrt{\eta_0}} \sm \Sigma}{\hspace{-0.5em}\big( | \nabla \arg| + |\nabla \varphi| \big)}{\d x} \\
& \leq \big( \sqrt{2}\|\arg\|_{L^2(B)} + \sqrt{2}\|\varphi\|_{L^2(B)} + 1 \big) \sigma \, , 
\end{split}
\end{equation*}
if $\eta_0 > 0$ is also chosen small enough such that 
\begin{equation*}
\sqrt{2}\integral{B_{\sqrt{\eta_0}} \sm \Sigma}{\big( | \nabla \arg| + |\nabla \varphi| \big)}{\d x}  < \sigma \,.
\end{equation*}
\end{proof}

 Now we are in a position to construct the discrete recovery sequence.

\begin{proposition}\label{p.smoothapprox}
Assume that $\e\ll\theta_{\e}\ll \e |\log \e|$. Let $\lambda_n := 2^{-n}$, $n \in \NN$ and let $V:=\{x_1,\dots,x_N\} \subset \lambda_n \ZZ^2 \cap \Omega$. Assume that $u \in C^\infty(\tilde  \Omega \sm V; \SS^1) \cap W^{1,1}(\tilde \Omega; \SS^1)$ with $\tilde \Omega \Supset \Omega$ has the following structure: $|\deg(u)(x_h)|=1$ for all $1\leq h\leq N$ and 
\begin{equation*}
u(x)=\left(\begin{smallmatrix} 1 &0 \\0 &\deg(u)(x_h)\end{smallmatrix}\right) \tfrac{x-x_h}{|x-x_h|}\quad\text{ in }\quad B_{\eta_0}(x_h)
\end{equation*}
for some $\eta_0>0$. Set moreover $\mu=\sum_{h=1}^N\deg(u)(x_h)\delta_{x_h}$. Then there exists a sequence $u_\e  \colon \Omega_\e \to \S_\e$ such that $\mu_{u_\e} \flat \mu$, $u_\e \to u$ in $L^1(\Omega;\RR^2)$, and
		\begin{equation*}
				\limsup_{\e \to 0}\Big( \frac{1}{\e \theta_\e} E_\e(u_\e) - 2\pi |\mu|(\Omega) |\log \e| \frac{\e}{\theta_\e}  \Big) \leq \int_{\Omega}|\nabla u|_{2,1}\,\mathrm{d} x \, .
		\end{equation*} 
\end{proposition}
\begin{remark}\label{rmk:also works with e log smaller theta}
We emphasize that the construction presented below also works under the sole assumption $\e \ll \theta_\e$. The scaling $\theta_\e \ll \e |\log \e|$  will be used only after \eqref{eq:almostrecovery} on, where we identify the flat limit of the  discrete vorticity measure. This observation will be useful to study the regimes $\theta_\e \sim \e |\log \e|$ and $\e |\log \e| \ll \theta_\e$ in \cite{Cic-Orl-Ruf}.
\end{remark}
\begin{proof}[Proof of Proposition \ref{p.smoothapprox}]
We divide the proof into several steps. First we define a good approximation very close to the singularity. Then we define an interpolation between this construction and the piecewise constant approximations provided by Lemma \ref{lemma:discretization of smooth wout sing} far from the singularities. In the third step we estimate the energy of this interpolation. In a final step we bound the energy and identify the flat limit of the discrete vorticity measures.

\ul{Step 1} Local discrete approximation of (degree $\pm 1$)-singularities.

\noindent We define a local recovery sequence close to the first singularity $x_1$. In the whole proof we will assume that $\deg(u)(x_1)=1$, so that 
	\begin{equation} \label{eq:def of x/x}
		u(x)=\tfrac{x-x_1}{|x-x_1|} \quad \text{in } B_{\eta_0}(x_1) \,.
		\end{equation}
	(If, instead, $\deg(u)(x_1)=-1$, we have $u(x)= \left(\begin{smallmatrix}
		&1 &0 \\ &0 &- 1
		\end{smallmatrix} \right) \frac{x-x_1}{|x-x_1|}$ and the construction below is adapted accordingly.) For simplification, we specify $u(x_1):=e_1$. Next we partition $\RR^2\setminus\{0\}$ according to the value of the angle in polar coordinates. More precisely, for $k=0,1,\dots,N_{\e}-1$ we set
	\begin{equation}\label{eq:defsector}
	S_{k,\e}:=\{x=r\exp(\iota \varphi)\in\RR^2:\,r>0\, ,\ \varphi\in [k\theta_{\e},(k+1)\theta_{\e})\}\,.
	\end{equation}
	Based on this partition, we approximate the functions $u$ with sequences $v_{\e}\in\mathcal{PC}_{\e}(\S_\e)$ defined on $\e\mathbb{Z}^2\setminus \{x_1\}$ by
	\begin{equation}\label{eq:defbestapproximation}
	v_{\e}(\e i) = \exp(\iota k\theta_{\e})\quad\quad\text{if }\e i-x_1\in S_{k,\e}\,
	\end{equation}	
	while $v_{\e}(\e i)=e_1$ if $\e i = x_1$. Note that $v_\e = \mathfrak{P}_\e(u)$ by definition of $\mathfrak{P}_\e$, see~\eqref{eq:defproj}. Then, writing $\e i-x_1= |\e i-x_1|\exp(\iota (k^{\e}_i\theta_{\e}+\phi^{\e}_i))$ with $\phi^{\e}_i\in[0,\theta_{\e})$, it holds that
	\begin{equation}\label{almostbestapproximation}
	|u(\e i)-v_{\e}(\e i)|\leq|\exp(\iota (k^{\e}_i\theta_{\e}+\phi^{\e}_i))-\exp(\iota k^{\e}_i\theta_{\e})|\leq \theta_{\e}\,,
	\end{equation}
	for $\e i \in B_{\eta_0}(x_1)$. We define the radius $r_{\e}=4 \e\theta_{\e}^{-1}$ (its role will become clear below). We start estimating the energy of $v_{\e}$ in $B_{2r_{\e}}(x_1)$. By a change of variables we may assume that $x_1=0$ in order to reduce the notation. Observe that for any two vectors $a,b\in\RR^2$ we have
	\begin{equation} \label{eq:difference of squares}
	|a|^2-|b|^2\leq |a-b|(|a|+|b|)\leq 2|b||a-b| + |a-b|^2\,.
	\end{equation}
	Hence, using also \eqref{almostbestapproximation},
	\begin{align}\label{eq:avoidprojection}
	\frac{1}{\e\theta_{\e}}E_{\e}(v_{\e};B_{2r_{\e}}) & \leq \frac{1}{\e\theta_{\e}}E_{\e}(u;B_{2r_{\e}})\nonumber
	+\frac{1}{2\e\theta_{\e}} \hspace{-1em} \sum_{\substack{\nn \\ \e i, \e j \in B_{2r_{\e}}}} \hspace{-1em} \e^2\left(\left|v_{\e}(\e i) - v_{\e}(\e j)\right|^2-\left|u(\e i) - u(\e j)\right|^2\right)\nonumber
	\\
	& \leq \frac{1}{\e\theta_{\e}}E_{\e}(u;B_{2r_{\e}})+\frac{C}{\e\theta_{\e}}\hspace{-1em} \sum_{\substack{\nn \\ \e i, \e j \in B_{2r_{\e}}}} \hspace{-1em} \e^2\left(\theta_{\e}|u(\e i)-u(\e j)|+\theta_{\e}^2\right).
	\end{align}
	We prove that the last sum vanishes, while the first right hand side term scales as $2\pi|\log \e|\frac{\e}{\theta_{\e}}$ in the sense that the difference vanishes. To this end, we estimate finite differences of $u$ away from the singularity. As for $t\in[0,1]$ and $i,j\in\ZZ^2$ with $|i-j|=1$ we have
	\begin{equation*}
	|(1-t) \e i+t\e j|\geq |\e i|-\e \, ,
	\end{equation*}
	for any $\e i,\e j\in\e\ZZ^2\setminus B_{2\e}$ with $|i-j|=1$ the regularity of $u$ in $B_{\eta_0}\setminus\{0\}$ implies
	\begin{equation*}
	|u(\e i)-u(\e j)|\leq \int_0^1|\nabla u(t \e i+(1-t)\e j)(\e i-\e j)|\,\mathrm{d}t \, .
	\end{equation*}
	Since $i-j\in\{\pm e_1,\pm e_2\}$, a direct computation yields the two cases
	\begin{equation}\label{eq:estimatecases}
	|u(\e i)-u(\e j)|\leq
	\begin{cases} 
	\int_0^1 \frac{|i\cdot e_2|}{|t i+(1-t) j|^2}\,\mathrm{d}t &\mbox{if $(i-j)\parallel e_1$} \, ,\vspace*{2mm}
	\\
	 \int_0^1 \frac{|i \cdot e_1|}{|t i+(1-t) j|^2}\,\mathrm{d}t &\mbox{if $(i-j)\parallel e_2$} \, ,
	\end{cases}
	\end{equation}

	\begin{figure}[H]
		\hspace{2cm}\scalebox{0.8}{
		\begin{tikzpicture}[scale=0.8] 
\def\l{0.25};

\fill[black!5] (14*\l+0.5*\l,14*\l+0.5*\l) --  (14*\l+0.5*\l,3*\l) -- (3*\l,3*\l) -- (3*\l,14*\l+0.5*\l);

\begin{scope}[yscale=-1]
\fill[black!5] (14*\l+0.5*\l,14*\l+0.5*\l) --  (14*\l+0.5*\l,3*\l) -- (3*\l,3*\l) -- (3*\l,14*\l+0.5*\l);
\end{scope}

\begin{scope}[xscale=-1]
\fill[black!5] (14*\l+0.5*\l,14*\l+0.5*\l) --  (14*\l+0.5*\l,3*\l) -- (3*\l,3*\l) -- (3*\l,14*\l+0.5*\l);
\end{scope}

\begin{scope}[xscale=-1, yscale=-1]
\fill[black!5] (14*\l+0.5*\l,14*\l+0.5*\l) --  (14*\l+0.5*\l,3*\l) -- (3*\l,3*\l) -- (3*\l,14*\l+0.5*\l);
\end{scope}

\fill[black!10] (0,0) circle (3.4cm);

\fill[black!15] (6*\l,6*\l) -- (-6*\l,6*\l) -- (-6*\l, -6*\l) -- (6*\l, -6*\l) -- (6*\l,6*\l);

\fill[black!20] (15*\l,-4*\l) -- (2*\l,-4*\l) -- (2*\l,4*\l) -- (15*\l,4*\l) -- (15*\l,-4*\l);

\foreach \i in {-14,...,14} {
\foreach \j in {-14,...,14} {
	\fill[black!30] (\l*\i,\l*\j) circle(0.5pt);
}
}

\fill (0,0) circle(1pt);

\draw[line width = 0.7pt] (0,0) circle (3.4cm);

\draw[line width = 0.7pt] (6*\l,6*\l) -- (-6*\l,6*\l) -- (-6*\l, -6*\l) -- (6*\l, -6*\l) -- (6*\l,6*\l);

\draw (-6*\l,0) node[anchor=west] {$Q_{6\varepsilon}$};

\draw[line width = 0.7pt] (2*\l,-4*\l) -- (2*\l,4*\l) -- (15*\l,4*\l) 
 (2*\l,-4*\l) -- (15*\l,-4*\l) ;

\draw[line width = 0.7pt] (14*\l+0.5*\l,3*\l) -- (3*\l,3*\l) -- (3*\l,14*\l+0.5*\l);
\draw (11*\l,14*\l) node[anchor=north] {$\varepsilon \mathcal{Q}^{(+,+)}_3$};

\begin{scope}[yscale=-1]
\draw[line width = 0.7pt] (14*\l+0.5*\l,3*\l) -- (3*\l,3*\l) -- (3*\l,14*\l+0.5*\l);
\draw (11*\l,14*\l) node[anchor=south] {$\varepsilon \mathcal{Q}^{(+,-)}_3$};
\end{scope}

\begin{scope}[xscale=-1]
\draw[line width = 0.7pt] (14*\l+0.5*\l,3*\l) -- (3*\l,3*\l) -- (3*\l,14*\l+0.5*\l);
\draw (11*\l,14*\l) node[anchor=north] {$\varepsilon \mathcal{Q}^{(-,+)}_3$};
\end{scope}

\begin{scope}[xscale=-1, yscale=-1]
\draw[line width = 0.7pt] (14*\l+0.5*\l,3*\l) -- (3*\l,3*\l) -- (3*\l,14*\l+0.5*\l);
\draw (11*\l,14*\l) node[anchor=south] {$\varepsilon \mathcal{Q}^{(-,-)}_3$};
\end{scope}

\draw (0,0) -- (0,-3.4);
\draw (0,3.4) node[anchor=north] {$B_{2 r_\varepsilon}$};
\draw (0,-2.8) node[anchor=west] {$2 r_\varepsilon$};

\foreach \k in {2,...,14} {
	\draw (\l*\k+0.05,\l/2-0.03) -- (\l*\k+0.05,\l/2+0.03);
	\draw (\l*\k+0.05,\l/2) -- (\l*\k+\l-0.05,\l/2);
	\draw (\l*\k+\l-0.05,\l/2-0.03) -- (\l*\k+\l -0.05,\l/2+0.03);
	\draw[<-,line width=0.1pt, black!80] (\l*\k + \l/2,\l/2+0.05) to[out=90,in=180] (4.5,2);
}
\draw (4.5,2) node[anchor=west] {$\sim \big\lceil \frac{2r_\varepsilon}{\varepsilon} \big\rceil$};
		\end{tikzpicture}
		}
	\caption{The trimmed quadrants $\e \mathcal{Q}^s_3$ used to bound the energy in~\eqref{eq:splitbytrim}. }
	\label{fig:trimming}
	\end{figure}

	To further simplify the energy, given a sign $s=(s_1,s_2)\in\{(+,+), (-,+), (-,-), (+,-)\}$ and $n\in\mathbb{N}$ we define the trimmed quadrants $\mathcal{Q}_n^{s}$ as
	\begin{equation} \label{eq:trimmed quadrants}
	\mathcal{Q}_n^{s}:=\{x\in\RR^2:s_1 \, x \cdot e_1 \geq n,\ s_2\,  x\cdot e_2 \geq n\}\,.
	\end{equation}
	Applying Jensen's inequality in \eqref{eq:estimatecases} and specifying $n=3$, we can bound the energy via
	\begin{equation}\label{eq:splitbytrim}
	\frac{1}{\e\theta_{\e}}E_{\e}(u;B_{2r_{\e}})\leq \sum_s\frac{1}{\e\theta_{\e}}E_{\e}(u,\e\mathcal{Q}_3^s\cap B_{2r_{\e}})+\frac{1}{\e\theta_{\e}}E_{\e}(u;Q_{6\e})
	+\frac{C}{\e\theta_{\e}}\sum_{k=2}^{\lceil 2r_{\e}/\e\rceil}\e^2 \frac{k^2}{(k-1)^4}\,,
	\end{equation}
	see Figure~\ref{fig:trimming}. The last sum is converging with respect to $k$, so that the second and third term can be estimated by
	\begin{equation}\label{eq:minorcontributions}
	\frac{1}{\e\theta_{\e}}E_{\e}(u;Q_{6\e})
	+\frac{C}{\e\theta_{\e}}\sum_{k=2}^{\lceil 2r_{\e}/\e\rceil}\e^2 \frac{k^2}{(k-1)^4}\leq C\frac{\e}{\theta_{\e}} \, ,
	\end{equation}
	where for the first term we used the estimate $|u(\e i)-u(\e j)|^2\leq 4$. On the trimmed quadrants $\mathcal{Q}_3^s$ we can use again Jensen's inequality in \eqref{eq:estimatecases} and a monotonicity argument to deduce
	\begin{align}\label{eq:trimtrick}
	E_{\e}(u;\e\mathcal{Q}_3^s\cap B_{2r_{\e}})& =\sum_{\substack{\e i\in\e\ZZ^2\cap B_{2r_{\e}}\\ i\in \mathcal{Q}_3^s}}\!\!\! \e^2|u(\e(i+s_1e_1))-u(\e i)|^2+|u(\e(i+s_2e_2))-u(\e i)|^2\nonumber
	\\
	& \leq\sum_{\substack{\e i\in\e\ZZ^2\cap B_{2r_{\e}}\\ i\in \mathcal{Q}_3^s}}\e^2 \frac{\e^2}{|\e i|^2}\leq \integral{\e\mathcal{Q}_{2}^s\cap B_{2r_{\e}}}{\frac{\e^2}{|x|^2}}{\,\mathrm{d}x} \, .
	\end{align}
	Note that we shifted the trimming in the last inequality to pass from discrete to continuum. We sum \eqref{eq:trimtrick} over all possible $s$ and, after multiplying with $\frac{1}{\e\theta_{\e}}$, we infer that
	\begin{equation}\label{eq:integralfree}
	\sum_s\frac{1}{\e\theta_{\e}}E_{\e}(u;\e\mathcal{Q}_3^s\cap B_{2r_{\e}})\leq \frac{1}{\e\theta_{\e}}\integral{B_{2r_{\e}}\setminus B_{\e}}{\frac{\e^2}{|x|^2}}{\,\mathrm{d}x}\leq 2\pi |\log \e|\frac{\e}{\theta_{\e}}\, ,
	\end{equation}
	since $r_{\e}=4\e\theta_{\e}^{-1} < 1$. 
	The combination of \eqref{eq:splitbytrim}, \eqref{eq:minorcontributions}, and \eqref{eq:integralfree} yields
	\begin{equation}\label{eq:XYprecise}
	\frac{1}{\e\theta_{\e}}E_{\e}(u;B_{2r_{\e}})-2\pi |\log \e|\frac{\e}{\theta_{\e}}\leq C\frac{\e}{\theta_{\e}} \, .
	\end{equation}
	To bound the remaining sum in \eqref{eq:avoidprojection}, note that on the one hand $r_{\e}=4\e\theta_{\e}^{-1}$ implies that
	\begin{equation}\label{eq:r_econtribution}
	\frac{1}{\e\theta_{\e}}\sum_{\substack{\nn \\ \e i, \e j \in B_{2r_{\e}}}}\e^2\theta_{\e}^2\leq C\frac{\theta_{\e}}{\e}(2r_{\e}+2\e)^2\leq C\frac{\e}{\theta_{\e}}.
	\end{equation}
	On the other hand, using the trimmed quadrants $\mathcal{Q}_3^s$ to split the sum as in \eqref{eq:splitbytrim} yields  
	\begin{align*}
	\frac{1}{\e\theta_{\e}}\!\!\!\!\! \sum_{\substack{\nn \\ \e i, \e j \in B_{2r_{\e}}}}\!\!\!\!\! \e^2\theta_{\e}|u(\e i)-u(\e j)|\leq \frac{ \sqrt{2}}{\e\theta_{\e}}\integral{B_{2r_{\e}}\setminus B_{\e}}{\theta_{\e}\frac{\e}{|x|}}{\,\mathrm{d}x}+C\e+\frac{C}{\e\theta_{\e}}\sum_{k=2}^{\lceil 2r_{\e}/\e\rceil}\e^2\theta_{\e}\frac{k}{(k-1)^2}.
	\end{align*}
	Note that due to the non-quadratic structure we have the additional constant $\sqrt{2}$ in front of the integral. Moreover, the last sum diverges logarithmically, but we have an additional factor $\theta_{\e}$ which compensates this growth. We conclude that
	\begin{equation*}
	\frac{1}{\e\theta_{\e}} \!\!\!\!\! \sum_{\substack{\nn \\ \e i, \e j \in B_{2r_{\e}}}}\!\!\!\!\! \e^2\theta_{\e}|u(\e i)-u(\e j)|\leq C\left(r_{\e}^2+\e+\e|\log \theta_{\e}|\right).
	\end{equation*}
	Since $\e\ll\theta_{\e}\ll 1$ the right hand side vanishes when $\e\to 0$. Thus this estimate, \eqref{eq:avoidprojection}, \eqref{eq:XYprecise}, and \eqref{eq:r_econtribution} imply that
	\begin{equation}\label{eq:exactconcentration}
	\limsup_{\e \to 0}\left(\frac{1}{\e\theta_{\e}}E_{\e}(v_{\e};B_{2r_{\e}})-2\pi|\log \e|\frac{\e}{\theta_{\e}}\right)\leq 0 \, . 
	\end{equation}
	
	Next we control the energy in $B_{\eta}(x_1)\setminus B_{r_{\e}}(x_1)$ for $0<\eta<\eta_0$, where $\eta_0$ is given by the assumptions. To this end, we need to examine the precise behavior of the sequence $v_{\e}$ for $i,j\in\ZZ^2$ satisfying $|i-j|=1$ and $|\e j-x_1|,|\e i-x_1|\geq r_{\e}$. The basic idea is that for many such pairs the energy contribution vanishes. Indeed, write such points as
	\begin{equation}\label{polarcoordinates}
	\e i-x_1 = r_i^{\e}\exp(\iota (k^{\e}_i\theta_{\e}+\phi^{\e}_i)),\quad\quad\e j-x_1 = r_j^{\e}\exp(\iota (k^{\e}_j\theta_{\e}+\phi^{\e}_j))
	\end{equation}
	with $k^{\e}_i,k^{\e}_j\in\{0,\dots,N_{\e}-1\}$ and $\phi^{\e}_i,\phi^{\e}_j\in[0,\theta_{\e})$. By \eqref{eq:geo and eucl} we obtain
	\begin{align*}
	\e & = \big|r_i^{\e}\exp(\iota (k^{\e}_i\theta_{\e}+\phi^{\e}_i))-r_j^{\e}\exp(\iota (k^{\e}_j\theta_{\e}+\phi^{\e}_j))\big|
	\\
	& \geq r_{\e}|\exp(\iota (k^{\e}_i\theta_{\e}+\phi^{\e}_i))-\exp(\iota (k^{\e}_j\theta_{\e}+\phi^{\e}_j))|-|r_i^{\e}-r_j^{\e}|
	\\
	& \geq \min_{n\in\{0,\pm 1\}}\frac{2r_{\e}}{\pi}| (k^{\e}_i-k_j^{\e})\theta_{\e}+\phi^{\e}_i-\phi^{\e}_j+2\pi n|-\e
	\\
	& \geq \min_{n\in\{0,\pm 1\}}\frac{2r_{\e}\theta_{\e}}{\pi}(| k^{\e}_i- k^{\e}_j+N_{\e} n|-1)-\e\,.
	\end{align*}
	Inserting $r_{\e}=4\e\theta_{\e}^{-1}$, the above estimate can be rearranged into
	\begin{equation*}
	\frac{\pi}{4}\geq\min_{n\in\{0,\pm 1\}}(| k^{\e}_i- k^{\e}_j+N_{\e} n|-1)\,. 
	\end{equation*}
	Since $k^{\e}_i- k^{\e}_j+N_{\e} n$ is an integer, we get the following two possibilities:
	\begin{itemize}
		\item[(i)] $k^{\e}_i-k^{\e}_j-N_{\e}n=0$, which is possible only for $n=0$, since $k^{\e}_i, k^{\e}_j \in \{0,\dots, N_\e -1\}$. This yields $v_{\e}(\e i)=v_{\e}(\e j)$ since $\e i$ and $\e j$ belong to the same sector;
		\item[(ii)] $k^{\e}_i-k^{\e}_j=\pm 1\mod(N_{\e})$, which implies $|v_{\e}(\e i)-v_{\e}(\e j)|\leq\theta_{\e}$. Moreover, since $k^{\e}_i\neq k^{\e}_j$ we infer that $\dist(\e i-x_1,\partial S_{k_i^{\e},\e})\leq \e$.
	\end{itemize}
With this information at hand, we can estimate the energy by bounding the number of all points in $\e\ZZ^2\cap B_{\eta}(x_1)$ which are $\e$-close to one of the lines in $\bigcup_{k=0}^{N_{\e}-1}\partial S_{k,\e}+x_1$. Since $N_{\e}\leq C\theta_{\e}^{-1}$, this leads to
\begin{align*}
E_{\e}(v_{\e};B_{\eta}(x_1)\setminus B_{r_{\e}}(x_1))&\leq C\theta_{\e}^2\sum_{k=0}^{N_{\e}-1}\e^2\#\{\e i\in\e\ZZ^2\cap B_{\eta}(x_1):\,\dist(\e i,\partial S_{k,\e}+x_1)\leq\e\}
\\
&\leq C\theta_{\e} (\eta+2\e)\e\,.
\end{align*}
Dividing the inequality by $\e\theta_{\e}$ we obtain for $\e$ small enough that
\begin{align*}
\frac{1}{\e\theta_{\e}}E_{\e}(v_{\e};B_{\eta}(x_1))&\leq\frac{1}{\e\theta_{\e}}E_{\e}(v_{\e};B_{2r_{\e}}(x_1))+\frac{1}{\e\theta_{\e}}E_{\e}(v_{\e};B_{\eta}(x_1)\setminus B_{r_{\e}}(x_1))
\\
&\leq\frac{1}{\e\theta_{\e}}E_{\e}(v_{\e};B_{2r_{\e}}(x_1))+C\eta\,,
\end{align*}
where we used that $r_{\e}\gg\e$ to split the energy via changing the inner radius from $r_{\e}$ to $2r_{\e}$. Subtracting the term $2\pi|\log \e|\tfrac{\e}{\theta_{\e}}$ and using \eqref{eq:exactconcentration}, we proved that for some $C<+\infty$
\begin{equation}\label{eq:step1final}
\limsup_{\e \to 0}\left(\frac{1}{\e\theta_{\e}}E_{\e}(v_{\e};B_{\eta}(x_1))-2\pi|\log \e|\frac{\e}{\theta_{\e}}\right)\leq C\eta\,.
\end{equation} 

\medskip 
\ul{Step 2} An interpolation between singular and piecewise constant approximations.

\noindent  We do the construction in the case where the singularity lies in the origin. The case of singularities contained in $\lambda \ZZ^2$ will be treated with a translation argument. Consider a cube ${Q}({\lambda})=[-2^{m(\lambda)}\lambda,2^{m(\lambda)}\lambda]^2$, where $\lambda=\lambda_k$ with $k\geq n$ will be small, but fixed in this step, and $1\ll m(\lambda)\in\mathbb{N}$ is chosen maximal such that $Q(\lambda)\subset B_{\eta/2}$ with fixed $0<\eta<\eta_0$. Note that the corners of $Q(\lambda)$ belong to $\lambda\ZZ^2$. Define then a sequence of dyadically shrinking cubes by $Q_{k}= [(-2^{m(\lambda)}+(2-2^{-k}))\lambda,(2^{m(\lambda)}-(2-2^{-k}))\lambda]^2$ for $k\geq 0$. Here the factor $2-2^{-k}$ is chosen as the value of the geometric sum $\sum_{l=0}^k2^{-l}$. For notational reasons we also set $Q_{-1}:=Q(\lambda)$ and $Q_{-2}:=[-(2^{m(\lambda)}+1)\lambda,(2^{m(\lambda)}+1)\lambda]$. Then for $k\geq 0$ the layer $L_{k}=\overline{Q_{k-1}\setminus Q_{k}}$ can be decomposed into finitely many closed cubes with disjoint interior and side lengths $2^{-k}\lambda$. Indeed, those cubes are given by the closures of the half-open cubes belonging to the family  
\begin{equation*}
\mathcal{Q}_{k} := \left\{\mathfrak{q}^z_k = \left(2^{-k}\lambda z+[0,2^{-k}\lambda)^2\right) \ : \  z \in \ZZ^2 \, , \ \mathfrak{q}^z_k \subset L_{k}  \right\}\,.
\end{equation*}
With a slight abuse of notation we also set
\begin{equation*}
\mathcal{Q}_{-1}:=\left\{\mathfrak{q}^z_{-1} = \left(\lambda z+[0,\lambda)^2\right) \ : \  z \in \ZZ^2  \, , \ \mathfrak{q}^z_{-1} \subset  L_{-1} \right\}\,.
\end{equation*}
For $k \geq -1$, a generic element of $\mathcal{Q}_k$ is of the form 
\begin{equation*}
\q^z_k = 2^{-\max\{k,0\}}\lambda z+[0,2^{-\max\{k,0\}}\lambda)^2 \subset L_{k}\,.
\end{equation*}
(See Figure~\ref{fig:dyadic decomposition}.) We introduced the square $Q_{-2}$ and the family of cubes $\mathcal{Q}_{-1}$ since they will be useful later to glue in the layer $L_{-1} = \ol{Q_{-2} \sm Q_{-1}}$ the construction of the recovery sequence $u_\e$ inside $Q(\lambda)$ and outside $Q(\lambda)$, respectively. The construction of $u_\e$ outside $Q(\lambda)$ will be based, as in Proposition~\ref{prop:construction of ueps}, on a piecewise constant approximation of $u$ on the $\lambda \ZZ^2$ lattice and its boundary value on $\de Q(\lambda)$ will agree with that of the construction from the inside. For this reason the cubes in $\mathcal{Q}_{-1}$ have volume $\lambda^2$, as those of $\mathcal{Q}_0$, instead of the notationally more consistent  volume~$(2\lambda)^2$.

We further choose $k_{\e}\in\mathbb{N}$ as the unique number such that \begin{equation}\label{eq:choicek_e}
2^{-k_{\e}} \leq\theta_{\e} <2^{-k_{\e}+1} \,.
\end{equation}
Note that, in particular, we have that
\begin{equation}\label{eq:minimalhole}
Q_{k_{\e}}\supset B_{(2^{m(\lambda)}-2)\lambda}\,.
\end{equation}

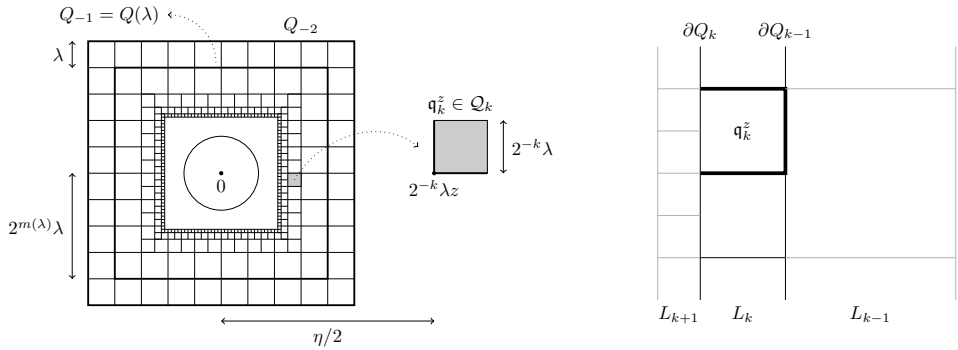
\begin{figure}[H] 
\scalebox{0.7}{
\begin{tikzpicture}
	\def\n{8}
	\def\l{1/2}
	\draw[line width=1pt] (0,0) rectangle (\n*\l,\n*\l);
	\draw[line width=1pt] (-\l,-\l) rectangle (\n*\l+\l,\n*\l+\l);
	\draw (3.5,4.5) node[anchor=south] {$Q_{-2}$};
	
	\draw[->,dotted] (2-0.1,4+0.1) to[out=90,in=0] (1,5);
	\draw (1,5) node[anchor=east] {$Q_{-1} = Q(\lambda)$};
	
	\foreach \y in {0,...,8}
		\draw (-\l,\y*\l) -- (0,\y*\l);
	
	\draw (\l,0) -- (\l,\n*\l);
	\foreach \y in {1,...,7}
		\draw (0,\y*\l) -- (\l,\y*\l);
	\draw (\l+\l/2,\l) -- (\l+\l/2,\n*\l-\l);
	\foreach \y in {2,...,14}
		\draw (\l,\y*\l/2) -- (\l + \l/2,\y*\l/2);
	\draw (\l+\l/2+\l/4,\l+\l/2) -- (\l+\l/2+\l/4,\n*\l-\l-\l/2);
	\foreach \y in {6,...,26}
		\draw (\l+\l/2,\y*\l/4) -- (\l + \l/2 + \l/4,\y*\l/4);
	\draw (\l+\l/2+\l/4+\l/8,\l+\l/2+\l/4) -- (\l+\l/2+\l/4+\l/8,\n*\l-\l-\l/2-\l/4);
	\foreach \y in {14,...,50}
		\draw (\l+\l/2+\l/4,\y*\l/8) -- (\l + \l/2 + \l/4 + \l/8,\y*\l/8);
		
	\begin{scope}[xshift=4cm,xscale=-1]
	\foreach \y in {0,...,8}
		\draw (-\l,\y*\l) -- (0,\y*\l);
	\draw (\l,0) -- (\l,\n*\l);
	\foreach \y in {1,...,7}
		\draw (0,\y*\l) -- (\l,\y*\l);
	\draw (\l+\l/2,\l) -- (\l+\l/2,\n*\l-\l);
	\foreach \y in {2,...,14}
		\draw (\l,\y*\l/2) -- (\l + \l/2,\y*\l/2);
	\draw (\l+\l/2+\l/4,\l+\l/2) -- (\l+\l/2+\l/4,\n*\l-\l-\l/2);
	\foreach \y in {6,...,26}
		\draw (\l+\l/2,\y*\l/4) -- (\l + \l/2 + \l/4,\y*\l/4);
	\draw (\l+\l/2+\l/4+\l/8,\l+\l/2+\l/4) -- (\l+\l/2+\l/4+\l/8,\n*\l-\l-\l/2-\l/4);
	\foreach \y in {14,...,50}
		\draw (\l+\l/2+\l/4,\y*\l/8) -- (\l + \l/2 + \l/4 + \l/8,\y*\l/8);
	\end{scope}
	
	\begin{scope}[xshift=4cm,rotate=90]
	\foreach \y in {0,...,8}
		\draw (-\l,\y*\l) -- (0,\y*\l);
	\draw (\l,0+\l) -- (\l,\n*\l-\l);
	\foreach \y in {2,...,6}
		\draw (0,\y*\l) -- (\l,\y*\l);
	\draw (\l+\l/2,\l+\l/2) -- (\l+\l/2,\n*\l-\l-\l/2);
	\foreach \y in {4,...,12}
		\draw (\l,\y*\l/2) -- (\l + \l/2,\y*\l/2);
	\draw (\l+\l/2+\l/4,\l+\l/2+\l/4) -- (\l+\l/2+\l/4,\n*\l-\l-\l/2-\l/4);
	\foreach \y in {8,...,24}
		\draw (\l+\l/2,\y*\l/4) -- (\l + \l/2 + \l/4,\y*\l/4);
	\draw (\l+\l/2+\l/4+\l/8,\l+\l/2+\l/4+\l/8) -- (\l+\l/2+\l/4+\l/8,\n*\l-\l-\l/2-\l/4-\l/8);
	\foreach \y in {16,...,48}
		\draw (\l+\l/2+\l/4,\y*\l/8) -- (\l + \l/2 + \l/4 + \l/8,\y*\l/8);
	\end{scope}
	
	\begin{scope}[xshift=0cm,yshift=4cm,rotate=-90]
	\foreach \y in {0,...,8}
		\draw (-\l,\y*\l) -- (0,\y*\l);
	\draw (\l,0+\l) -- (\l,\n*\l-\l);
	\foreach \y in {2,...,6}
		\draw (0,\y*\l) -- (\l,\y*\l);
	\draw (\l+\l/2,\l+\l/2) -- (\l+\l/2,\n*\l-\l-\l/2);
	\foreach \y in {4,...,12}
		\draw (\l,\y*\l/2) -- (\l + \l/2,\y*\l/2);
	\draw (\l+\l/2+\l/4,\l+\l/2+\l/4) -- (\l+\l/2+\l/4,\n*\l-\l-\l/2-\l/4);
	\foreach \y in {8,...,24}
		\draw (\l+\l/2,\y*\l/4) -- (\l + \l/2 + \l/4,\y*\l/4);
	\draw (\l+\l/2+\l/4+\l/8,\l+\l/2+\l/4+\l/8) -- (\l+\l/2+\l/4+\l/8,\n*\l-\l-\l/2-\l/4-\l/8);
	\foreach \y in {16,...,48}
		\draw (\l+\l/2+\l/4,\y*\l/8) -- (\l + \l/2 + \l/4 + \l/8,\y*\l/8);
	\end{scope}
	
	\draw[fill=black] (2,2) circle(0.03);
	\draw (2,2) node[anchor=north] {$0$};
	
	\draw[<->] (-\l -0.3,-\l+0.5) -- (-\l-0.3,2);
	\draw (-\l -0.3,1) node[anchor=east] {$2^{m(\lambda)} \lambda$};
	\draw[<->] (-\l -0.3,4+\l) -- (-\l-0.3,4+\l-1/2);
	\draw (-\l -0.3,4+\l-1/4) node[anchor=east] {$\lambda$};
	\draw[<->] (2,-\l-0.3) -- (6,-\l-0.3);
	\draw (2+2,-\l-0.3) node[anchor=north] {$\eta/2$};

	\draw[fill=black!20] (4-1/2-1/4,2-1/4) rectangle (4-1/2,2);
	\draw[->,dotted] (4-1/2-1/4+1/8,2-1/4+1/8) to[out=0+60,in=180-40] (6-0.3,2+0.5);
	
	\draw[fill=black!20] (6,2) rectangle (7,3);
	\draw[line width=1pt] (6,2) -- (7,2)
	(6,2) -- (6,3);
	\draw (6.5,3) node[anchor=south] {$\q_k^z \in \mathcal{Q}_k$};
	
	\draw[fill=black] (6,2) circle(0.03);
	\draw (6,2)  node[anchor=north] {$2^{-k}\lambda z$};
	\draw[<->] (7+0.3,3) -- (7+0.3,2);
	\draw (7+0.3,2.5) node[anchor=west] {$2^{-k}\lambda$};

	\draw (2,2) circle (0.7cm);

	\begin{scope}[xshift=11cm, yshift=2cm, scale=0.8]
		\draw 
			(0,-3) -- (0,3)
			(2,-3) -- (2,3)
			(0,-2) -- (2,-2);
			\draw (1,-3) node[anchor=north] {$L_k$};
			\draw (0,3) node[anchor=south] {$\partial Q_k$};

			\draw[line width=2pt]
			(0,0) -- (2,0) -- (2,2) -- (0,2);
			\draw (1,1) node {$\q_k^z$};

			\draw[color=black!30]
			(6,-3) -- (6,3)
			(2,2) -- (6,2)
			(2,-2) -- (6,-2);
			\draw (4,-3) node[anchor=north] {$L_{k-1}$};
			\draw (2,3) node[anchor=south] {$\partial Q_{k-1}$};

			\draw[color=black!30]
			(-1,-3) -- (-1,3)
			(-1,-2) -- (0,-2)
			(-1,-1) -- (0,-1)
			(-1,0) -- (0,0)
			(-1,1) -- (0,1)
			(-1,2) -- (0,2);
			\draw (-0.5,-3) node[anchor=north] {$L_{k+1}$};
	\end{scope}
\end{tikzpicture}	
}
\caption{On the left: Dyadic decomposition of $Q(\lambda)$ and example of a square belonging to the family~$\mathcal{Q}_k$ (in the picture, $k=1$). The ball contained in all squares $Q_k$ is given by~\eqref{eq:minimalhole}. On the right: Sides of a cube $\q_k^z$ contained in the layer $L_k$ where we define the boundary conditions.}
\label{fig:dyadic decomposition}
\end{figure}

To each (non-empty, half-open) cube $\q_{k}^z=2^{-\max\{k,0\}}\lambda z+[0,2^{-\max\{k,0\}}\lambda)^2\cap L_{k}\in\mathcal{Q}_{k}$, we associate the value 
\begin{equation*}
u_{k,\e}^{z}=u(2^{-\max\{k,0\}}\lambda (z+\tfrac{1}{2}e_1+\tfrac{1}{2}e_2)) \, ,
\end{equation*}
where $2^{-\max\{k,0\}}\lambda (z+\tfrac{1}{2}e_1+\tfrac{1}{2}e_2)$ is the midpoint of the cube $\q_k^z$. We use these values $u_{k,\e}^{z}$ to define an interpolation similar to the one in the proof of Proposition~\ref{prop:construction of ueps}, but on a family of shrinking cubes. In order to obtain quantitative energy estimates, we need a bound on the differences of the values $u_{k,\e}^z$ between cubes which touch at their boundaries. A key ingredient will be the estimate
\begin{equation}\label{eq:xmodxcontinuous}
\left|\frac{x}{|x|}-\frac{y}{|y|}\right|\leq \frac{\big|x|y|-x|x|+x|x|-y|x|\big|}{|x||y|}\leq 2\frac{|x-y|}{|y|} \, ,
\end{equation}
which is valid for $x,y\in\RR^2\setminus\{0\}$. Due to \eqref{eq:minimalhole}  it holds that $0 \notin \q_k^z$ for $-1\leq k\leq k_{\e}$.
Hence, for two touching cubes $\q_{k_1}^{z_1}$ and $\q_{k_2}^{z_2}$ with $-1\leq k_1, k_2 \leq k_{\e}$ (i.e., $\ol \q_{k_1}^{z_1} \cap \ol \q_{k_2}^{z_2} \neq \emptyset$), the estimate \eqref{eq:xmodxcontinuous} implies the bound
\begin{equation}\label{eq:closeness1}
|u_{k_1,\e}^{z_1}-u_{k_2,\e}^{z_2}|\leq 2\frac{  \sqrt{2} \lambda (2^{-k_1}+2^{-k_2})}{\displaystyle\min_{l=1,2} |2^{-\max\{k_l,0\}}\lambda(z_l+\tfrac{1}{2}e_1+\tfrac{1}{2}e_2)|} \, ,
\end{equation}
where we used that the distance between midpoints is bounded by the sum of the diameters of the cubes. Assuming that $m(\lambda)\geq 2$, the inclusion \eqref{eq:minimalhole} implies that the denominator can be estimated from below via
\begin{equation*}
\left|2^{-\max\{k_l,0\}}\lambda(z_l+\tfrac{1}{2}e_1+\tfrac{1}{2}e_2)\right|\geq (2^{m(\lambda)}-2)\lambda\geq 2^{m(\lambda)-1}\lambda\,.
\end{equation*}
In combination with \eqref{eq:closeness1} and the bound $2^{-k_{\e}}\leq\theta_{\e}$ (cf. \eqref{eq:choicek_e}) we obtain
\begin{equation}\label{eq:closeness2}
|u_{k_1,\e}^{z_1}-u_{k_2,\e}^{z_2}|\leq 2^{ 3-m(\lambda)} (2^{-k_1}+2^{-k_2})\leq 16\max_{l=1,2}2^{k_{\e}-k_l-m(\lambda)}\theta_{\e} \, .
\end{equation}

Next, we define the piecewise constant function $\bar{w}_{\e} \colon Q_{-2}\setminus Q_{k_{\e}}\to\S_{\e}$ via
\begin{align}\label{eq:againPC}
\bar{w}_{\e}(x)=
u_{k,\e}^z &&\mbox{if $x\in \left(2^{-\max\{k,0\}}\lambda z+[0,2^{-\max\{k,0\}}\lambda)^2\right)\cap L_k \, ,\,-1\leq k\leq k_{\e}$} \, .  
\end{align} 
Note that this function is pointwise well-defined except on parts of $\partial Q_{-2}$ since we consider half-open cubes. In order to define an interpolation between cubes which approximates the piecewise constant function $\bar{w}_{\e}$, we introduce again boundary conditions. In each cube~$\q_k^z \in \mathcal{Q}_k$, we define the boundary conditions only on those sides which are not contained in $\partial Q_k$  (recall that $\partial Q_k$ is the inner part of the boundary of the layer $L_k$). On the side contained in $\partial Q_k$ (if there is any) we define the boundary condition via the cubes in $L_{k+1}$ (cf.\ Figure~\ref{fig:dyadic decomposition}). To fix ideas, in what follows one can use an iterative definition starting with $k=k_{\e}$, for which we neglect the inner boundary. 






For a generic side $S=\{2^{-k}\lambda z^{\prime}+te_i:\,t\in[0,2^{-k}\lambda]\}$ with $z^{\prime}\in\mathbb{Z}^2$, $k\geq 0$ and $i\in\{1,2\}$, and three values $w=(w^1,w^2,w^3)\in(\SS^1)^3$, we set $b^{\e}_{S,k}[w]\colon S\to\SS^1$ as
\begin{equation*}
b^{\e}_{S,k}[w](x)=b^{\e}_{2^{k}S}[w](2^{k}x) \, ,
\end{equation*}
where $b^{\e}_{S^{\prime}}[w]$ is defined in~\eqref{eq:defbcases} for every side $S^{\prime}=\{\lambda z^{\prime}+te_i:\,t\in[0,\lambda]\}$. In this proof we work with the constant $c_0 := 393$ in \eqref{eq:defbcases}; this choice will be clear only after formula~\eqref{eq:b will be constant}. Given a side~$S$ as above satisfying additionally $S\subset L_k$ with $k\geq 0$ and $S\nsubseteq\partial Q_k$ (recall that the layer~$L_k$ is closed), we specify the three values $w=w^{\e}_{S}$ on $S$ by
\begin{equation}\label{eq:defv_S}
w^{\e}_{S}=\Big(\bar{w}_{\e}(2^{-k}\lambda z^{\prime}),{\rm mid}\big((\bar{w}_{\e})^-_{S},(\bar{w}_{\e})^+_{S}),\bar{w}_{\e}(2^{-k}\lambda(z^{\prime}+e_i) \big) \Big) \, ,
\end{equation}
where $(\bar{w}_{\e})^-_{S}$ and $(\bar{w}_{\e})^+_{S}$ denote the (constant) traces along the side $S$ of the function~$\bar{w}_{\e}$ defined in \eqref{eq:againPC} (note that on sides in $\partial Q_k$ the trace from outside $L_k$ may be non-constant because the cubes shrink). It is only here where we have to use the values in the layer $L_{-1}$. 

Fix a cube $\q_k^z\in\mathcal{Q}_k$ ($k\geq 0$) and define the boundary values $b_{k,\e}[z]\colon \partial \q_k^z\setminus\partial Q_k\to\SS^1$ by
\begin{align*}
b_{k,\e}[z](x)=
b_{S,k}^{\e}[w^{\e}_{S}](x) \quad\mbox{if $x=2^{-k}\lambda z^{\prime}+t e_i \in S$ for some $z^{\prime}\in\mathbb{Z}^2$ and $t\in[0,2^{-k}\lambda]$} \, .
\end{align*} 
Having in mind the definition \eqref{eq:defbcases}, on each side $S$ the function $b_{k,\e}[z]$ satisfies the Lipschitz-estimate
\begin{align}\label{eq:quantitativeLip}
\big|b_{k,\e}[z](x)-b_{k,\e}[z](y)\big|&=\big|b^{\e}_{2^kS}[w_S^{\e}](2^k x)-b^{\e}_{2^kS}[w_S^{\e}](2^k y)\big|\nonumber
\\
& \leq  \max_{i=1,3}\geo\big((w_S^{\e})^i,(w_{S}^{\e})^2\big)\frac{2^k\theta_{\e}}{c_0 \e}|x-y|\nonumber
\\
& \leq \max_{i=1,3}\big|(w_S^{\e})^i-(w_{S}^{\e})^2\big| \frac{\pi 2^{k-1}}{c_0} \frac{\theta_{\e}}{\e} |x-y| \, ,
\end{align} 
where we used \eqref{eq:geo and eucl} in the last inequality. We continue with estimating the right hand side of~\eqref{eq:quantitativeLip}. On the one hand, equation \eqref{eq:closeness2} implies
\begin{align}\label{eq:midpointest}
|(w^{\e}_S)^{2} -(\bar w_{\e})^{\pm}_S| &\leq\geo\left({\rm mid}((\bar w_{\e})^+_S,(\bar w_{\e})^{-}_S),(\bar w_{\e})^{\pm}_S\right)=\frac{1}{2}\geo\left((\bar w_{\e})^+_S,(\bar w_{\e})^{-}_S\right)\nonumber
\\
& \leq \frac{\pi}{4}|(\bar w_{\e})^+_S-(\bar w_{\e})^{-}_S|\leq C_1  2^{k_{\e}-k-m(\lambda)} \theta_{\e} \, ,
\end{align}
with $C_1 = 32$, where in the last inequality we also used that due to the definition \eqref{eq:againPC} and the fact that $S\subset \partial \q_k^z$ we have $(\bar w_{\e})^{\pm}_S=u_{k_{\pm},\e}^{z_{\pm}}$ for some $z_{\pm}\in\ZZ^2$ and $-1\leq k_{\pm}\leq k_{\e}$ with $|k-k_{\pm}|\leq 1$. On the other hand, observe that in the definition of $w_S^{\e}$ in \eqref{eq:defv_S} the points $2^{-k}\lambda z^{\prime}$ and $2^{-k}\lambda (z^{\prime} +e_i)$ belong to $S$. Thus the cubes $\q_{k_1}^{z_1}$ and $\q_{k_3}^{z_3}$ used in the definition~\eqref{eq:againPC} for $(w_S^{\e})^1 = \bar w_\e(2^{-k} \lambda z')$ and~$(w_S^{\e})^3 = \bar w_\e(2^{-k} \lambda (z' + e_i))$, respectively, must touch both the cubes $\q_{k_{\pm}}^{z_{\pm}}$ used in the definition~\eqref{eq:againPC} for~$(\bar w_{\e})^{\pm}_S$. Hence, again due to~\eqref{eq:closeness2},
\begin{align}\label{eq:v_Sestimate} 
|(w_S^{\e})^i-(\bar w_{\e})^{\pm}_S|\leq C_1  2^{k_{\e}-k-m(\lambda)} \theta_{\e} \, ,\quad\quad \text{for }i=1,3 \, .
\end{align}
Combining the last two estimates with \eqref{eq:quantitativeLip} and the bound $2^{k_{\e}}\theta_{\e}\leq 2$ yields
\begin{equation}\label{eq:LiponS}
\left|b_{k,\e}[z](x)-b_{k,\e}[z](y)\right|\leq \frac{\theta_{\e}}{\e}|x-y| \, ,
\end{equation}
where we used that $2 C_1  2^{k_\e - k - m(\lambda)}  \theta_\e \frac{\pi 2^{k-1}}{c_0} \leq 1$. Next observe that the locally defined boundary values yield a function 
\begin{align}\label{eq:bglobally}
&b_{\e} \colon \bigcup_{0\leq k\leq k_{\e}}\bigcup_{\q_k^z\in\mathcal{Q}_k}\partial \q_k^z\setminus \partial Q_{k_{\e}}\to\SS^1, &x\mapsto b_\e(x) := b_{k,\e}[z](x)\quad\text{if }x\in\partial \q_{k}^z\setminus\partial Q_k \, .
\end{align}

We briefly explain the idea how to construct the recovery sequence. In $Q_{k_{\e}}$ we put the value of the function $v_{\e}$ used in Step~1 and defined in~\eqref{eq:defbestapproximation}, namely, we approximate $\tfrac{x}{|x|}$ close to its singularity. In the first layer $L_{k_{\e}}$ we keep this construction and then we start an interpolation scheme with respect to the cubes $\q_k^z$, where we put the value~$u_{k,\e}^z$ in most part of a cube. The boundary conditions $b_{k,\e}[z]$ help to control interactions between different cubes. In the estimates we can allow for multiplicative constants since the total contribution will be proportional to $2^{m(\lambda)}\lambda\sim\eta$. However, a precise dependence on the energy with respect to the layer number $k$ is crucial since we have to sum over all layers.

Now let us start with the details. For the moment fix $0\leq k<k_{\e}$. Given a cube $\q_k^z\in\mathcal{Q}_k$, let $P_{k,z} \colon \q_k^z\to\partial \q_k^z$ be any function such that $|P_{k,z}(x)-x|=\dist(x,\partial \q_k^z)$ for all $x\in \q_{k}^z$. Set $ \bar u_{\e}  \colon \e\ZZ^2\cap \q_k^z\to\SS^1$ as 
\begin{equation*}
\bar u_{\e} (\e i)=
\Geo\left[b_{\e}(P_{k,z}(\e i)),u_{k,\e}^z\right]\left(\theta_{\e}\e^{-1}\dist\big(\e i,\partial \q_k^z\big)\right),
\end{equation*}
with the extended geodesics given by Definition \ref{d.geodesic} and $b_{\e}$ given by \eqref{eq:bglobally}. Since in general $\bar{u}_{\e}(\e i)\notin\S_\e$, we project it. The function $u_\e$ in the square~$Q(\lambda)$ is then given by
\begin{equation}\label{eq:globalu_e}
u_{\e}(\e i):=
\begin{cases}
v_{\e}(\e i) &\mbox{if $\e i\in Q_{k_{\e}-1}$} \, ,
\\ 
\mathfrak{P}_{\e}(\bar{u}_{\e}(\e i)) &\mbox{if $\e i\in \q_k^z$ for some $\q_k^z\in\mathcal{Q}_k$ with $0\leq k<k_{\e}$} \, ,
\end{cases}
\end{equation}  
with the operator $\mathfrak{P}_{\e}$ defined in \eqref{eq:defproj}. In this step we are interested in the energy restricted to $Q(\lambda)$ and for this reason we defined $u_\e$ only in $Q(\lambda)$. The sequence~$u_\e$ will be defined later in Step~4 outside $Q(\lambda)$, that means, far from the singularity, as in Proposition~\ref{prop:construction of ueps}.

First let us identify the $L^1(Q(\lambda))$-limit of $u_{\e}$. To this end, observe that for all $\e i\in \q_k^z$ with $0\leq k<k_{\e}$ we have by Definition \ref{d.geodesic}
\begin{equation}\label{eq:interiorrigidity}
u_{\e}(\e i)=\mathfrak{P}_{\e}(u_{k,\e}^z) \quad\quad\text{if }\quad\theta_{\e}\e^{-1}\dist(\e i,\partial \q_k^z)\geq \geo(b_{\e}(P_{k,z}(\e i)),u_{k,\e}^z) \,.
\end{equation}
We need to quantify the dependence on $k$ in the right hand side. Let $S\subset\partial \q_k^z$ be a side such that $P_{k,z}(\e i)\in S$. Since the boundary datum $b_{\e}$ restricted to $S$ interpolates via geodesic arcs between the three elements of the vector $w_S^{\e}$ defined in \eqref{eq:defv_S} and by construction $u_{k,\e}^z\in\{(\bar w_{\e})^{-}_S,(\bar w_{\e})^{+}_S\}$, it follows from \eqref{eq:midpointest} and \eqref{eq:v_Sestimate} (with $k+1$ in place of $k$ if $S \subset \partial Q_k$, which improves the estimate) that
\begin{align}\label{eq:closenessbc}
\geo(b_{\e}(P_{k,z}(\e i)),u_{k,\e}^z)&\leq \geo\left(b_{\e}(P_{k,z}(\e i)),{\rm mid}((\bar w_{\e})^{+}_S,(\bar w_{\e})^{-}_S)\right)+C_1  2^{k_{\e}-k-m(\lambda)} \theta_{\e}\nonumber
\\
&\leq \max_{i=1,3}\geo\left((w_S^{\e})^i,{\rm mid}((\bar w_{\e})^{+}_S,(w_{\e})^{-}_S)\right)+C_1 2^{k_{\e}-k-m(\lambda)} \theta_{\e}\nonumber
\\
&\leq C_2   2^{k_{\e}-k-m(\lambda)} \theta_{\e} \, ,
\end{align}
for $C_2 = \big(\tfrac{\pi}{2} + 1\big)C_1 \leq 96$. In particular, the condition \eqref{eq:interiorrigidity} implies that
\begin{equation}\label{eq:layerthickness}
u_{\e}(\e i)=\mathfrak{P}_{\e}(u_{k,\e}^z) \quad\quad\text{if }\quad\dist(\e i,\partial \q_k^z)\geq C_2 \,2^{k_{\e}-k}\e \, .
\end{equation}
Since $2^{k_{\e}}\theta_{\e}\leq 2$ (cf. \eqref{eq:choicek_e}), the term $2^{k_{\e}}\e$ vanishes when $\e\to 0$. As the measure of each~$\q_k^z$ is $2^{-k}\lambda$, we deduce from \eqref{almostbestapproximation}  that a.e.\ in $Q(\lambda)$ (and thus in $L^1(Q(\lambda))$) it holds that
\begin{equation}\label{eq:deflimit}
u_{\e}\to u^{\lambda}_0=\begin{cases} \displaystyle
 \tfrac{x}{|x|} &\mbox{on $Q_{\infty}:= [(-2^{m(\lambda)}+2)\lambda,(2^{m(\lambda)}-2)\lambda]^2$,}
\\
\\
\displaystyle \tfrac{2^{-k-1}\lambda(2z+e_1+e_2)}{|2^{-k-1}\lambda(2z+e_1+e_2)|} &\mbox{if $x\in \q_k^z\cap L_k$ for some $k\in\mathbb{N}\cup\{0\}$} \, .
\end{cases}
\end{equation}
Notice that $Q_\infty = \bigcap_{k = 0}^\infty Q_k$ and that $u^\lambda_0 = \tfrac{x}{|x|}$  except in the layer $Q(\lambda) \sm Q_\infty$, whose thickness is~$2 \lambda$, thus infinitesimal when $\lambda \to 0$. 

Below we bound the differences $u_{\e}(\e i)-u_{\e}(\e j)$ for all $\e i,\e j\in\e\mathbb{Z}^2\cap Q(\lambda)$ with $|i-j|=1$.

\medskip
\ul{Substep 2.1} (interactions within a single cube)

\noindent Consider first $\e i,\e j\in \e\ZZ^2\cap \q_k^z$ with $0\leq k<k_{\e}$ and $|i-j|=1$. We treat several cases:

\medskip
\noindent {\it Case 1}: If $\dist(\e i,\partial \q_k^z)\geq C_2 \, 2^{k_{\e}-k}\e$ and $\dist(\e j,\partial \q_k^z)\geq C_2 \,2^{k_{\e}-k}\e$, then by \eqref{eq:layerthickness}
\begin{equation*}
|u_{\e}(\e i)-u_{\e}(\e j)|=|\mathfrak{P}_\e(u_{k,\e}^z)-\mathfrak{P}_\e(u_{k,\e}^z)|=0 \, .
\end{equation*}
By the Lipschitz continuity of $\dist(\cdot, \de \q_k^z)$, we can from now on assume that 
\begin{equation}\label{eq:k_standing assump}
\max\{\dist(\e i,\partial \q_k^z),\dist(\e j,\partial \q_k^z)\}< (C_2+1)2^{k_{\e}-k}\e \, .
\end{equation}

\noindent {\it Case 2}: We first analyze when $P_{k,z}(\e i)$ and $P_{k,z}(\e j)$ lie on different $1$-dimensional boundary segments $S_i\neq S_j$ of $\q_{k}^z$. We claim that $P_{k,z}(\e i)$ and $P_{k,z}(\e j)$ are then close to a node of the lattice  $2^{-k}\lambda\mathbb{Z}^2$. Indeed, denote by $\Pi_{S_i}$ and $\Pi_{S_j}$ the projections onto the subspaces spanned by the segments $S_i$ and $S_j$, respectively. Assumption \eqref{eq:k_standing assump} and the defining property of~$P_{k,z}$ imply that
\begin{equation*}
|P_{k,z}(\e i)-P_{k,z}(\e j)|\leq \e|i-j|+\dist(\e i,\partial \q_k^z)+\dist(\e j,\partial \q_k^z)\leq 2(C_2+1)2^{k_{\e}-k}\e+\e\,.
\end{equation*} 
Hence for $\e$ small enough the sides $S_i$ and $S_j$ cannot be parallel because $2^{k_{\e}}\e\ll\lambda$. Therefore the point $\Pi_{S_i}(\Pi_{S_j}(\e i))$ belongs to $S_i\cap S_j\subset 2^{-k}\lambda\ZZ^2$. We claim that
\begin{equation}\label{eq:batcorner}
b_{\e}(P_{k,z}(\e i)) = b_{\e}(P_{k,z}(\e j)) = \bar{w}_{\e}(\Pi_{S_i}(\Pi_{S_j}(\e i)))\,.
\end{equation}
Indeed, denote by $0\leq k^*_i,k^*_j\leq k_{\e}$ the layer numbers and by $S_i^*\subset S_i$ and $S_j^*\subset S_j$ the sides satisfying
\begin{equation*}
b_{\e}(P_{k,z}(\e i))=b^{\e}_{2^{k^*_i}S_i^*}[w_{S^*_i}^{\e}](2^{k^*_i}P_{k,z}(\e i)) \, ,\quad\quad b_{\e}(P_{k,z}(\e j))=b^{\e}_{2^{k^*_j}S_j^*}[w_{S^*_j}^{\e}](2^{k^*_j}P_{k,z}(\e j)) \, .
\end{equation*}
(The sides $S_i^*$ and $S_j^*$ are needed due to the fact that~$S_i$ or~$S_j$ may be contained in~$\de Q_k$, where~$b_{\e}$ is defined using the cubes which decompose the layer~$L_{k+1}$; if, for instance, $S_i$ is not contained in~$\de Q_k$, then~$k_i^* = k$ and $S_i^* = S_i$.) Since by the dyadic construction $S^*_i$ either agrees with $S_i$ or is exactly one half of the side~$S_i$, it follows that $\Pi_{S_i}(\Pi_{S_j}(\e i))$ is an endpoint of $S^*_i$. By the same reasoning it is also an endpoint of $S^*_j$.
Since $\Pi_{S_i}(\Pi_{S_j}(\e i))=\Pi_{S_j}(\Pi_{S_i}(\e j))$, it then suffices to show that $2^{k_i^*}P_{k,z}(\e i)$ and $2^{k_j^*}P_{k,z}(\e j)$ are sufficiently close to $2^{k_i^*}\Pi_{S_i}(\Pi_{S_j}(\e i))$ and $2^{k_j^*}\Pi_{S_j}(\Pi_{S_i}(\e j))$, respectively, since by construction the boundary datum is constant in a neighborhood of the endpoints of a side. The $1$-Lipschitz continuity of $\Pi_{S_i}$ and $\Pi_{S_j}$ combined with \eqref{eq:k_standing assump} yields
\begin{align}\label{eq:k_ilocate0}
|P_{k,z}(\e i)-\Pi_{S_i}(\Pi_{S_j}(\e i))|&\leq |\e i-\Pi_{S_j}(\e i)|\nonumber
\leq |\e i- \e j|+|\e j-\Pi_{S_j}(\e j)|+|\Pi_{S_j}(\e i)-\Pi_{S_j}(\e j)|   
\\
&\leq 2\e+(C_2+1)2^{k_{\e}-k}\e \, .
\end{align}
Similarly we can derive the estimate
\begin{equation}\label{eq:k_jlocate0}
|P_{k,z}(\e j)-\Pi_{S_j}(\Pi_{S_i}(\e j))|\leq 2\e+(C_2+1)2^{k_{\e}-k}\e \, .
\end{equation}
By \eqref{eq:choicek_e}, for $\e$ small enough both terms can be bounded by $2^{-k+1}(C_2+2)\e\theta_{\e}^{-1}$. Since $k\leq k^*_i,k^*_j\leq k+1$, multiplying \eqref{eq:k_ilocate0} by $2^{k_i^*}$ and \eqref{eq:k_jlocate0} by $2^{k_j^*}$ yields
\begin{equation}\label{eq:b will be constant}
\max\left\{|2^{k_i^*}P_{k,z}(\e i)-2^{k_i^*}\Pi_{S_i}(\Pi_{S_j}(\e i))|,|2^{k_j^*}P_{k,z}(\e j)-2^{k_j^*}\Pi_{S_j}(\Pi_{S_i}(\e j))|\right\}\leq C_3 \e\theta_{\e}^{-1},
\end{equation}
where $C_3 = 4(C_2+2) = 392 < c_0$, $c_0$ being the constant in the definition~\eqref{eq:defbcases} (thus explaining the choice $c_0 = 393$). The estimate \eqref{eq:b will be constant} thus implies \eqref{eq:batcorner}. 

Having in mind that $\frac{\theta_{\e}}{\e}|\dist(\e i,\partial I)-\dist(\e j,\partial I)|\leq \theta_{\e}$,
the $1$-Lipschitz continuity of $\Geo[\bar{w}_{\e}(\Pi_{S_i}(\Pi_{S_j}(\e i))),u_{k,\e}^z]$ and the formula for $\bar{u}_{\e}$ then yield that $|\bar{u}_{\e}(\e i)-\bar{u}_{\e}(\e j)|\leq \theta_{\e}$. From the definition of the function $\mathfrak{P}_{\e}$ and the previous estimate we infer that
\begin{equation}\label{eq:k_verygoodbound}
|u_{\e}(\e i)-u_{\e}(\e j)|\leq \theta_{\e} \, .
\end{equation}

\noindent {\it Case 3}: It remains to treat the case of points $i$ and $j$ such that $P_{k,z}(\e i)=\Pi_{S_i}(\e i)$ and $P_{k,z}(\e j)=\Pi_{S_i}(\e j)$. Here we use the Lipschitz-continuity of $b_{\e}$ on $S_i$. Note that $b_{\e}$ might be defined separately on two smaller sides contained in $S_i$, but nevertheless the Lipschitz property \eqref{eq:LiponS} holds on the whole $S_i$ due to convexity. Moreover, we want to apply the stability estimate of Lemma \ref{l.ongeodesics}. To this end, observe that  by~\eqref{eq:closenessbc} and~\eqref{eq:choicek_e} we have
\begin{equation*}
|b_{\e}(P_{k,z}(\e i))-u_{k,\e}^z|\leq C_2 2^{k_{\e}-k-m(\lambda)} \theta_{\e}\leq 2 C_2  2^{-m(\lambda)} 
\end{equation*}
and the right hand side can be made arbitrarily small (specifically, $2 C_2  2^{-m(\lambda)}  < c$, where~$c$ is the constant given in Lemma~\ref{l.ongeodesics}) since $m(\lambda)\gg 1$ for small $\lambda$. The same estimate holds with $i$ replaced by $j$. To reduce notation, we set $d_{\e,i}=\theta_{\e}\e^{-1}\dist(\e i,\partial \q_k^z)$ and $d_{\e,j}=\theta_{\e}\e^{-1}\dist(\e j,\partial \q_k^z)$. Then by the triangle inequality,~\eqref{eq:LiponS}, and Lemma \ref{l.ongeodesics} we have
\begin{align*}
|\bar{u}_{\e}(\e i)-\bar{u}_{\e}(\e j)|& \leq \big|\Geo[b_{\e}(P_{k,z}(\e i)),u_{k,\e}^z](d_{\e,i})-\Geo[b_{\e}(P_{k,z}(\e i)),u_{k,\e}^z](d_{\e,j})\big|
\\
& \quad +\big|\Geo[b_{\e}(P_{k,z}(\e i)),u_{k,\e}^z](d_{\e,j})-\Geo[b_{\e}(P_{k,z}(\e j)),u_{k,\e}^z](d_{\e,j})\big|
\\
& \leq |d_{\e,i}-d_{\e,j}|+\geo\left(b_{\e}(P_{k,z}(\e i)),b_{\e}(P_{k,z}(\e j))\right)
\\
& \leq \theta_{\e}+\tfrac{\pi}{2} \,\theta_{\e}\e^{-1}|\Pi_{S_i}(\e i)-\Pi_{S_i}(\e j)|\leq (1+\tfrac{\pi}{2}) \theta_{\e} \, .
\end{align*}
Hence we deduce the weaker but still sufficient bound 
\begin{equation}\label{eq:k_worstbound}
|u_{\e}(\e i)-u_{\e}(\e j)|\leq 3 \theta_{\e} \, .
\end{equation}

\ul{Substep 2.2} (interactions between different cubes)

\noindent Now we consider lattice points $\e i\in \q_{k_i}^{z_i}$ and $\e j\in \q_{k_j}^{z_j}$ with $\q_{k_i}^{z_i}\neq \q_{k_j}^{z_j}$ and $|i-j|=1$. In this substep we assume that $0\leq k_i,k_j \leq k_{\e} -1$, that means, we consider only the layers where we interpolate. Assume without loss of generality that $k_i\leq k_j$. We also have to consider the numbers $k^*_i$ and $k^*_j$ characterized by the property
\begin{equation*}
P_{k_i,z_i}(\e i)\in L_{k^*_i}\setminus \partial Q_{k_i^*},\quad\quad P_{k_j,z_j}(\e j)\in L_{k^*_j}\setminus \partial Q_{k_j^*}\,,
\end{equation*}
that means, those values which determine the rescaling of the boundary conditions. Note that from the definition of $P_{k,z}$ it follows that
\begin{equation}\label{eq:relstar}
k_i\leq k_i^*\leq k_i+1,\quad\quad k_j\leq k_j^*\leq k_j+1\,.
\end{equation}
Since all cubes $\q_k^z$ are half-open and oriented along the coordinate axes, there exists a side $S_{ij}$ of $\partial \q_{k_j}^{z_j}$ such that the segment $[\e i, \e j]$ intersects $S_{ij}$ orthogonally and additionally 
\begin{equation} \label{eq:nontrivial inclusion of side}
	S_{ij}\subset \partial \q_{k_i}^{z_i} \cap \partial \q_{k_j}^{z_j} \,,
\end{equation} 
where we used that $k_j \geq k_i$ to ensure the inclusion. In particular,
\begin{equation} \label{eq:dist smaller than eps}
	\dist(\e i,\partial \q_{k_i}^{z_i})+\dist(\e j,\partial \q_{k_j}^{z_j})\leq \e
\end{equation}
which implies that
\begin{equation}\label{eq:projectionsclose}
|P_{k_i,z_i}(\e i)-P_{k_j,z_j}(\e j)|\leq |P_{k_i,z_i}(\e i)-\e i|+|\e i-\e j|+|\e j-P_{k_j,z_j}(\e j)|\leq 2\e\,.
\end{equation} 
Moreover, in analogy to the estimate \eqref{eq:generalbounds} we deduce the bound
\begin{equation}\label{eq:k_generalbounds}
| \bar{u}_{\e}(\e i)  -b_{\e}(P_{k_i,z_i}(\e i))|+| \bar{u}_{\e}(\e j)  -b_{\e}(P_{k_j,z_j}(\e j))|\leq\theta_{\e} \,.
\end{equation}
Note that the above estimate does not give information on $|\bar{u}_{\e}(\e i) - \bar{u}_{\e}(\e j)|$ since, {\em a priori}, $b_{\e}(P_{k_i,z_i}(\e i))$ might differ  from $b_{\e}(P_{k_j,z_j}(\e j))$. We will show that this is not the case. To this end, we shall prove two alternatives: 
\begin{equation}\label{eq:alternatives}
\begin{split}
&{\rm (i)} \quad P_{k_i,z_i}(\e i)\in S_{ij}\text{ and }P_{k_j,z_j}(\e j)\in S_{ij}\,;\\
&{\rm (ii)}\quad \dist(2^{k_i^*}P_{k_i,z_i}(\e i),\lambda\ZZ^2)\leq 2^{k_i+3}\e\text{ and }\dist(2^{k_j^*}P_{k_j,z_j}(\e j),\lambda\ZZ^2)\leq 2^{k_j+2}\e\,.
\end{split} 
\end{equation}
Indeed, first assume that $P_{k_j,z_j}(\e j)\notin S_{ij}$. Then there exists another facet $S_j$ of~$\q_{k_j}^{z_j}$, $S_j \neq S_{ij}$, such that $P_{k_j,z_j}(\e j)\in S_j$. Since $\dist(\e j,S_{ij})\leq \e$ and $\dist(\e j,S_j)\leq\e$, the sides $S_{ij}$ and $S_j$ cannot be parallel since the distance between parallel sides of $\partial \q_{k_j}^{z_j}$ is given by $2^{-k_j}\lambda\geq\tfrac{1}{2}\theta_{\e}\lambda\gg\e$. Hence $\Pi_{S_j}(\Pi_{S_{ij}}(\e j))\in S_j\cap S_{ij}\subset2^{-k_j}\lambda\ZZ^2$, so that
\begin{equation}\label{eq:auxest}
\dist(P_{k_j,z_j}(\e j),2^{-k_j}\lambda\ZZ^2)=\dist(\Pi_{S_j}(\e j),2^{-k_j}\lambda\ZZ^2)\leq |\e j - \Pi_{S_{ij}}(\e j)|\leq \e \, .
\end{equation}
In particular, applying \eqref{eq:relstar} we deduce from the above estimate that
\begin{equation}\label{eq:locatePj}
\begin{split}
	\dist(2^{k_j^*}P_{k_j,z_j}(\e j),\lambda\ZZ^2)
	&  =  2^{k_j^*} \dist(P_{k_j,z_j}(\e j),2^{-k_j^*}\lambda\ZZ^2) \\
	& \leq 2^{k_j + 1}\dist(P_{k_j,z_j}(\e j),2^{-k_j}\lambda\ZZ^2)
	\leq 2^{k_j + 1}\e\,,
\end{split}
\end{equation}
where we used that $2^{-k_j} \lambda \ZZ^2 \subset 2^{-k_j^*} \lambda \ZZ^2$. For the point $P_{k_i,z_i}(\e i)$ consider first the case $P_{k_i,z_i}(\e i)\in S_{ij}$. Due to what we aim to prove we then assume that $P_{k_j,z_j}(\e j)\in S_j\setminus S_{ij}$ as above, so that \eqref{eq:locatePj} holds. Hence \eqref{eq:projectionsclose} and \eqref{eq:auxest} imply
\begin{equation*}
\dist(P_{k_i,z_i}(\e i),2^{-k_j}\lambda \ZZ^2)\leq\dist(P_{k_j,z_j}(\e j),2^{-k_j}\lambda\ZZ^2)+2\e\leq 3\e  < 2^2 \e \, .
\end{equation*}
In order to conclude the claimed estimate, observe that the condition $P_{k_i,z_i}(\e i)\in S_{ij}\subset \partial \q_{k_j}^{z_j} \subset L_{k_j}$ forces $k_i^*\geq k_j$ so that
\begin{align}\label{eq:locatePi}
\dist(2^{k_i^*}P_{k_i,z_i}(\e i),\lambda\ZZ^2)& = 2^{k^*_i}\dist(P_{k_i,z_i}(\e i),2^{-k_i^*}\lambda\ZZ^2)\nonumber
\\
&\leq 2^{k_i+1}\dist(P_{k_i,z_i}(\e i),2^{-k_j}\lambda\ZZ^2)\leq 2^{k_i+3}\e\,,
\end{align}
where we used that $2^{-k_j} \lambda \ZZ^2 \subset 2^{-k_i^*}\lambda \ZZ^2$. On the contrary, if $P_{k_i,z_i}(\e i)\notin S_{ij}$, denote by $S_i$ a facet of $\q_{k_i}^{z_i}$ such $P_{k_i,z_i}(\e i)\in S_i$, $S_i \neq S_{ij}$. Then $S_i$ and $S_{ij}$ do not lie on the same straight line. To show this, we argue by contradiction. Assume that $S_i\subset {\rm span}(S_{ij})$. Since the segment $[\e i,\e j]$ is orthogonal to $S_{ij}$, this would imply the false statement
\begin{equation*}
\Pi_{S_i}(\e i)=\Pi_{S_{ij}}(\e i)=\Pi_{S_{ij}}(\e j)\in S_{ij}\,,
\end{equation*}
where the last inclusion holds since $\e j\in \q_{k_j}^{z_j}$ and $S_{ij}$ is a side of the cube $\q_{k_j}^{z_j}$. Since~$S_i$ and~$S_{ij}$ can neither be parallel for $\e$ small enough, we conclude that   $\Pi_{S_i}(\Pi_{S_{ij}}(\e i))\in {\rm span}(S_i)\cap {\rm span}(S_{ij})$. Since, by \eqref{eq:nontrivial inclusion of side}, $S_{ij}\subset\partial \q_{k_i}^{z_i}$, we know that $\Pi_{S_i}(\Pi_{S_{ij}}(\e i))\in 2^{-k_i}\lambda\ZZ^2$. Thus the defining property of $S_{ij}$ yields
\begin{equation*}
\dist(P_{k_i,z_i}(\e i),2^{-k_i}\lambda\ZZ^2)=\dist(\Pi_{S_i}(\e i),2^{-k_i}\lambda\ZZ^2)\leq |\Pi_{S_{ij}}(\e i)-\e i|\leq |\e i-\e j|=\e
\, .
\end{equation*}
Again in combination with \eqref{eq:relstar} this inequality implies the estimate
\begin{equation}\label{eq:locatePi2}
\dist(2^{k_i^*}P_{k_i,z_i}(\e i),\lambda\ZZ^2)\leq 2^{k_i+1}\dist(P_{k_i,z_i}(\e i),2^{-k_i}\lambda\ZZ^2)\leq 2^{k_i+1}\e \, .
\end{equation}
It remains to establish an estimate for $\dist(2^{k_j^*}P_{k_j,z_j}(\e j),\lambda\ZZ^2)$ when $P_{k_i,z_i}(\e i)\notin S_{ij}$ and $P_{k_j,z_j}(\e j)\in S_{ij}$. In this case we have
\begin{align*}
\dist(P_{k_j,z_j}(\e j),2^{-k_j}\lambda\ZZ^2)& = \dist(\Pi_{S_{ij}}(\e j),2^{-k_i}\lambda\ZZ^2)
\leq |\Pi_{S_{ij}}(\e j)-\Pi_{S_{ij}}(\Pi_{S_i}(\e i))| \\
& \leq |\e j-\e i|+|\e i-\Pi_{S_i}(\e i)|\leq 2\e \, ,
\end{align*} 
where we used the inclusion $2^{-k_i} \lambda \ZZ^2 \subset 2^{-k_j} \lambda \ZZ^2$ (recall the assumption $k_i \leq k_j$ at the beginning of Substep~2.2). From the above inequality we deduce the estimate
\begin{equation}\label{eq:locatePj2}
\dist(2^{k_j^*}P_{k_j,z_j}(\e j),\lambda\ZZ^2)\leq 2^{k_j+1}\dist(P_{k_j,z_j}(\e j),2^{-k_j}\lambda\ZZ^2)\leq 2^{k_j+2}\e\,.
\end{equation}
Combining the estimates \eqref{eq:locatePj}, \eqref{eq:locatePi}, \eqref{eq:locatePi2}, and \eqref{eq:locatePj2} we proved the claimed alternatives (i) or (ii) in \eqref{eq:alternatives}. We analyze them separately below.

\medskip
\noindent {\it Case 5}: Assume $\dist(2^{k_i^*}P_{k_i,z_i}(\e i),\lambda\ZZ^2)\leq 2^{k_i+3}\e$ and $\dist(2^{k_j^*}P_{k_j,z_j}(\e j),\lambda\ZZ^2)\leq 2^{k_j+2}\e$ (that means, alternative~(ii)) and denote by $\lambda \bar{z}_i,\lambda \bar{z}_j\in \lambda\ZZ^2$ points realizing the minimal distance. We start by observing that $2^{-k_i^*}\lambda\bar{z}_i=2^{-k_j^*}\lambda\bar{z}_j$. Indeed, on the one hand we use \eqref{eq:projectionsclose} to estimate
\begin{equation*}
|2^{-k_i^*}\lambda\bar{z}_i-2^{-k_j^*}\lambda\bar{z}_j|\leq|2^{-k_i^*}\lambda\bar{z}_i-P_{k_i,z_i}(\e i)|+|P_{k_j,z_j}(\e j)-2^{-k_j^*}\lambda\bar{z}_j|+2\e\leq 14\e\,.
\end{equation*} 
On the other hand, since both points on the left hand side belong to $2^{-\max\{k_i^*,k_j^*\}}\lambda\ZZ^2\subset 2^{-k_j-1}\lambda\ZZ^2$ and $2^{-k_j-1}\lambda\geq \tfrac{1}{2}\theta_{\e}\lambda\gg\e$ by~\eqref{eq:choicek_e}, we deduce that $2^{-k_i^*}\lambda\bar{z}_i=2^{-k_j^*}\lambda\bar{z}_j$. We set $p_{ij} := 2^{-k_i^*}\lambda\bar{z}_i=2^{-k_j^*}\lambda\bar{z}_j$.

Let now $S_i$ and $S_j$ be the sides of the cubes in $\mathcal{Q}_{k_i^*}$ and $\mathcal{Q}_{k_j^*}$, respectively, such that $P_{k_i,z_i}(\e i) \in S_i$, $P_{k_j,z_j}(\e j) \in S_j$, and 
\begin{equation}\label{eq:sidechoice}
b_{\e}(P_{k_i,z_i}(\e i))=b^{\e}_{2^{k_i^*}S_i}[w^{\e}_{S_i}](2^{k_i^*}P_{k_i,z_i}(\e i)) \,,\quad  b_{\e}(P_{k_j,z_j}(\e i))=b^{\e}_{2^{k_j^*}S_j}[w^{\e}_{S_j}](2^{k_j^*}P_{k_j,z_j}(\e j))\,.
\end{equation} 
We claim that $p_{ij} \in S_i \cap S_j$. Indeed, since by assumption 
\begin{equation*}
|p_{ij}-P_{k_i,z_i}(\e i)|\leq 8\e\,,\quad\quad
|p_{ij}-P_{k_j,z_j}(\e j)|\leq 4 \e
\end{equation*}
and for a given side $S\subset\partial \q_k^z$ with $0\leq k<k_{\e}$ and $z\in\ZZ^2$ it holds that $\dist(S,2^{-k}\lambda\ZZ^2\setminus S)\geq 2^{-k}\lambda\gg \e$, the claim follows by a triangle inequality argument. Moreover, recalling that $k_i,k_j \leq k_\e {-}1$, property \eqref{eq:choicek_e} yields $2^{k_i+3}\e\leq 4\e\theta_{\e}^{-1} < c_0 \e\theta_{\e}^{-1}$ and $2^{k_j+2}\e\leq 2\e\theta_{\e}^{-1} < c_0 \e\theta_{\e}^{-1}$, $c_0$ being the constant used in the definition~\eqref{eq:defbcases}. Thus we conclude from \eqref{eq:sidechoice} and the definition of the boundary condition that
\begin{equation*}
	b_{\e}(P_{k_i,z_i}(\e i))=b_{\e}(p_{ij}) = b_{\e}(P_{k_j,z_j}(\e j))\,,
\end{equation*}   
where we used that $p_{ij}$ must be an endpoint of $S_i$ and $S_j$. 
%
Combined with \eqref{eq:k_generalbounds} we infer
\begin{equation*}
|\bar{u}_{\e}(\e i)-\bar{u}_{\e}(\e j)|\leq |\bar{u}_{\e}(\e i)-b_{\e}(P_{k_i,z_i}(\e i)) |+|\bar{u}_{\e}(\e j)-b_{\e}(P_{k_j,z_j}(\e j))|\leq \theta_{\e} \, , 
\end{equation*}
which by the definition of $\mathfrak{P}_{\e}$ allows to conclude that
\begin{equation}\label{eq:k_verygoodbound3}
|u_{\e}(\e i)-u_{\e}(\e j)|\leq\theta_{\e} \, . 
\end{equation}

\noindent{\it Case 6}:  Finally we analyze the case $P_{k_i,z_i}(\e i)\in S_{ij}$ and $P_{k_j,z_i}(\e j)\in S_{ij}$ (that means, alternative~(i)). Since by assumption the line segment $[\e i,\e j]$ intersects $S_{ij}$ orthogonally and $S_{ij}$ is a side of both cubes $\q_{k_i}^{z_i}$and $\q_{k_j}^{z_j}$, we know that $P_{k_i,z_i}(\e i)=P_{k_j,z_j}(\e j)$. In combination with
estimate \eqref{eq:k_generalbounds} we therefore obtain
\begin{equation*}
|\bar{u}_{\e}(\e i)-\bar{u}_{\e}(\e j)|   = |\bar{u}_{\e}(\e i)- b_{\e}(P_{k_i,z_i}(\e i))| + |b_{\e}(P_{k_j,z_j}(\e j)) -\bar{u}_{\e}(\e j) |\leq \theta_{\e}   \, ,
\end{equation*} 
which yields the estimate
\begin{equation}\label{eq:k_optbound}
|u_{\e}(\e i)-u_{\e}(\e j)|\leq  \theta_{\e} \, . 
\end{equation}

\medskip
\ul{Substep 2.3} (interactions between $Q_{k_{\e}-1}$ and the layers)

\noindent In this step we consider the case where $\e i\in\e\ZZ^2\cap Q_{k_{\e}-1}$ but $\e j\in\e\ZZ^2\setminus Q_{k_{\e}-1}$. Since $|\e i-\e j|=\e$, it follows that $\e i\in L_{k_{\e}}$ and $\e j\in L_{k_{\e}-1}$, that means, the last and the last but one layers. Indeed, the thickness of the last layer is $2^{-k_{\e}}\lambda\geq \tfrac{1}{2}\theta_{\e} \lambda\gg \e$, so that the claim follows by a triangle inequality argument. Let $z_j\in\ZZ^2$ be such that $\e j\in \q_{k_{\e}-1}^{z_j}$. From the definition of $u_{\e}$ in~\eqref{eq:globalu_e} and \eqref{almostbestapproximation} we infer 
\begin{align*}
|u_{\e}(\e i)-u_{\e}(\e j)|&=|v_{\e}(\e i)-\mathfrak{P}_{\e}(\bar{u}_{\e}(\e j))|\leq |v_{\e}(\e i)-\bar{u}_{\e}(\e j)|+\theta_{\e}
\\
& \leq |v_{\e}(\e i)-b_{\e}(P_{z_j,k_{\e}-1}(\e j))|+ |b_{\e}(P_{z_j,k_{\e}-1}(\e j)) -\bar{u}_{\e}(\e j)| + \theta_{\e} \\
&\leq |v_{\e}(\e i)-b_{\e}(P_{z_j,k_{\e}-1}(\e j))|+2\theta_{\e}\,,
\end{align*}
where the last inequality can be proven as the estimates \eqref{eq:generalbounds} and \eqref{eq:k_generalbounds}. Using the general estimate \eqref{eq:closenessbc} with $k=k_{\e}-1$ we can further bound the last term to conclude that
\begin{equation}\label{eq:nointerpolation1}
|u_{\e}(\e i)-u_{\e}(\e j)|\leq |v_{\e}(\e i)-u_{k_{\e}-1,\e}^{z_j}|+C\,\theta_{\e}=|v_{\e}(\e i)-u(2^{ -k_{\e} + 1}\lambda(z_j+\tfrac{1}{2} e_1+ \tfrac{1}{2} e_2))|+C\,\theta_{\e} \, .
\end{equation}
Recall that $2^{ -k_{\e} + 1 }\lambda( z_j+\tfrac{1}{2} e_1+ \tfrac{1}{2}e_2)$ is the midpoint of the cube $\q_{k_{\e}-1}^{z_j}$, so that by \eqref{eq:choicek_e}
\begin{align*}
|\e i - 2^{ -k_{\e} + 1}\lambda( z_j+ \tfrac{1}{2} e_1+ \tfrac{1}{2} e_2)|&\leq |\e i-\e j|+|\e j-2^{ -k_{\e} + 1}\lambda(z_j+\tfrac{1}{2} e_1+ \tfrac{1}{2} e_2)|
\\
&\leq \e+2^{-k_{\e}+1}\lambda\leq\e+2\theta_{\e}\lambda\,.
\end{align*}
We insert this bound with the estimates \eqref{almostbestapproximation} and \eqref{eq:xmodxcontinuous} in \eqref{eq:nointerpolation1} to obtain
\begin{align*}
|u_{\e}(\e i)-u_{\e}(\e j)|&\leq |u(\e i)-u( 2^{ -k_{\e} + 1}\lambda( z_j+\tfrac{1}{2}e_1+ \tfrac{1}{2}e_2))|+C\,\theta_{\e}
\\
&\leq 2\frac{|\e i - 2^{ -k_{\e} + 1}\lambda(z_j+ \tfrac{1}{2} e_1+ \tfrac{1}{2} e_2)|}{|2^{ -k_{\e} + 1}\lambda(z_j+\tfrac{1}{2}e_1+\tfrac{1}{2}e_2)|}+C\,\theta_{\e}\leq C\frac{\e+2\theta_{\e}\lambda}{(2^{m(\lambda)}- 2)\lambda}+C\,\theta_{\e},
\end{align*}
where for the last inequality we used the set inclusion \eqref{eq:minimalhole}. Since $\e\ll\theta_{\e}$, for $\lambda>0$ fixed we can assume that $\e\leq\theta_{\e}\lambda$, so that the last estimate turns into the bound
\begin{equation}\label{eq:nointerpolation2}
|u_{\e}(\e i)-u_{\e}(\e j)|\leq C\,\theta_{\e} \, .
\end{equation}

\medskip
\ul{Step 3} (energy estimates in $Q(\lambda)$)

\noindent Let us first summarize what we have proven so far. By our choice of $m(\lambda)$ at the beginning of Step 2 we have $Q(\lambda)\subset B_{\eta/2}$ and~\eqref{eq:minimalhole}. Hence we can use the bound \eqref{eq:step1final} of Step~1 to control the energy due to interactions with both points in $Q_{k_{\e}-1}$, where $u_\e = v_\e$, cf.~\eqref{eq:globalu_e}. For the interactions with at least one point in $Q(\lambda)\setminus Q_{k_{\e}-1}$, we showed in Substeps 2.1--2.3 (cf.\ \eqref{eq:k_verygoodbound}, \eqref{eq:k_worstbound}, \eqref{eq:k_verygoodbound3}, \eqref{eq:k_optbound}, and~\eqref{eq:nointerpolation2}) that the bound
\begin{equation}\label{eq:fineconstruction}
|u_{\e}(\e i)-u_{\e}(\e j)|\leq C\,\theta_{\e}
\end{equation}
holds with a uniform constant $C<+\infty$. In order to obtain precise estimates on the energy due to interactions with at least one point in $Q(\lambda)\setminus Q_{k_{\e}-1}$, we have to count the number of lattice points $\e i, \e j$ satisfying $u_{\e}(\e i)\neq u_{\e}(\e j)$. For such points~\eqref{eq:fineconstruction} will suffice.

Fix such $\e i, \e j$. Then there exists a cube $\q_k^z\in\mathcal{Q}_k$ with $0\leq k< k_{\e}$ and $z\in\ZZ^2$ with
\begin{equation} \label{eq:count these instead}
\dist(\e i,\partial \q_k^z)\leq C\, 2^{k_{\e}-k}\e\,.
\end{equation}
Indeed, if $\e i, \e j \in Q(\lambda)\setminus Q_{k_{\e}-1}$ and they belong to the same cube of $\mathcal{Q}_k$ (Substep~2.1), then this is a consequence of~\eqref{eq:k_standing assump}. If $\e i, \e j \in Q(\lambda)\setminus Q_{k_{\e}-1}$, but they belong to two different cubes (Substep~2.2), then this follows from~\eqref{eq:dist smaller than eps}. Finally, if, for instance, $\e i \in Q_{k_\e -1}$ and $\e j \in Q(\lambda)\setminus Q_{k_{\e}-1}$ (Substep~2.3), then this is a consequence of the fact that $\e i \in L_{k_\e}$ and $\e j \in L_{k_\e -1}$ (see also~\eqref{eq:dist smaller than eps}). Therefore it suffices to count lattice points that satisfy~\eqref{eq:count these instead}. 

From a covering argument with cubes of volume $\e^2$ and \eqref{eq:choicek_e} we infer that
\begin{align*}
\e^2\#\{\e i\in\e\ZZ^2:\,\dist(\e i,\partial \q_k^z)\leq C 2^{k_{\e}-k}\e\}&\leq 4(2^{-k}\lambda+2(C2^{k_{\e}-k}\e+\e))(2C2^{k_{\e}-k}\e+2\e)
\\
&\leq C\, (2^{-k}\lambda+\e(2^{-k}\theta_{\e}^{-1}+1))(2^{-k}\theta_{\e}^{-1}+1)\e
\\
&\leq C\,2^{-2k}\lambda\e\theta_{\e}^{-1},
\end{align*} 
where in the last inequality we also used that $\e\theta_{\e}^{-1}\ll\lambda$ for $\e$ small enough. Next recall that the number of cubes $\q_{k}^z$ in the layer $L_k$ can be roughly bounded by
\begin{equation*}
\#\mathcal{Q}_k\leq C\,2^{m(\lambda)}2^k.
\end{equation*}
Combining the previous two estimates with \eqref{eq:fineconstruction} we can estimate the energy of $u_{\e}$ via
\begin{align*}
\frac{1}{\e\theta_{\e}}E_{\e}(u_{\e};Q(\lambda))&\leq\frac{1}{\e\theta_{\e}}E_{\e}(v_{\e};B_{\eta})+\frac{C\theta_{\e}^2}{\e\theta_{\e}}2^{m(\lambda)}\lambda\sum_{k=0}^{k_{\e}} 2^{-k}\e\theta_{\e}^{-1}
\\
&\leq\frac{1}{\e\theta_{\e}}E_{\e}(v_{\e};B_{\eta})+ C 2^{m(\lambda)}\lambda\,.
\end{align*}
Due to the choice of $m(\lambda)$ it holds that $2^{m(\lambda)}\lambda\leq\eta$. Subtracting the term $2\pi|\log \e|\tfrac{\e}{\theta_{\e}}$ and inserting the upper bound \eqref{eq:step1final} we conclude that
\begin{equation}\label{eq:step3final}
\limsup_{\e \to 0}\left(\frac{1}{\e\theta_{\e}}E_{\e}(u_{\e};Q(\lambda))-2\pi|\log \e|\frac{\e}{\theta_{\e}}\right)\leq C\,\eta\,.
\end{equation}
We emphasize that $Q(\lambda)$ implicitly depends on $\eta$ through the quantity $2^{m(\lambda)}\lambda$.

\medskip
\ul{Step 4} (from local to global constructions)

\noindent We are now in a position to define $u_\e$ globally. In this step we stress again the dependence on~$n$ of~$\lambda_n$. We start by repeating the construction presented in Step~2 around each singularity $x_h$ of $u$, by defining $u_\e$ as in~\eqref{eq:globalu_e} (combined with a reflection if $\deg(u)(x_h)=-1$) in the squares $Q(\lambda_n,x_h) = Q(\lambda_n) + x_h$.

To define $u_\e$ outside the squares $Q(\lambda_n,x_h)$, we first observe that the square $x_h + [-2^{m(\lambda_n) + 1} \lambda_n, 2^{m(\lambda_n) + 1} \lambda_n]$ is not contained in $B_{\eta/2}(x_h)$, since $m(\lambda_n)$ has been chosen as the maximal integer such that $Q(\lambda_n,x_h) = x_h + [-2^{m(\lambda_n)} \lambda_n, 2^{m(\lambda_n)} \lambda_n] \subset B_{\eta/2}(x_h)$. This yields $\eta/4 \leq 2^{m(\lambda_n) + 1 } \lambda_n$ and thus, by~\eqref{eq:minimalhole}, 
\begin{equation} \label{eq:far enough from singularity}
	Q^{\lambda_n}_k(x_h) = x_h + Q_k  \supset B_{(2^{m(\lambda_n)}-2)\lambda}(x_h)\supset B_{\eta/16}(x_h) \, .
\end{equation}
Note that here we stress the dependence of $Q^{\lambda_n}_k(x_h)$ on $\lambda_n$, in contrast to the notation adopted for $Q_k$ in Step~2. We recall that $Q^{\lambda_n}_{-1}(x_h) = Q(\lambda_n,x_h)$.

Applying Lemma~\ref{lemma:discretization of smooth wout sing} with $O=\Omega \sm \bigcup_{h=1}^N\ol B_{\eta/16}(x_h)$ and $\tilde O = \tilde \Omega \sm \bigcup_{h=1}^N\ol B_{\eta/32}(x_h)$ to $u \in C^{\infty}(\tilde \Omega \sm \bigcup_{h=1}^N\ol B_{\eta/32}(x_h);\SS^1)$, we get a sequence of piecewise constant functions $u_n \in \PC_{\lambda_n}(\SS^1)$ such that $u_n \to u$ strongly in $L^1(\Omega \sm \bigcup_{h=1}^N\ol B_{\eta/16}(x_h);\RR^2)$ and, by~\eqref{eq:far enough from singularity},
\begin{equation} \label{eq:energy of pc far from singularity}
\limsup_{n \to +\infty}  \integral{J_{u_n} \cap (\Omega \sm \bigcup_{h=1}^N Q_0^{\lambda_n}(x_h))^{\lambda_n}}{\hspace{-1em}\geo(u_n^-,u_n^+) |\nu_{u_n}|_1}{ \d \H^1} \leq \integral{\Omega}{|\nabla u|_{2,1}}{\d x} \, .
\end{equation}
Notice that the squares $Q_0^{\lambda_n}(x_h)$ have vertices on the lattice $\lambda_n \ZZ^2$.

Let us fix $n$ large enough. For $\e i \in\e\ZZ^2\setminus\bigcup_{h=1}^N Q_0^{\lambda_n}(x_h)$ we define $u'_{\e}(\e i)$ as the recovery sequence given in the proof of Proposition~\ref{prop:construction of ueps} for the piecewise constant function~$u_n \in \PC_{\lambda_n}(\SS^1)$ with the constant $c_0=393$ in \eqref{eq:defbcases}. 
Then we define $u_\e(\e i) := u_\e'(\e i)$ for~$\e i \in\e\ZZ^2\setminus\bigcup_{h=1}^N Q(\lambda_n,x_h) = \e\ZZ^2\setminus\bigcup_{h=1}^N Q^{\lambda_n}_{-1}(x_h)$. This completes the definition of~$u_\e$ in~$\e \ZZ^2$. 

We claim that 
\begin{equation} \label{eq:equal on the last layer}
	\e i \in \e \ZZ^2 \cap Q^{\lambda_n}_{-1}(x_h) \text{ and }  \dist(\e i, \de Q^{\lambda_n}_{-1}(x_h)) \leq  \e \implies  u_\e(\e i) = u_\e'(\e i) \, ,
\end{equation} 
that means, the two constructions given by Step 2 and Proposition \ref{prop:construction of ueps} are identical. Indeed, first note that the assumptions on $\e i$ above imply that $\e i\in L_0$. Hence we find $\q_0^{z_0}\in\mathcal{Q}_{0}$ such that $\e i\in\q_0^{z_0}$. We now consider the two cases $P_{0,z_0}(\e i)\in\partial\q_0^{z_0} \setminus\partial Q^{\lambda_n}_{-1}(x_h)$ and $P_{0,z_0}(\e i)\in \partial Q^{\lambda_n}_{-1}(x_h)$. If $P_{0,z_0}(\e i)\in\partial\q_0^{z_0} \setminus\partial Q^{\lambda_n}_{-1}(x_h)$, let $S_i \subset \partial\q_0^{z_0}$ be the side such that~$P_{0,z_0}(\e i) \in S_i$. By the assumption in~\eqref{eq:equal on the last layer}, $S_i$ is not contained in $\de Q^{\lambda_n}_{0}(x_h)$ and thus it intersects a side $S_0$ of $\q_0^{z_0}$ such that $S_0 \subset \de Q^{\lambda_n}_{-1}(x_h)$. In particular, $\Pi_{S_i}(\Pi_{S_0}(\e i))\in\lambda_n \ZZ^2$ is an endpoint of $S_i$ and, by~\eqref{eq:equal on the last layer},
\begin{equation*} 
|P_{0,z_0}(\e i)-\Pi_{S_i}(\Pi_{S_0}(\e i))|=|\Pi_{S_i}(\e i)-\Pi_{S_i}(\Pi_{S_0}(\e i))|\leq |\e i-\Pi_{S_0}(\e i)|\leq \e\ll c_0\frac{\e}{\theta_{\e}}\,,
\end{equation*}
where we used that $S_0$ is the side such that $\dist(\e i, S_0) = \dist(\e i, \de Q^{\lambda_n}_{-1}(x_h)) \leq \e$. Since $P_{0,z_0}(\e i)$ is close enough to the corner $ p_{i,0} := \Pi_{S_i}(\Pi_{S_0}(\e i)) \in \lambda_n \ZZ^2$, the boundary condition used for the definition of $u_\e$ at $P_{0,z_0}(\e i)$ agrees with its value at the corner, cf.~\eqref{eq:defbcases}. Thus
\begin{equation*}
    \begin{split}
        u_\e(\e i) & = \Geo\Big[ b_\e( p_{i,0} ) , u^{z_0}_{0,\e} \Big]\big( \theta_\e \e^{-1} \dist(\e i, \de \mathfrak{q}^{z_0}_{0}) \big)\\
        & = \Geo\Big[u\big(p_{i,0} + \lambda_n ( \tfrac{1}{2} e_1 + \tfrac{1}{2} e_2 ) \big)  , u\big( \lambda_n ( z_0 +  \tfrac{1}{2} e_1 + \tfrac{1}{2} e_2 ) \big) \Big] \big(\theta_\e \e^{-1} \dist(\e i, \de \mathfrak{q}^{z_0}_{0}) \big) \, .
    \end{split}
\end{equation*}
The same holds true for $u'_\e$. This concludes the proof of~\eqref{eq:equal on the last layer} when $P_{0,z_0}(\e i) \in \de \mathfrak{q}^{z_0}_0 \sm \de Q^{\lambda_n}_{-1}(x_h)$.  If, instead, $P_{0,z_0}(\e i) \in \de Q^{\lambda_n}_{-1}(x_h)$, let $S_i$ be the side of $\mathfrak{q}^{z_0}_0$ such that $P_{0,z_0}(\e i) \in S_i$. Then the two 3-tuples of values $v_{S_i}(u)$ and $w^{\e}_{S_i}$ defined in \eqref{eq:defvSu} and \eqref{eq:defv_S}, respectively, coincide (note that, by definition, both cubes in $\mathcal{Q}_0$ and $\mathcal{Q}_{-1}$ have size $\lambda$). Then~\eqref{eq:equal on the last layer} follows in this case too.

Taking into account \eqref{eq:deflimit} on each $Q(\lambda_n,x_h)$, the function $u_{\e}\in\mathcal{PC}_{\e}(\S_\e)$ converges in $L^1(\Omega;\RR^2)$ to the function $\mathring u_n \in L^1(\Omega;S^1)$ defined by  
\begin{equation*}
\mathring u_n(x) :=
\begin{cases}
u_n(x) &\mbox{if $x\in \Omega\setminus\bigcup_{h=1}^N Q(\lambda_n,x_h)$} \, ,
\\
\left(\begin{smallmatrix} 1 & 0\\ 0 &\deg(u)(x_h)\end{smallmatrix}\right) u^{\lambda_n}_0(x-x_h) &\mbox{if $x\in Q(\lambda_n,x_h)$ for some $1\leq h\leq N$} \, ,  
\end{cases}
\end{equation*}
We remark that the precise structure for fixed $\lambda_n$ is not important. Just note that due to the fact that $Q(\lambda_n,x_h)\subset B_{\eta/2}(x_h)$ and~\eqref{eq:deflimit}, the layer in each $Q(\lambda_n,x_h)$ where~$\mathring u_n$ differs from $\big(\tfrac{x-x_h}{|x-x_h|}\big)^{\pm 1}$, and thus from $u$, is of thickness $2\lambda_n$. Consequently
\begin{equation}\label{eq:lambdalimit}
\mathring u_n\to u\quad\text{in }L^1(\Omega;\RR^2)\quad\text{as } n \to +\infty\,.
\end{equation}

It remains to estimate the energy of $u_{\e}$ in terms of $\lambda_n$ and $\eta$. In particular, we need to estimate the interactions between the square $Q(\lambda_n,x_h)$ and the exterior. 
Thanks to~\eqref{eq:equal on the last layer} we can split the energy as 
\begin{align*}
\frac{1}{\e\theta_{\e}}E_{\e}(u_{\e};\Omega)-2\pi N|\log \e|\frac{\e}{\theta_{\e}}& \leq \frac{1}{\e\theta_{\e}}E_{\e}\Big(u'_{\e};\Omega\setminus \bigcup_{h=1}^N Q^{\lambda_n}_0(x_h)\Big) \\
& \quad + \sum_{h=1}^N\left(\frac{1}{\e\theta_{\e}}E_{\e}(u_{\e};Q(\lambda_n,x_h))-2\pi|\log \e|\frac{\e}{\theta_{\e}}\right) \, .
\end{align*}
By \eqref{eq:step3final} and \eqref{eq:limiteps} we can pass to the limit in $\e$ and conclude that
\begin{align}\label{eq:almostrecovery}
\limsup_{\e \to 0}\left(\frac{1}{\e\theta_{\e}}E_{\e}(u_{\e};\Omega)-2\pi N|\log \e|\frac{\e}{\theta_{\e}}\right)& \leq \int_{J_{u_{n}}\cap(\Omega\setminus \bigcup_{h=1}^NQ^{\lambda_n}_0(x_h))^{\lambda_n}}{ \hspace{-1em}\geo(u_{n}^-,u_{n}^+)|\nu_{u_{n}}|_1}{\mathrm{d}\mathcal{H}^1}\nonumber
\\
&\quad +C\,\eta\,.
\end{align}
Before we can conclude, we have to identify the flat limit of the vorticity measures $\mu_{u_{\e}}$ associated with the sequence $u_{\e}$. Since the right hand side in \eqref{eq:almostrecovery} is finite, Proposition~\ref{prop:compactness M vortices} implies that (up to a subsequence) $\mu_{u_{\e}}\flat\bar{\mu}$ for some $\bar{\mu}=\sum_{k=1}^N d_k\delta_{y_k}$ with $d_k\in\ZZ$ and  $|\bar{\mu}|(\Omega)\leq N$ (we allow $d_k=0$ in order to sum from $1$ to $N$). We claim that $\bar{\mu}=\mu$ with $\mu$ defined in the statement of the proposition. Here comes the argument. Fix $x_0\in\Omega\setminus\{x_1,\dots,x_N\}$. Since the singular part $2\pi N|\log \e|\tfrac{\e}{\theta_{\e}}$ of the estimate \eqref{eq:almostrecovery} is concentrated in the set $\bigcup_{h=1}^NB_{2r_{\e}}(x_h)$ (cf.~\eqref{eq:exactconcentration}) and $r_{\e}\to 0$, we deduce that for $0<\rho\ll\lambda$ small enough $\limsup_{\e \to 0}\frac{1}{\e\theta_{\e}}E_{\e}(u_{\e};B_{\rho}(x_0))<+\infty$. Since we assume here that $\theta_{\e}\ll \e|\log \e|$, Remark \ref{rmk:vortices vanish} yields that $\mu_{u_{\e}}\mres B_{\rho}(x_0)\flat 0$. Testing this convergence with a Lipschitz-function $\varphi\in C_c^{0,1}(B_{\rho}(x_0))$ such that $\varphi(x_0)=1$ we obtain that $x_0\notin\{y_1,\dots,y_N\}$ (or $x_0=y_k$ for some $k$ with $d_k=0$). Since $x_0\in\Omega\setminus\{x_1,\dots,x_N\}$ was arbitrary, we can write $\bar{\mu}=\sum_{h=1}^N d_h\delta_{x_h}$. It remains to prove that $d_h=\deg(u)(x_h)$ for all $1\leq h\leq N$. Note that for $\rho\ll\lambda$ it holds that $u_{\e}=v_{\e}$ on each $B_{\rho}(x_h)$, where $v_{\e}$ is defined in \eqref{eq:defbestapproximation}. Due to \eqref{eq:step1final} we have for $\e$ small enough
\begin{equation}\label{eq:loggrowth}
\frac{1}{\e^2}E_{\e}(v_{\e};B_{\rho}(x_h))\leq C\eta\frac{\theta_{\e}}{\e}+2\pi|\log \e|\leq C|\log \e|\,.
\end{equation}
Hence we can apply \cite[Proposotion 5.2]{Ali-Cic-Pon}, which states that in dimension $2$ the flat convergence of $\mu_{v_{\e}}\mres B_{\rho}(x_h)$ is equivalent to the flat convergence of the (normalized) Jacobians of the piecewise affine interpolation of $v_{\e}$ on $B_{\rho}(x_h)$. Denote this piecewise affine interpolation and the one associated to the function $u$ on $B_{\rho}(x_h)$ by $\widehat{v}_{\e}$ and $\widehat{u}(\e)$, respectively. Inserting the estimate \eqref{almostbestapproximation} in the definition of the piecewise affine interpolation one can show that
\begin{equation}
|\widehat{v}_{\e}(x)-\widehat{u}(\e)(x)|\leq C\theta_{\e}\quad\text{for all }x\in B_{\rho}(x_h)\,.
\end{equation}
Taking into account one more time the estimate \eqref{eq:loggrowth}, we conclude that
\begin{align}\label{eq:diffpiecewiseaffine}
&\|\widehat{v}_{\e}-\widehat{u}(\e)\|_{L^2(B_{\rho/2}(x_h))}\left(\|\nabla\widehat{v}_{\e}\|_{L^2(B_{\rho/2}(x_h))}+\|\nabla\widehat{u}(\e)\|_{L^2(B_{\rho/2}(x_h))}\right)\nonumber
\\
\leq&\, C\theta_{\e}\left(\frac{1}{\e^2}E_{\e}(v_{\e};B_{\rho}(x_h))+\frac{1}{\e^2}E_{\e}(u;B_{\rho}(x_h))\right)^{\! \frac{1}{2}}\leq C \theta_{\e}|\log \e|^{\frac{1}{2}}\,,
\end{align}
where the bound $\frac{1}{\e^2}E_{\e}(u;B_{\rho}(x_h))\leq C|\log \e|$ can be proven with similar arguments used to show \eqref{eq:XYprecise}. The above right hand side vanishes when $\e\to 0$. Thus \cite[Lemma 3.1]{Ali-Cic-Pon} implies that the Jacobians fulfill $\mathrm{J}\widehat{v}_{\e}-\mathrm{J}\widehat{u}(\e)\flat 0$. Recalling that $u=\big(\tfrac{x-x_h}{|x-x_h|}\big)^{\pm 1}$ on $B_{\rho}(x_h)$, it follows from Step 1 of the proof of \cite[Theorem 5.1 (ii)]{Ali-Cic} that $\tfrac{1}{\pi}\mathrm{J}\widehat{u}(\e)\flat \deg(u)(x_h)\delta_{x_h}$. Fixing again $\varphi\in C_c^{0,1}(B_{\rho}(x_h))$ such that $\varphi(x_h)=1$, the above arguments imply 
\begin{equation*}
d_h=\langle \mu,\varphi\rangle=\lim_{\e \to 0}\langle \mu_{u_{\e}},\varphi\rangle = \lim_{\e \to 0}\langle \tfrac{1}{\pi}\mathrm{J}\widehat{v}_{\e},\varphi\rangle=\lim_{\e \to 0}\langle \tfrac{1}{\pi}\mathrm{J}\widehat{u}(\e),\varphi\rangle=\deg(u)(x_h)
\end{equation*}
as claimed.

Since the limit measure equals $\mu$ for all $\lambda$, we deduce from the $L^1(\Omega)$-lower semicontinuity of the $\Gamma\hbox{-}\limsup$, \eqref{eq:lambdalimit}, and  \eqref{eq:energy of pc far from singularity} that 
\begin{equation*}
\Gamma\hbox{-}\limsup_{\e \to 0}\left(\frac{1}{\e\theta_{\e}}E_{\e}-2\pi N|\log \e|\frac{\e}{\theta_{\e}}\right)(u,\mu)\leq C\,\eta+\int_{\Omega}|\nabla u|_{2,1}\,\mathrm{d}x\,.
\end{equation*}
The claim then follows by the arbitrariness of $0<\eta<\eta_0$ (recall that $|\mu|(\Omega)=N$).
\end{proof}

Together with Propositions \ref{prop:compactness M vortices} and \ref{prop:lb for M vortices} the next result finishes the proof of Theorem~\ref{thm:e smaller theta with vortices}.
\begin{proposition}[$M$ vortices, Upper bound]\label{prop:limsup theta<<eloge}
 Assume that $\e\ll\theta_{\e}\ll \e|\log\e|$. Let $\mu = \sum_{h=1}^N d_h \delta_{x_h}$ with $|\mu|(\Omega) = M \in \NN$ and let $u \in BV(\Omega;\SS^1)$. Then 
\begin{equation*}
\Gamma\text{-}\limsup_{\e \to 0} \Big( \frac{1}{\e \theta_\e} E_\e - 2 \pi M |\log \e|\frac{\e}{\theta_\e}\Big)(u,\mu)\leq \integral{\Omega}{|\nabla u|_{2,1} }{\d x} + |\DD^{(c)} u|_{2,1}(\Omega) + \mathcal{J}(\mu, u;\Omega) \, .
\end{equation*}
\end{proposition}
\begin{proof}
Fix $\sigma > 0$. By the definition~\eqref{eq:def of surface L} of $\mathcal{J}$ there exists a $T \in \D_2(\Omega \x \RR^2)$, with $T \in \cart(\Omega_\mu \x \SS^1)$, $\de T|_{\Omega \x \RR^2} = - \mu \x \llbracket \SS^1 \rrbracket$, and $u_T = u$ such that 
\begin{equation} \label{eq:from BV to cart}
\integral{\Omega \x \RR^2}{\Phi(\vec{T})}{\d |T|} \leq \integral{\Omega}{|\nabla u|_{2,1} }{\d x} + |\DD^{(c)} u|_{2,1}(\Omega) + \mathcal{J}(\mu, u;\Omega) + \sigma \, .
\end{equation}
In the previous inequality we applied Lemma~\ref{lemma:parametric in terms of u} with $\Omega_\mu$ in place of $\Omega$  (cf. Remark \ref{rmk:punctured domains}).

Due to Lemma~\ref{lemma:approximation with sing} we find an open set $\tilde \Omega \supset  \supset \Omega$ and a sequence of maps $u_k \in C^\infty(\tilde \Omega_\mu;\SS^1) \cap W^{1,1}( \tilde{\Omega};\SS^1)$ such that $u_k \to u$ in $L^1(\Omega;\RR^2)$, $|G_{u_k}|(\Omega \x \RR^2) \to |T|(\Omega \x \RR^2)$, and $\deg(u_k)(x_h) = d_h$ for $h = 1,\dots,N$. Reshetnyak's Continuity Theorem implies that 
\begin{equation} \label{eq:from cart to smooth}
\integral{\Omega}{|\nabla u_k|_{2,1}}{\d x} = \integral{\Omega \x \RR^2}{\Phi(\vec{G}_{u_k})}{\d |G_{u_k}|} \leq \integral{\Omega \x \RR^2}{\Phi(\vec{T})}{\d |T|} + \sigma \, ,
\end{equation}
for $k$ large enough. In the first equality we applied Lemma~\ref{lemma:parametric in terms of u} \& Remark \ref{rmk:punctured domains} to $u_k$ in $\Omega_\mu$.
Applying Lemmata~\ref{lemma:splitting degree}, \ref{lemma:correct points}, and  \ref{lemma:modifications near sing} we reduce to the assumptions in Proposition~\ref{p.smoothapprox}. By the lower semicontinuity of the $\Gamma$-$\limsup$ with respect to the strong $L^1$-convergence of~$u$ and the flat convergence of $\mu$, we conclude the proof.
\end{proof}
\noindent {\bf Acknowledgments.}  The work of M.\ Cicalese was supported by the DFG Collaborative Research Center TRR 109, “Discretization in Geometry and Dynamics”. G.\ Orlando has received funding from the Alexander von Humboldt Foundation and the European Union’s Horizon 2020 research and innovation programme under the Marie Sk\l odowska-Curie grant agreement No 792583. M.\ Ruf acknowledges financial support from the European Research Council under the European Community's Seventh Framework Program (FP7/2014-2019 Grant Agreement QUANTHOM 335410).

\end{document}